\documentclass[11pt, article]{IEEEtran}
\onecolumn
\usepackage{amsmath,amsfonts,amssymb,amscd,amsthm,dsfont}
\usepackage{graphicx,epsfig}
\usepackage{color}
\usepackage{mathtools}
\usepackage{bbm}
\usepackage{url}
\usepackage{balance}
\usepackage{epstopdf}
\usepackage{subfigure,enumitem}
\usepackage{cite,url,hyperref}
\usepackage{verbatim}
\usepackage{bm}
\usepackage{multicol}
\usepackage{booktabs}
\usepackage{multirow}
\usepackage{multicol}
\usepackage{algorithm}
\usepackage{algpseudocode}
\newcommand{\mE}{\mathbb{E}}

\newcommand{\nn}{\nonumber}

\newtheorem{theorem}{\textbf{Theorem}}

\newtheorem{lemma}{\textbf{Lemma}}

\newtheorem{definition}{\textbf{Definition}}
\newtheorem{remark}{\textbf{Remark}}

\usepackage{color}

\allowdisplaybreaks

\begin{document}
\title{Quickest Change Detection in Autoregressive Models
}
\author{Zhongchang Sun \quad Shaofeng Zou
	\thanks{Zhongchang Sun and Shaofeng Zou are with the Department of Electrical Engineering, University at Buffalo, Buffalo, NY 14228 USA (e-mail: \href{mailto:zhongcha@buffalo.edu}{zhongcha@buffalo.edu}, \href{mailto:szou3@buffalo.edu}{szou3@buffalo.edu}).}
}
\maketitle
\begin{abstract}
The problem of quickest change detection (QCD) in autoregressive (AR) models is investigated. A system is being monitored with sequentially observed samples. At some unknown time, a disturbance signal occurs and changes the distribution of the observations. The disturbance signal follows an AR model, which is dependent over time. Before the change, observations only consist of measurement noise, and are independent and identically  distributed (i.i.d.). After the change, observations consist of the disturbance signal and the measurement noise, are dependent over time, which essentially follow a continuous-state hidden Markov model (HMM). The goal is to design a stopping time to detect the disturbance signal as quickly as possible subject to false alarm constraints. Existing approaches for general non-i.i.d.\ settings and discrete-state HMMs cannot be applied due to their high computational complexity and memory consumption, and they usually assume some asymptotic stability condition. In this paper, the asymptotic stability condition is firstly theoretically proved for the AR model by a novel design of forward variable and auxiliary Markov chain. 
% Since the disturbance signal is unobservable and the distribution of the current observation is independent of the previous observations given the current disturbance signal, the AR model is equivalent to a hidden Markov model (HMM). 
A computationally efficient Ergodic CuSum algorithm that can be updated \textit{recursively} is then constructed and is further shown to be asymptotically optimal. 
% The main idea is that the likelihood ratio function of the observations is represented by the integration of a Gaussian function, whose parameters can be updated recursively. An auxiliary Markov chain is further constructed using the parameters of the Gaussian function and the observation. The test statistic is constructed using the likelihood ratio function, which is represented as a function of the auxiliary Markov chain. By showing the ergodicity of the auxiliary Markov chain and applying the ergodic theorem to the auxiliary Markov chain, the asymptotic optimality of the Ergodic CuSum algorithm is proved. 
The data-driven setting where the disturbance signal parameters are unknown is further investigated, and an online and computationally efficient gradient ascent CuSum algorithm is designed. The algorithm is constructed by iteratively updating the estimate of the unknown parameters based on the maximum likelihood principle and the gradient ascent approach. The lower bound on its average running length to false alarm is also derived for practical false alarm control. Simulation results are provided to demonstrate the performance of the proposed algorithms.
\end{abstract}
\begin{IEEEkeywords}
	Hidden Markov model, forward variable, sequential change detection, asymptotic optimality, non-i.i.d..
\end{IEEEkeywords}

\section{Introduction}
The problem of quickest change detection (QCD) has been widely studied in the literature \cite{poor-hadj-qcd-book-2009,veeravalli2013quickest,tartakovsky2014sequential,xie2021sequential}, where the goal is to detect an abrupt change in the data-generating distribution as quickly as possible, subject to false alarm constraints. For the basic setting where the observations are independent over time, algorithms, e.g.,  cumulative sum (CuSum) algorithm and Shiryaev-Roberts algorithm have been proposed and have been shown to be optimal  \cite{page1954continuous,pollak1985optimal,moustakides1986optimal,ritov1990decision,tartakovsky1995asymptotic}. However, for a wide range of practical applications, observations may not be independent over time. For example, in power systems,
the faulty signal is dependent in time \cite{gu2000armodel} and such dependency is usually exploited by an autoregressive (AR) model \cite{sidorov2010non, li2013monitor, li2016cooperative, kurt2018distributed, chen2018quickest}.

In this paper, we investigate the problem of QCD in AR models. Specifically, before the change, the observed signal only consists of Gaussian measurement noise. After the change, a non-i.i.d. disturbance signal occurs in the system, and the observation consists of the disturbance signal and the measurement noise. The disturbance signal follows an AR model, and thus the observation equivalently follows a hidden Markov model (HMM) \cite{cappe2007inference}. The goal is to detect the occurrence of the disturbance signal as quickly as possible subject to false alarm constraints.

\subsection{Related Works}
The problem in this paper is closely related to QCD in Markov chains \cite{Yakir1994optimal,polansky2007detecting, raghavan2010quickest, darkhovsky2011change, xian2016online, chen2022change, ford2023exactly,lai1998information,zhang2023data1}, where the pre- and/or post-change samples follow the Markov model. In \cite{Yakir1994optimal}, optimal change detection algorithms were proposed for finite-state Markov chains under the Bayesian setting and the minimax setting. 
The continuous state Markov chain was studied in \cite{lai1998information}, where the CuSum algorithm was proved to be asymptotically optimal and its performance was characterized. In \cite{raghavan2010quickest}, the change detection in sensor networks under the Bayesian setting was studied where the change propagates across sensors and its propagation was modeled as a Markov process. It was shown that an exactly optimal algorithm
for QCD of Markov chains under the Bayesian setting is a threshold test on the posterior of no change has happened in \cite{ford2023exactly}.
In \cite{xian2016online}, two CuSum type schemes were proposed for the QCD of finite-state Markov chain with unknown post-change transition probabilities, and bounds on the average detection delay (ADD) and the average running length (ARL) were further given for the proposed schemes. A maximum mean discrepancy (MMD) based method for change detection of Markov kernels with unknown post-change kernel was proposed in \cite{chen2022change,zhang2023data1} for general state Markov chain. However, the algorithms and analyses for Markov chains cannot be applied to our problem directly. This is because with measurement noise, the post-change samples in this paper essentially follow an HMM, and the hidden state is not directly observable. 
% Since the unobservable hidden state follows a Markov chain and the distribution of observations depends on the hidden state, the analysis for the Markov chain can not be extended to HMM directly. %Specifically, given the observation at time $t-1$, the distribution of the observation at time $t$ is not independent of all previous observations at time $1, 2, \cdots, t-2$. Therefore, it is more challenging to analyze the convergence of the joint likelihood function of the observations and thus the performance bounds than the Markov model.

The problem of QCD in discrete-state HMMs has been studied in \cite{chen2000detection,gereneser2004change,dayanik2009sequential, Fuh2015quickest,fuh2018asymptotic,ford2023hidden}.
In \cite{chen2000detection}, Page's test was extended to the QCD of discrete-state HMMs and the ARL under both hypotheses were approximated. In \cite{dayanik2009sequential}, the detection of the change point at which the discrete-state HMM enters the absorbing state and the identification of the absorbing set was investigated and optimal decision rule was provided. The QCD problem in two-state HMM was studied in \cite{Fuh2015quickest}, where two computationally efficient schemes were developed. In \cite{fuh2018asymptotic}, the Shiryaev algorithm was proved to be  asymptotically optimal under some regularity conditions for the QCD in discrete-state HMMs under the Bayesian setting. The exact optimal algorithm for QCD in discrete-state HMM under the Bayesian setting was established in \cite{ford2023hidden}.
These works mainly focus on the discrete-state HMM. The AR model in our paper is an HMM model with a continuous state space. To the best of the author's knowledge, research on QCD in continuous HMMs is rather limited. The performance characterization of Shiryaev algorithm for the discrete-state HMM relies on the limiting theory for products of random matrices \cite{Bougerol1988ThormesLP} and the computational complexity scales with the size of the state space. For the continuous-state Markov chain, the theory of random matrices is not applicable anymore since the number of state is uncountably infinite. 

%In \cite{li2013monitor,chen2018quickest,kurt2018distributed}, the QCD problem in AR models was studied, and data-driven approaches was developed. However, no theoretical guarantees on the developed algorithm were provided. 
%We further note that the CuSum algorithm in \cite{lai1998information} for general non-i.i.d. samples is not applicable for our problem since the joint distribution of all samples  depends on the change point when there is one.  
The AR models are commonly used in power system to characterize the disturbance signal and inter-area oscillations \cite{sidorov2010non, li2013monitor, li2016cooperative, kurt2018distributed, chen2018quickest}. In \cite{bagshaw1977sequential,davis1995testing,gombay2008change,li2013monitor,kang2014parameter,hilgert2016change,akashi2018change}, the change detection in the AR model was studied where the observation is an AR time series. Our problem is fundamentally different from \cite{bagshaw1977sequential,davis1995testing,gombay2008change,li2013monitor,kang2014parameter,hilgert2016change,akashi2018change} since the disturbance signal which has AR structure is unobservable in our problem. The observation is a noised version of the AR time series which makes our problem more challenging.
The quickest detection of cyber-attacks in discrete-time linear dynamic system, which is an AR model, was studied in \cite{kurt2018distributed}. The Kalman filter \cite{kalman1960new} was first performed to estimate the state and the generalized CuSum algorithm was further designed using the state estimation. Though the Kalman filter is effective in state estimation, the performance of the detection algorithm is difficult to characterize due to the estimation error. In \cite{kalus2015distributed}, a robust algorithm was proposed for distributed change detection in AR models with noised observations. However, no theoretical guarantee was provided.
In our paper, we design an Ergodic CuSum algorithm directly and show that it is asymptotically optimal under the Lorden's criterion \cite{lorden1971procedures}.

In our paper, we consider both the model-based setting where parameters of the disturbance signal are known and the data-driven setting where parameters of the disturbance signal are unknown. For the data-driven setting, our problem can be viewed as a composite QCD problem \cite{siegmund1995using,lai1998information,banerjee2015composite,zou2018quickest,mei2006sequential,Brodsky2005,pergamenchtchikov2019asymptotically, cao2018sequential}. Existing works mainly assume that the samples are i.i.d., while in the AR models, the observations are dependent in time. In \cite{chen2018quickest}, the QCD in AR model was studied where the properties of the post-change signal are unknown. Based on the assumption that the parameters of the post-change signal were small, the generalized likelihood ratio test (GLRT) was proposed. Our work is different from the works in \cite{chen2018quickest} since we do not assume that the parameter of the post-change signal is small. Moreover, simulation results demonstrate that our proposed algorithm outperforms the GLRT in \cite{chen2018quickest}.

% The CuSum algorithm for dependent observations in \cite{lai1998information} is not applicable for our problem since it is computationally expensive. 
\subsection{Contributions and Major Challenges}
For the model-based setting, we  design an Ergodic CuSum algorithm which can be updated recursively and thus is computationally efficient, and we theoretically prove that it is asymptotically optimal. While the general theory for non-i.i.d. setting were developed in e.g., \cite{lai1998information,tartakovsky2005general,baron2006asymptotic,tartakovsky2017asymptotic}, it relies on some asymptotic stability condition that the normalized log likelihood ratio between the post- and pre-change distributions converges to some finite and positive number. However, there are very few studies that try to verify this condition, e.g., discrete HMMs in \cite{fuh2018asymptotic} and Markov model in \cite{tartakovsky2005general}.
For the AR model, whether such stability condition holds remains unexplored. 
In this paper, we show that the normalized log likelihood ratio converges to some $\mathcal{K}>0$ almost surely under the post-change distribution, i.e., satisfies the stability condition.

The difficulty of analyzing the convergence of the log likelihood ratio lies in that the log likelihood function for the AR model has a non-additive form due to the non-i.i.d. nature of the observation. 
For the discrete-state HMM in \cite{fuh2018asymptotic}, the likelihood ratio function is represented by the $L_1$-norm of products of Markov random matrices and thus the log likelihood ratio has an additive form. However, the Markov random matrices techniques can not be applied to our problem since our AR model has uncountably infinite hidden states. 
To overcome this difficulty, we represent the likelihood function as the integration of the product of a sequence of functions based on the hidden Markov structure of the observation. We further show that the likelihood function can be represented as the integration of a scaled Gaussian density function. We define this scaled Gaussian density function as the forward variable and show that it can be updated recursively. Since the scaled Gaussian density function has parameterized representation, we develop a novel approach to write the conditional likelihood ratio using the parameter of the forward variable and the observation. The log likelihood ratio can then be written as the sum of the conditional log likelihood ratio and thus has an additive form. To characterize the convergence of the log likelihood ratio, we design an auxiliary Markov chain using the parameter of the forward variable and the observation and represent the likelihood function as a function of the auxiliary Markov chain. We further show that the stationary distribution of this auxiliary Markov chain exists and then show that $\mathcal{K}$ exists by applying the ergodic theorem \cite{roberts2004general, meyn2012markov} to the auxiliary Markov chain. With the convergence of the log likelihood ratio, we derive the universal lower bound on the worst-case average detection delay (WADD). 

For the CuSum algorithm designed for the general non-i.i.d. setting, e.g.,  \cite{lai1998information}, it is computationally expensive for our AR model since the likelihood ratio for our AR model depends on the change point and we need to compute the likelihood ratio for every possible change point. Motivated by the dependence of the WADD lower bound on ${\mathcal{K}}$, we aim to find a computationally efficient test statistic that has a positive drift of $\mathcal{K}$ under the post-change distribution and has a negative drift under the pre-change distribution. We propose a novel Ergodic CuSum algorithm by exploiting the ergodicity of the underlying HMM and employing the likelihood ratio of the observations with change point at time $1$. Since the likelihood function can be written as a function of the observation and the forward variable, our Ergodic CuSum algorithm can then be updated recursively, and thus is computationally efficient. 
From the ergodic theorem \cite{roberts2004general, meyn2012markov} of Markov chain, the long term statistical performance of our test statistic is independent of the state when the change occurs. Based on this fact, the test statistic of our Ergodic CuSum algorithm converges to $\mathcal{K}$ in the post-change regime despite of the fact that time 1 is not the true change point. Under the pre-change distribution, we show that our test statistic has a negative drift. Moreover, since our test statistic is a likelihood ratio function of the current observation, the ARL lower bound can be derived following the proof of the ARL lower bound for general non-i.i.d. case.
Therefore, we show that when satisfying the constraint on the average running length (ARL), the WADD upper bound of our Ergodic CuSum algorithm matches with the universal lower bound and thus the asymptotic optimality of our Ergodic CuSum algorithm follows. 
% We also provide numerical results to validate the results.

For the data-driven setting, the generalized likelihood ratio test (GLRT) which replaces the unknown post-change parameter with its maximum likelihood estimate (MLE) \cite{lai1998information} is not computationally efficient. We propose an online gradient ascent CuSum algorithm (OGA-CuSum) based on the online convex optimization (OCO) algorithm \cite{raginsky2012sequential, cao2018sequential}. We iteratively update the estimate of the unknown parameters based on the maximum likelihood principle using gradient ascent. Specifically, at each time, we compute the gradient of the log likelihood ratio function with respect to the unknown parameters based on the current observation and update the estimate in the direction of the gradient. We then replace the unknown parameter in the Ergodic CuSum with its estimate to construct the OGA-CuSum. The OGA-CuSum can be updated recursively using only the most recent sample and thus is memory and computationally efficient. We derive a lower bound on its ARL so that a threshold can be chosen analytically to control the false alarm. We also provide simulation results to demonstrate the performance of our algorithms. 

\subsection{Paper Organization}
The remainder of this paper is organized as follows. In Section \ref{sec:problemmodel}, we present the problem formulation. In Section \ref{sec:waddlower}, we present the design of a forward variable and establish the universal lower bound on WADD. In Section \ref{sec:algorithm}, we develop a computationally efficient Ergodic CuSum algorithm and show that it is asymptotically optimal under Lorden's criterion. In Section \ref{sec:datadriven}, we design an online gradient ascent CuSum algorithm for the case where the post-change parameters are unknown. Numerical results are provided to demonstrate the performance of our algorithms in \ref{sec:numerical}. In Section \ref{sec:conclusion}, we present some concluding remarks.

\section{Problem formulation}\label{sec:problemmodel}
Consider a system being monitored by a sequentially observed signal $\bm{y}_t \in\mathbb{R}^K, t = 1, 2,\cdots$. At some unknown time $t_0$, a disturbance signal occurs in the system and changes the distribution of the observed signal. Specifically, before the change point $t_0$, the observed signal consists of the measurement noise $\bm{\nu}_t$ only, which is Gaussian distributed and independent over time:
\begin{flalign}
\bm{y}_t = \bm{\nu}_t \sim \mathcal{N}(\bm 0, \bm{I}),\ t < t_0,
\end{flalign}
where $\bm I$ is the $K\times K$ identity matrix. The post-change signal consists of the disturbance signal $\bm {x}_t$ and the measurement noise:
\begin{flalign}\label{eq:ymod}
\bm{y}_t = \bm{x}_t + \bm{\nu}_t,\ t \geq t_0,
\end{flalign}
Note that the identity covariance matrix of $\bm\nu_t$ can be satisfied by applying a linear transformation to whiten the noise in the observation $\bm{y}_t$. The disturbance signal follows an AR model \cite{sidorov2010non, li2013monitor, li2016cooperative, kurt2018distributed, chen2018quickest}:
\begin{flalign}\label{eq:highorder}
    \bm x_t = \sum_{i=1}^q \bm A_i \bm x_{t-i} + \bm\omega_t,
\end{flalign}
where $\bm A_i \in \mathbb{R}^{K\times K}$ is the matrix coefficient and is invertible, $\bm\omega_t \in \mathbb R^K$ is the innovation noise vector and $\bm\omega_t\sim \mathcal{N}(\bm 0, \bm R_\omega)$ and $q$ denotes the order of the AR model. 
The disturbance signal $\bm x_t$ is not directly observable. 
The goal is to detect the change at time $t_0$ as soon as possible subject to false alarm constraints.

Denote by $\mathbb{P}_\infty$ the probability measure when there is no change, and denote by $p_\infty$ the corresponding probability density. 
% Under $\mathbb{P}_\infty$, the probability density of $\bm y_t$ given $\bm y_1,\cdots ,\bm y_{t-1}$ is $p_\infty(\bm y_t)$. 
For any $t_0>0$, denote by $\mathbb{P}_{t_0}$ the probability measure when the change happens at $t_0$, and denote by $p_{t_0}$ the corresponding probability density. It is clear that 
\begin{align}
    p_\infty(\bm y_t|\bm y_1,\cdots ,\bm y_{t-1})&=p_\infty(\bm y_t),\nn\\
    p_{t_0}(\bm y_t|\bm y_{1},\cdots, \bm y_{t-1})&=p_{\infty}(\bm y_t), \text{ if } t < t_0,\nn\\
    p_{t_0}(\bm y_t|\bm y_{1},\cdots, \bm y_{t-1})&=p_{t_0}(\bm y_t|\bm y_{t_0},\cdots, \bm y_{t-1}),\text{ if } t \geq t_0.
\end{align}

% $$, $p_{t_0}(\bm y_t|\bm y_{1},\cdots, \bm y_{t-1})=p_{\infty}(\bm y_t)$ for $t < t_0$, and $p_{t_0}(\bm y_t|\bm y_{1},\cdots, \bm y_{t-1})=p_{t_0}(\bm y_t|\bm y_{t_0},\cdots, \bm y_{t-1})$ for $t \geq t_0$. %If the change occurs at time $t_0$, then the probability density of $\bm y_{t_0}$ changes from $p_\infty$ to $p_{t_0}$. We use $p_{t_0}$ to denote the probability density in the post-change phase to emphasize its dependence on the change point $t_0$.

The AR model is equivalent to an HMM. Let $f_1(\bm x_{t_0}, \bm x_{t_0+1}, \cdots, \bm x_{t_0+q-1})$ be the joint probability density of the initial state $(\bm x_{t_0}, \bm x_{t_0+1}, \cdots, \bm x_{t_0+q-1})$. Let $f(\bm x_t|\bm x_{t-q}, \cdots, \bm x_{t-1})$ denote the transitional probability density of the HMM and $g(\bm y_t|\bm x_t)$ be the conditional probability density of the observation $\bm y_t$ given the hidden state $\bm x_t$. 
The probability density $p_{t_0}$ then satisfies that for $t\geq t_0+q$
\begin{flalign}
p_{t_0}(\bm y_1, \cdots, \bm y_t) = &p_{\infty}(\bm y_1, \cdots, \bm y_{t_0-1})\cdot \int f_1(\bm x_{t_0}, \bm x_{t_0+1}, \cdots, \bm x_{t_0+q-1})g(\bm y_{t_0}|\bm x_{t_0})\cdots g(\bm y_{t_0+q-1}|\bm x_{t_0+q-1})\nn\\&  f(\bm x_{t_0+q}|\bm x_{t_0}, \cdots, \bm x_{t_0+q-1})\cdots g(\bm y_{t}|\bm x_{t}) d\bm x_{t_0}\bm x_{t_0+1}\cdots\bm x_{t}.
\end{flalign}

In this paper, we consider a deterministic but unknown change point $t_0$. The goal is to detect the change as quickly as possible subject to false alarm constraints based on the sequentially observed samples. Let $\mathcal{F}_t$ be the $\sigma$-algebra generated by the first $t$ samples $\bm y_1, \bm y_2, \cdots, \bm y_t$. A stopping time $\tau$ is a random variable with the property that for each $t$, the event $\{\tau = t\} \in \mathcal{F}_t$. We define the worse-case average detection delay (WADD) under Lorden's criterion \cite{lorden1971procedures} and the average running length (ARL) for any stopping time $\tau$ as follows:
\begin{flalign}
\text{WADD}(\tau) &\triangleq \sup_{t_0\geq 1}\text{esssup} \mathbb{E}_{t_0}\left[(\tau-t_0)^+|\bm y_1, \cdots, \bm y_{t_0-1}\right],\nn\\
\text{ARL}(\tau) &\triangleq \mathbb{E}_\infty[\tau],
\end{flalign}
where $\mathbb{E}_{t_0}$ ($\mathbb{E}_\infty$) denotes the expectation under the probability measure $\mathbb{P}_{t_0}$ ($\mathbb{P}_{\infty}$). The goal is to design a stopping time that minimizes the WADD subject to the constraint on the ARL:
\begin{flalign}\label{eq:goal}
\inf_{\tau:\text{ARL}(\tau)\geq \gamma} \text{WADD}(\tau),
\end{flalign}
where $\gamma>0$ is a pre-specified threshold.

%\section{Hidden Markov Model}\label{sec:hmm}
%Therefore, after the change point, $\{\bm x_t, \bm y_t\}_{t=t_0}^{\infty}$ can be viewed as an HMM, where $\bm x_t$ is the unobservable hidden state. Moreover, we have that the transition probability $p_1(\bm x_t|\bm x_{t-1})$ is a Gaussian distribution with mean $\bm A\bm x_{t-1}$ and covariance matrix $\bm R_\omega$ and $p_1(\bm y_t|\bm x_t)$ is a Gaussian distribution with mean $\bm x_t$ and covariance matrix $\bm I$. 

\begin{comment}
\begin{proof}
We have that for any $(\bm x, \bm y)$ and any set $E\in\mathbb{R}^{K\times K}$ such that $\pi(E)>0$,
\begin{flalign}
&P\big((\bm x, \bm y), E\big) \nn\\&= \int_E p\big((\bm x^\prime, \bm y^\prime)\big|(\bm x, \bm y)\big)d\bm x^\prime\bm y^\prime \nn\\&= \int_E \frac{1}{\sqrt{(2\pi)^K \det(\bm R_\omega)}}\exp\Big(-\frac{1}{2}(\bm x^\prime - \bm A \bm x)^T \bm R_\omega^{-1}(\bm x^\prime - \bm A \bm x)\Big)\nn\\&\hspace{0.2cm}\frac{1}{\sqrt{(2\pi)^K}}\exp\Big(-\frac{1}{2}(\bm y^\prime - \bm x^\prime)^T (\bm y - \bm x^\prime)\Big)d\bm x^\prime\bm y^\prime\nn\\& > 0.
\end{flalign}
From the definition of irreducible \cite{meyn2012markov} Markov chain, we have that $\{\bm x_t, \bm y_t\}_{t=1}^{\infty}$ is $\pi$-irreducible.
\end{proof}
\end{comment}

\section{Universal Lower Bound on WADD}\label{sec:waddlower}
For convenience, we first present results for the first-order AR model to illustrate our approach. 
We will then show the generalization to any $q$-th order AR model. We note that any $q$-th order AR model can be equivalently converted to a first-order AR model (see details in Appendix \ref{app:convert}).

Specifically, the first-order AR model is defined as:
\begin{flalign}\label{eq:xmod}
\bm x_t = \bm A \bm x_{t-1} + \bm\omega_t,
\end{flalign}
where $\bm A \in \mathbb{R}^{K\times K}$ is the matrix coefficient and is invertible. The initial disturbance signal $\bm x_{t_0}$ is assumed to be Gaussian with probability density $f_1(\bm x_{t_0})$ and is independent of the pre-change observations. We assume that the operator norm of $\bm A$ is strictly less than 1. This assumption guarantees the stability of this system \cite{bhatia2002stability}.

In the following, we derive a universal lower bound on the WADD for any $\tau$ satisfying the false alarm constraint: $\mE_\infty[\tau]\geq \gamma$. To derive the universal lower bound, we first prove the following stability condition \cite{lai1998information,tartakovsky2005general,baron2006asymptotic,tartakovsky2017asymptotic}: convergence of the  log likelihood ratio
\begin{flalign}\label{eq:logratio}
\lim_{t\rightarrow\infty}\frac{1}{t}\log\frac{p_{t_0}(\bm y_{t_0},\cdots, \bm y_{t_0+t-1})}{p_\infty(\bm y_{t_0},\cdots, \bm y_{t_0+t-1})}.
\end{flalign}

Note that for the AR model, the log likelihood ratio in \eqref{eq:logratio} does not admit an additive form, and thus its convergence analysis is challenging. To overcome this difficulty, we first introduce a forward variable $\alpha_t$, which is a scaled Gaussian density function and can be updated recursively. We show that the conditional density $p_{t_0}(\bm y_t|\bm y_{t_0}, \cdots, \bm y_{t-1})$ can be represented using $\bm y_{t}$ and the parameters of the forward variable $\alpha_t$. The log likelihood ratio in \eqref{eq:logratio} can then be written in an additive form. To further prove the convergence of the log likelihood ratio, we construct an auxiliary Markov chain using $\bm y_t$ and the parameter of the forward variable $\alpha_t$. We show that the stationary distribution of this auxiliary Markov chain exists and then show the convergence of the log likelihood ratio in \eqref{eq:logratio} by applying the ergodic theorem \cite{roberts2004general,meyn2012markov}. The universal lower bound on the WADD can then be derived.

\subsection{Forward Variable}
We first introduce a forward variable, which plays an important role in deriving the universal lower bound on the WADD. Moreover, our proposed algorithm in Section \ref{sec:algorithm} can be recursively updated using the forward variable.
%In the rest of this section, to simplify the notation, we assume that the change point $t_0 = 1$. 
%Under the post-change distribution, we have that 
%\begin{flalign*}
%&p^{t_0}(\bm y_{t_0}, \cdots, \bm y_t) = \int p_1(\bm x_{t_0})p_1(\bm y_{t_0}|\bm x_{t_0})p_1(\bm x_{t_0+1}|\bm x_{t_0})\nn\\&\hspace{1cm}\cdots p_1(\bm y_t|\bm x_t)d\bm x_{t_0} \bm x_{t_0+1}\cdots\bm x_t. 
%\end{flalign*}

Note that for the i.i.d. setting, the log likelihood ratio of $\bm y_{t_0}, \cdots, \bm y_t$ can be written as the sum of the log likelihood ratio of each individual sample. For the discrete state HMM in \cite{fuh2018asymptotic}, the likelihood ratio function is represented by the $L_1$-norm of products of Markov random matrices and thus the log likelihood ratio has an additive form. However, the Markov random matrices techniques can not be applied to our problem since our AR model has uncountably infinite hidden states. 

Observe that $p_{t_0}(\bm y_{t_0}, \cdots, \bm y_t) = \int f_1(\bm x_{t_0})g(\bm y_{t_0}|\bm x_{t_0})f(\bm x_{t_0+1}|\bm x_{t_0})\cdots g(\bm y_{t}|\bm x_{t}) d\bm x_{t_0}\bm x_{t_0+1}\cdots\bm x_{t}$ is the integration of the products of a sequence of functions. This motivates us to represent the likelihood ratio function by replacing the Markov random matrices and the $L_1$-norm of matrices in \cite{fuh2018asymptotic} with functions and integration of functions, respectively. Moreover, we leverage the Gaussian property of the innovation noise and the measurement noise in our AR model to further show that these functions can be represented using the parameters of the Gaussian density function. Therefore, the likelihood ratio function admits an additive form and can be updated efficiently.

We define the following forward variable to compute the likelihood function after the change point ($t\geq t_0$),
\begin{flalign}
\alpha_t(\bm x_t) = p_{t_0}(\bm y_{t_0}, \cdots, \bm y_t, \bm x_t).
\end{flalign}
Note that $\alpha_t(\bm x_t)$ also depends on $\bm y_{t_0}, \cdots, \bm y_t$. For notational convenience, we treated $\bm y_{t_0}, \cdots, \bm y_t$ as known parameters here and write $\alpha_t$ as only a function of $\bm x_t$.
We then have that 
\begin{flalign}
p_{t_0}(\bm y_{t_0}, \cdots, \bm y_t) = \int \alpha_t(\bm x_t)d\bm x_t.
\end{flalign}
It can be easily verified that the following recursion holds for the forward variable:
\begin{flalign}
\alpha_{t+1}(\bm x_{t+1}) = \int \alpha_t(\bm x_{t}) f(\bm x_{t+1}|\bm x_t)g(\bm y_{t+1}|\bm x_{t+1}) d\bm x_t. \nn
\end{flalign}

We first provide a formal definition for the Gaussian function.

\begin{definition}
A Gaussian function $f^\prime(\cdot): \mathbb{R}^K\rightarrow \mathbb{R}$ is a function of the form 
\begin{flalign}
    f^\prime(\bm x) = a\exp\Big(-\frac{1}{2}(\bm x-\bm \mu)^\top \bm\Sigma^{-1}(\bm x-\bm \mu)\Big),
\end{flalign}
where $a$ is a constant, $\bm \mu\in\mathbb{R}^K$ is an arbitrary vector and $\bm\Sigma\in\mathbb{R}^{K\times K}$ is a positive definite matrix.
\end{definition}

In the following lemma, we show that $\alpha_t(\bm x_t)$ is a Gaussian function of $\bm x_t$.
\begin{lemma}\label{lemma:forward}
% If the distribution of $\bm x_{t_0}$ is Gaussian, 
$\alpha_t(\bm x_t)$ is a Gaussian function of $\bm x_t$.
\end{lemma}
\begin{proof}
% Though the proof is straightforward, we provide the proof here to introduce the update rule of the forward variable. 

We prove this by induction. 
If the distribution of $\bm x_{t_0}$ is Gaussian, $\alpha_1(\bm x_{t_0}) = f_1(\bm x_{t_0})g(\bm y_{t_0}|\bm x_{t_0})$ is a Gaussian distribution times a coefficient, which is a Gaussian function. This can be proved by computing the product of two Gaussian density functions. 
Assume that for some $c_{t-1},\bm\Sigma_{t-1}$ and $\bm\mu_{t-1}$,
\begin{flalign}
&\alpha_{t-1}(\bm x_{t-1}) = \frac{c_{t-1}}{\sqrt{(2\pi)^K\det(\bm \Sigma_{t-1})}}\exp\big(-\frac{1}{2}(\bm x_{t-1}-\bm\mu_{t-1})^\top\bm\Sigma_{t-1}^{-1}(\bm x_{t-1}-\bm\mu_{t-1})\big),
\end{flalign}
which is a Gaussian function of $\bm x_{t-1}$. We will show that $\alpha_t(\bm x_{t})$ is also a Gaussian function of $\bm x_t$. 
From the product and the convolution of Gaussian density functions, we have that
\begin{flalign}
&\alpha_t(\bm x_{t}) = \int \alpha_t(\bm x_{t-1}) f(\bm x_{t}|\bm x_{t-1})g(\bm y_{t}|\bm x_{t}) d\bm x_{t-1} \nn\\&= \frac{c_{t}}{\sqrt{(2\pi)^K\det(\bm \Sigma_{t})}}\exp\big(-\frac{1}{2}(\bm x_{t}-\bm\mu_t)^\top\bm\Sigma_{t}^{-1}(\bm x_{t}-\bm\mu_t)\big),
\end{flalign}
where 
\begin{flalign}\label{eq:updaterule}
&\bm\Sigma_t = (\bm A\bm\Sigma_{t-1}\bm A^\top + \bm R_{\omega})(\bm A\bm\Sigma_{t-1}\bm A^\top + \bm R_{\omega} +\bm I)^{-1},\nn\\&
\bm \mu_{t} = (\bm A\bm\Sigma_{t-1}\bm A^\top + \bm R_{\omega} +\bm I)^{-1}\bm A\bm\mu_{t-1} + (\bm A\bm\Sigma_{t-1}\bm A^\top + \bm R_{\omega})(\bm A\bm\Sigma_{t-1}\bm A^\top + \bm R_{\omega} +\bm I)^{-1}\bm y_t,
\end{flalign}
and 
\begin{flalign}\label{eq:forwardsigma}
&\frac{c_t}{c_{t-1}} = \frac{1}{\sqrt{(2\pi)^K\det(\bm A\bm\Sigma_{t-1}\bm A^\top + \bm R_{\omega})}}\frac{1}{\sqrt{\det((\bm A\bm\Sigma_{t-1}\bm A^\top + \bm R_{\omega})^{-1} + \bm I)}}\nn\\
&\hspace{0.5cm}\cdot\exp\bigg(-\frac{1}{2}\Big((\bm A\bm\mu_{t-1})^\top(\bm A\bm\Sigma_{t-1}\bm A^\top + \bm R_{\omega})^{-1}(\bm A\bm\mu_{t-1})\nn\\
&\hspace{0.5cm}+\bm y_{t}^\top \bm y_t - \big((\bm A\bm\Sigma_{t-1}\bm A^\top + \bm R_{\omega})^{-1}(\bm A\bm\mu_{t-1}) + \bm y_t\big)^\top\nn\\&\hspace{0.5cm}\big((\bm A\bm\Sigma_{t-1}\bm A^\top + \bm R_{\omega})^{-1} + \bm I\big)^{-1}\big((\bm A\bm\Sigma_{t-1}\bm A^\top + \bm R_{\omega})^{-1}(\bm A\bm\mu_{t-1}) + \bm y_t\big) \Big)\bigg).
\end{flalign}
Moreover, from the update rule of $\bm \Sigma_t$, it can be seen that $\bm \Sigma_t$ is positive definite. Therefore, $\alpha_t(\bm x_t)$ is a scaled Gaussian distribution of $\bm x_t$ with coefficient $c_t$.
% 
% This completes the proof.
\end{proof}
With the forward variable, the conditional density $p_{t_0}(\bm y_t|\bm y_{t_0}, \cdots, \bm y_{t-1})$ can be written as follows:
\begin{flalign}
&p_{t_0}(\bm y_t|\bm y_{t_0}, \cdots, \bm y_{t-1}) = \frac{p_{t_0}(\bm y_{t_0},\cdots, \bm y_t)}{p_{t_0}(\bm y_{t_0},\cdots, \bm y_{t-1})}\nn\\&  = \frac{\int\alpha_t(\bm x_t)d\bm x_t}{\int\alpha_{t-1}(\bm x_{t-1})d\bm x_{t-1}} = \frac{c_t}{c_{t-1}},
\end{flalign}
where the last equality is due to the fact that $\alpha_t(\bm x_t)$ is a scaled Gaussian density function of $\bm x_t$ with coefficient $c_t$. Note that before the change, the observations are i.i.d.. Therefore, it suffices to consider the additive form for the post-change log likelihood function. We have the following additive form for the log likelihood function $\log p_{t_0}(\bm y_{t_0},\cdots, \bm y_t)$:
\begin{flalign}\label{eq:additive}
&\log p_{t_0}(\bm y_{t_0},\cdots, \bm y_t)\nn\\& =\log \frac{p_{t_0}(\bm y_{t_0},\cdots, \bm y_t)}{p_{t_0}(\bm y_{t_0},\cdots, \bm y_{t-1})} + \log \frac{p_{t_0}(\bm y_{t_0},\cdots, \bm y_{t-1})}{p_{t_0}(\bm y_{t_0},\cdots, \bm y_{t-2})}  + \cdots + \log\frac{p_{t_0}(\bm y_{t_0}, \bm y_{t_0+1})}{p_{t_0}(\bm y_{t_0})} + \log p_{t_0}(\bm y_{t_0})\nn\\&
 =\sum_{i=t_0}^t \frac{c_i}{c_{i-1}},
\end{flalign}
where $\frac{c_{t_0}}{c_{t_0-1}} =\log p_{t_0}(\bm y_{t_0})$. As shown in \eqref{eq:forwardsigma}, $\frac{c_t}{c_{t-1}}$ is a function of $\bm y_t, \bm \mu_{t-1}, \bm \Sigma_{t-1}$. Therefore, the log likelihood function can be written in an additive form using the parameters of the forward variable.

It can be seen that at each time step, we only need to compute $c_t, \bm \mu_t, \bm \Sigma_t$, and the forward variable $\alpha_t$ can be fully recovered. Moreover, $c_t, \bm \mu_t, \bm \Sigma_t$ can be updated recursively, and thus the likelihood function $p_{t_0}(\bm y_1, \cdots, \bm y_t)$ can be computed efficiently using the forward variable $\alpha_t$.

\subsection{Universal Lower Bound on WADD}\label{sec:waddorder1}
The general theory for QCD with non-i.i.d. samples relies on the assumption that the normalized log likelihood ratio between the post- and pre-change distributions converges to some finite and positive number \cite{lai1998information,tartakovsky2005general,baron2006asymptotic,tartakovsky2017asymptotic}. For the AR model, whether such assumption holds remains unexplored. Here, we first show that as $t \rightarrow \infty$, the limit of $\frac{1}{t}\log\frac{p_{t_0}(\bm y_{t_0},\cdots, \bm y_{t_0+t-1})}{p_\infty(\bm y_{t_0},\cdots, \bm y_{t_0+t-1})}$ exists and is positive, denoted by $\mathcal{K}$, and further provide the universal lower bound on WADD in the following theorem. The explicit expression of $\mathcal{K}$ will be provided later after we introduce necessary notations.
\begin{theorem}\label{theorem:waddlow}
We have that 
\begin{flalign}\label{eq:convergek}
\lim_{t\rightarrow\infty}\frac{1}{t}\log\frac{p_{t_0}(\bm y_{t_0},\cdots, \bm y_{t_0+t-1})}{p_\infty(\bm y_{t_0},\cdots, \bm y_{t_0+t-1})} = \mathcal{K} 
\end{flalign}
almost surely under $\mathbb{P}_{t_0}$ with $\mathcal{K}>0$. Moreover,
as $\gamma\rightarrow\infty$, 
\begin{flalign}\label{eq:waddup}
\inf_{\tau:\text{ARL}(\tau)\geq \gamma}\text{WADD}(\tau)\geq \frac{\log\gamma}{\mathcal{K}}(1+o(1)).
\end{flalign}
\end{theorem}
%Recall that in \cite{lai1998information}, it is assumed that $\frac{1}{t}\sum_{i=t_0}^{t_0+t} \log \frac{p_{t_0}(\bm y_i|\bm y_{1}, \cdots, \bm y_{i-1})}{p_\infty(\bm y_i|\bm y_{1}, \cdots, \bm y_{i-1})}$ converges in probability under $\mathbb{P}_{t_0}$ to some positive constant $\mathcal{K}$. To derive the lower bound of WADD in Theorem \ref{theorem:waddlow}, it is necessary to prove the existence of $\mathcal{K}$ for the AR model. 
Since $\bm y_{t_0}, \bm y_{t_0+1}, \cdots, \bm y_t$ are dependent, the Law of Large Number used in the i.i.d. setting is not applicable, even if the log likelihood function has an additive form in \eqref{eq:additive}. We develop a novel approach to show the convergence of $\lim_{t\rightarrow\infty}\frac{1}{t}\log \frac{p_{t_0}(\bm y_{t_0}, \cdots, \bm y_{t_0+t-1})}{p_\infty(\bm y_{t_0}, \cdots, \bm y_{t_0+t-1})}$ under $\mathbb{P}_{t_0}$.

Observe that $\frac{c_t}{c_{t-1}}$ is a function of $\bm y_t, \bm \mu_{t-1}, \bm \Sigma_{t-1}$ and $\bm y_t, \bm \mu_{t}, \bm \Sigma_{t}$ are dependent over time. This motivates us to formulate a new Markov chain using $\bm y_t, \bm \mu_{t}, \bm \Sigma_{t}$ and apply the ergodic theorem \cite{roberts2004general,meyn2012markov} to characterize the convergence of $\lim_{t\rightarrow\infty}\frac{1}{t}\log \frac{p_{t_0}(\bm y_{t_0}, \cdots, \bm y_{t_0+t-1})}{p_\infty(\bm y_{t_0}, \cdots, \bm y_{t_0+t-1})}$ under $\mathbb{P}_{t_0}$.

%Therefore, starting from any state $\{\bm x_{t_0}, \bm y_{t_0}\}$, after a sufficiently long time, the distribution of $\{\bm x_t, \bm y_t\}$ will not change anymore, which is crucial in characterizing the long term statistical performance of $\{\bm x_t, \bm y_t\}$.

From the update rule of $\bm \Sigma_t$ in \eqref{eq:updaterule}, it can be seen that $\bm \Sigma_t$ does not depend on $\bm y_t$ and thus is deterministic. 
% 
% Therefore, if we formulate a Markov chain using $\bm y_t, \bm \mu_{t}, \bm \Sigma_{t}$ directly, the new Markov chain is not irreducible and thus the ergodic theorem of Markov chain \cite{roberts2004general,meyn2012markov} can not be applied.
We then show that $\bm \Sigma_t$ converges as $t\rightarrow \infty$ and let $\bm \Sigma^* \triangleq \lim_{t\rightarrow\infty}\bm \Sigma_t$. Let $\bm\mu^*_t$ be the parameter of the forward variable when $\bm\Sigma_{t_0} = \bm\Sigma^*$. We formulate an auxiliary Markov chain using $\bm y_t, \bm\mu^*_t$ as an intermediate step to prove the convergence of the log likelihood ratio. Specifically,
denote by $p_{t_0}^*(\bm y_{t_0},\cdots, \bm y_t)$ the likelihood function when $\bm \Sigma_{t_0} = \bm\Sigma^*$. With the auxiliary Markov chain $\{\bm y_t, \bm\mu^*_t\}_{t={t_0}}^\infty$, we show that $\lim_{t\rightarrow\infty} \frac{1}{t}p_{t_0}^*(\bm y_{t_0},\cdots, \bm y_{t_0+t-1})$ converges under $\mathbb{P}_{t_0}$. The convergence of $\lim_{t\rightarrow\infty}\frac{1}{t}\log \frac{p_{t_0}(\bm y_{t_0}, \cdots, \bm y_{t_0+t-1})}{p_\infty(\bm y_{t_0}, \cdots, \bm y_{t_0+t-1})}$ under $\mathbb{P}_{t_0}$ is then proved by showing that 
\begin{flalign}
&\lim_{t\rightarrow\infty}\frac{1}{t}\big(\log p_{t_0}(\bm y_{t_0}, \cdots, \bm y_{t_0+t-1}) - \log p_{t_0}^*(\bm y_{t_0}, \cdots, \bm y_{t_0+t-1})\big) = 0
\end{flalign}
under $\mathbb{P}_{t_0}$ almost surely.

We first show that $\bm \Sigma_t$ converges in the following Lemma.
\begin{lemma}\label{lemma:sigma}
As $t\rightarrow\infty$, $\bm \Sigma_{t}$ converges to $\bm \Sigma^*$.
\end{lemma}
\begin{proof}
The proof can be found in Appendix \ref{sec:lemma4}.
\end{proof}

To construct our auxiliary Markov chain, let $\bm\Sigma_{t_0} = \bm\Sigma^*$. Then $\bm\Sigma_{t} = \bm\Sigma^*,\forall t\geq t_0$. It follows that $\bm\mu^*_t$ can be updated recursively as follows:
\begin{flalign}\label{eq:mustar}
\bm\mu^*_t &= (\bm A\bm\Sigma^*\bm A^\top + \bm R_{\omega} +\bm I)^{-1}\bm A\bm\mu_{t-1}^* + (\bm A\bm\Sigma^*\bm A^\top + \bm R_{\omega})(\bm A\bm\Sigma^*\bm A^\top + \bm R_{\omega} +\bm I)^{-1}\bm y_t \nn\\& = (\bm A\bm\Sigma^*\bm A^\top + \bm R_{\omega} +\bm I)^{-1}\bm A\bm\mu_{t-1}^* + \bm\Sigma^*\bm y_t.
\end{flalign}
\begin{comment}
Consider the random process $\{\bm x_t, \bm y_t, \bm \mu_t^*\}_{t=1}^{\infty}$. From the update rule of $\bm\mu^*_t$, we have that 
\begin{flalign}
&p_1(\bm\mu_t^*|\bm x_1, \cdots, \bm x_{t}, \bm y_1,\cdots, \bm y_{t}, \bm \mu_{1}^*, \cdots, \bm \mu_{t-1}^*)\nn\\& = p_1(\bm\mu_{t}^*|\bm \mu_{t-1}^*, \bm y_t).
\end{flalign}
We then have that
\begin{flalign}
&p_1(\bm x_t, \bm y_t, \bm\mu_t^*|\bm x_1, \cdots, \bm x_{t-1}, \bm y_1,\cdots, \bm y_{t-1}, \bm \mu_{1}^*, \cdots, \bm \mu_{t-1}^*)\nn\\& =p_1(\bm x_t|\bm x_{t-1})p_1(\bm y_t|\bm x_t) p_1(\bm\mu_{t}^*|\bm \mu_{t-1}^*, \bm y_t)\nn\\& = p_1(\bm x_t, \bm y_t, \bm \mu_t^*|\bm x_{t-1}, \bm y_{t-1}, \bm \mu_{t-1}^*).
\end{flalign}
Therefore, $\{\bm x_t, \bm y_t, \bm \mu_t^*\}_{t=1}^{\infty}$ is a Markov chain.

We note that it is challenging to derive the stationary distribution of $\{\bm x_t, \bm y_t, \bm \mu_t^*\}_{t=1}^{\infty}$ or even prove the existence directly from the definition of stationary distribution, because the dependence of $\bm x_t, \bm y_t$ and $\bm\mu_t^*$ is implicit. 
\end{comment}

In the following, we denote the stationary distribution of a Markov chain in the post-change phase by $\pi$. For examples, $\pi(\bm x_t)$ is the stationary distribution of $\{\bm x_t\}_{t=t_0}^{\infty}$, $\pi(\bm x_t, \bm y_t)$ is the stationary distribution of $\{\bm x_t, \bm y_t\}_{t=t_0}^{\infty}$. We first show that $\{\bm x_t, \bm y_t\}_{t=t_0}^{\infty}$ has a unique stationary distribution in the following lemma, which will be used to formulate and characterize the property of the auxiliary Markov chain. 
\begin{lemma}\label{lemma:station}
The HMM $\{\bm x_t, \bm y_t\}_{t=t_0}^{\infty}$ is $\pi$-irreducible. Moreover, $\pi(\bm x_t) g(\bm y_t|\bm x_t)$ is the unique stationary distribution of $\{\bm x_t, \bm y_t\}_{t=t_0}^{\infty}$, where $\pi(\bm x_t)$ is a Gaussian distribution with mean $\bm 0$ and covariance matrix $\bm \Sigma = \sum_{i=0}^\infty (\bm A^\top)^i\bm R_{\omega}\bm A^{i}$.
\end{lemma}
\begin{proof}
The detailed proof can be found in Appendix \ref{sec:lemma1}.
\end{proof}
In the following lemma, we show that $\{\bm y_t, \bm\mu^*_t\}_{t=t_0}^\infty$ is a Markov chain and the stationary distribution of $\{\bm y_t, \bm\mu^*_t\}_{t=t_0}^\infty$ exists and is unique.
\begin{lemma}\label{lemma:ymumarkov}
$\{\bm y_t, \bm\mu^*_t\}_{t=t_0}^\infty$ is a Markov chain. Moreover, $\int\pi(\bm\mu_{t-1}^*)p_{t_0}(\bm y_t, \bm \mu^*_t|\bm\mu^*_{t-1})d\bm\mu^*_{t-1}$ is the unique stationary distribution of $\{\bm y_t, \bm\mu^*_t\}_{t=t_0}^\infty$, where $\pi(\bm \mu^*_{t-1})$ is the stationary distribution of $\bm\mu^*_{t-1}$ and is guaranteed to exist. 
\end{lemma}
\begin{proof}
The proof can be found in Appendix \ref{sec:lemma5}.
\end{proof}

With the auxiliary Markov chain $\{\bm y_t, \bm\mu^*_t\}_{t={t_0}}^\infty$ and its stationary distribution, we are ready to show the convergence of $\lim_{t\rightarrow\infty}\frac{1}{t}\log \frac{p_{t_0}(\bm y_{t_0}, \cdots, \bm y_{t_0+t-1})}{p_\infty(\bm y_{t_0}, \cdots, \bm y_{t_0+t-1})}$ and prove Theorem \ref{theorem:waddlow}.

\begin{proof}[Proof sketch of Theorem \ref{theorem:waddlow}]
We first consider the auxiliary Markov chain $\{\bm y_t, \bm\mu^*_t\}_{t={t_0}}^\infty$ and show that 
\begin{flalign}
\lim_{t\rightarrow\infty}\frac{1}{t}\log p_{t_0}^*(\bm y_{t_0},\cdots, \bm y_{t_0+t-1}) = \mathbb{E}_\pi[h(\bm\mu, \bm y)]
\end{flalign}
almost surely under $\mathbb{P}_{t_0}$ by applying the ergodic theorem \cite{roberts2004general,meyn2012markov} to the Markov chain $\{\bm y_t, \bm\mu^*_t\}_{t=t_0}^\infty$, where $h(\bm\mu,\bm y)$ is a quadratic function of $\bm\mu, \bm y$ and $(\bm y, \bm \mu)$ follows the stationary distribution of the auxiliary Markov chain $\{\bm y_t, \bm\mu^*_t\}_{t={t_0}}^\infty$. The explicit expression of $h(\bm\mu,\bm y)$ is as follows 
\begin{flalign}\label{eq:hmuy}
    &h(\bm\mu, \bm y)= \log\Big( \frac{1}{\sqrt{(2\pi)^K\det(\bm A\bm\Sigma^*\bm A^\top + \bm R_{\omega})}}\frac{1}{\sqrt{\det((\bm A\bm\Sigma^*\bm A^\top + \bm R_{\omega})^{-1} + \bm I)}}\Big) \nn\\&\hspace{0.2cm}-\frac{1}{2}\Big(((\bm A\bm\Sigma^*\bm A^\top + \bm R_{\omega} +\bm I)(\bm\mu - \bm\Sigma^*\bm y))^\top(\bm A\bm\Sigma^*\bm A^\top + \bm R_{\omega})^{-1}((\bm A\bm\Sigma^*\bm A^\top + \bm R_{\omega} +\bm I)(\bm\mu - \bm\Sigma^*\bm y))\nn\\&\hspace{0.2cm}+\bm y^\top \bm y - \big((\bm A\bm\Sigma^*\bm A^\top + \bm R_{\omega})^{-1}((\bm A\bm\Sigma^*\bm A^\top + \bm R_{\omega} +\bm I)(\bm\mu - \bm\Sigma^*\bm y)) + \bm y\big)^\top\nn\\&\hspace{0.2cm}\big((\bm A\bm\Sigma^*\bm A^\top + \bm R_{\omega})^{-1} + \bm I\big)^{-1} \big((\bm A\bm\Sigma^*\bm A^\top + \bm R_{\omega})^{-1}((\bm A\bm\Sigma^*\bm A^\top + \bm R_{\omega} +\bm I)(\bm\mu - \bm\Sigma^*\bm y)) + \bm y\big) \Big).
\end{flalign}
Moreover, if we let $\bm y = \bm y_t$ and $\bm \mu = \bm\mu^*_t$, we have that $h(\bm\mu^*_t,\bm y_t) = p^*_{t_0}(\bm y_t|\bm y_{t_0}, \cdots, \bm y_{t-1})$. 

We then show that $\lim_{t\rightarrow\infty}\frac{1}{t}\log p_{t_0}(\bm y_{t_0}, \cdots, \bm y_{t_0+t-1}) = \mathbb{E}_\pi[h(\bm\mu, \bm y)]$ under $\mathbb{P}_{t_0}$ almost surely by showing that 
\begin{flalign}
&\lim_{t\rightarrow\infty}\frac{1}{t}\big(\log p_{t_0}(\bm y_{t_0}, \cdots, \bm y_{t_0+t-1}) - \log p_{t_0}^*(\bm y_{t_0}, \cdots, \bm y_{t_0+t-1})\big) = 0
\end{flalign}
under $\mathbb{P}_{t_0}$ almost surely. 

Since the observations are independent before the change point, we have that 
\begin{align}
\lim_{t\rightarrow\infty}\frac{1}{t}\log p_\infty(\bm y_{t_0}, \cdots, \bm y_{t_0+t-1}) = \lim_{t\rightarrow\infty}\frac{1}{t}\sum_{i=t_0}^{t_0+t-1} \log p_\infty(\bm y_i).\nn    
\end{align}
We then have that under $\mathbb{P}_{t_0}$, 
$$\lim_{t\rightarrow\infty}\frac{1}{t}\sum_{i=t_0}^{t_0+t-1} \log p_\infty(\bm y_i) = \mathbb{E}_{\pi}[\log p_\infty(\bm y)]$$ almost surely from the ergodic theorem of Markov chain \cite{roberts2004general,meyn2012markov}.

Let 
\begin{align}\label{eq:def_K}
    \mathcal{K} = \mathbb{E}_\pi[h(\bm\mu, \bm y)] - \mathbb{E}_\pi[\log p_\infty(\bm y)].
\end{align}
It then follows that for any initial state $\bm y_{t_0}, \bm\mu_{t_0}$, $$\lim_{t\rightarrow\infty}\frac{1}{t}\log\frac{p_{t_0}(\bm y_{t_0}, \cdots, \bm y_{t_0+t-1})}{p_\infty(\bm y_{t_0}, \cdots, \bm y_{t_0+t-1})} = \mathcal{K}$$ under $\mathbb{P}_{t_0}$ almost surely. Therefore, for any $\eta>0$,
\begin{flalign}\label{eq:stronglaw}
\lim_{t\rightarrow\infty}&\sup_{t_0\geq 1}\text{esssup}\mathbb{P}_{t_0}\Bigg\{\max_{k\leq t}\sum_{i=t_0}^{t_0+k-1}\log\frac{p_{t_0}(\bm y_i|\bm y_{t_0}, \cdots, \bm y_{i-1})}{p_\infty(\bm y_i)}\geq \mathcal{K}(1+\eta)t\big|\bm y_1, \cdots, \bm y_{t_0-1}\Bigg\} = 0.
\end{flalign}
Then \eqref{eq:waddup} follows from \cite[Theorem 1]{lai1998information}. The full proof can be found in Appendix \ref{sec:theorem1}.
\end{proof}

\begin{comment}
\begin{theorem}\label{theorem:special}
When $\bm\Sigma_1 = \bm\Sigma^*$, we have that $\lim_{t\rightarrow\infty}\frac{1}{t}\log p_1^*(\\\bm y_1,\cdots, \bm y_t) = \mathbb{E}_\pi[h(\bm\mu^*, \bm y)]$ almost surely, where $h(\bm\mu^*_t,\bm y_t) =\log p_1^*(\bm y_t|\bm y_1, \cdots, \bm y_{t-1})$ is a function of $\bm\mu^*_t, \bm y_t$.
\end{theorem}
\begin{proof}
The proof of Theorem \ref{theorem:special} can be found in Appendix \ref{sec:theorem1}
\end{proof}
\end{comment}

\begin{remark}
    Results in \cite{lai1998information} assumes $\lim_{t\rightarrow\infty}\frac{1}{t}\log\frac{p_{t_0}(\bm y_{t_0}, \cdots, \bm y_{t_0+t-1})}{p_\infty(\bm y_{t_0}, \cdots, \bm y_{t_0+t-1})}$ exists under $\mathbb{P}_{t_0}$ while in our results, we prove its existence and characterize its value.
\end{remark}

% of $\lim_{t\rightarrow\infty}\frac{1}{t}\log\frac{p_{t_0}(\bm y_{t_0}, \cdots, \bm y_{t_0+t})}{p_\infty(\bm y_{t_0}, \cdots, \bm y_{t_0+t})}$, which is $\mathcal{K}$.

\subsection{$q$-th Order AR Models}\label{sec:generalization_lower}
In this section, we show that our results for first-order AR models can be generalized to any $q$-th order AR models.

In the following theorem, we show that for any $q$-th order AR model, as in Section \ref{sec:waddorder1}, as $t \rightarrow \infty$, the limit of $\frac{1}{t}\log\frac{p_{t_0}(\bm y_{t_0},\cdots, \bm y_{t_0+t-1})}{p_\infty(\bm y_{t_0},\cdots, \bm y_{t_0+t-1})}$ exists and is positive, denoted by $\widetilde{\mathcal{K}}$. The expression of $\widetilde{\mathcal{K}}$ can be derived similarly as in \eqref{eq:def_K}. 
The universal lower bound on WADD then follows from \cite[Theorem 1]{lai1998information}.
\begin{theorem}\label{theorem:waddloworder2}
For a $q$-th order AR model, we have that 
\begin{flalign}
\lim_{t\rightarrow\infty}\frac{1}{t}\log\frac{p_{t_0}(\bm y_{t_0},\cdots, \bm y_{t_0+t-1})}{p_\infty(\bm y_{t_0},\cdots, \bm y_{t_0+t-1})} = \widetilde{\mathcal{K}} 
\end{flalign}
almost surely under $\mathbb{P}_{t_0}$ where $\widetilde{\mathcal{K}}>0$. Moreover,
as $\gamma\rightarrow\infty$, 
\begin{flalign}\label{eq:wadduporderq}
\inf_{\tau:\text{ARL}(\tau)\geq \gamma}\text{WADD}(\tau)\geq \frac{\log\gamma}{\widetilde{\mathcal{K}}}(1+o(1)).
\end{flalign}
\end{theorem}
The proof of Theorem \ref{theorem:waddloworder2} is similar to the proof of Theorem \ref{theorem:waddlow}, and the idea is to convert the $q$-th order AR model to a first-order AR model, and then apply the proof of Theorem \ref{theorem:waddlow}. Below, we  provide a proof sketch.

\begin{proof}[Proof Sketch]
Note that for a $q$-th order AR model, it can be converted to a first-order AR model (See Appendix \ref{app:convert}).  Let $\lfloor x\rfloor$ denote the greatest integer less than or equal to $x$. 
We have that 
\begin{flalign}\label{eq:qconverge}
&\lim_{t\rightarrow\infty}\frac{1}{t}\log\frac{p_{t_0}(\bm y_{t_0}, \cdots, \bm y_{t_0+t-1})}{p_\infty(\bm y_{t_0}, \cdots, \bm y_{t_0+t-1})} \nn\\
& =\lim_{t\rightarrow\infty}\frac{1}{t}\bigg(\log\frac{p_{t_0}(\bm y_{t_0}, \cdots, \bm y_{t_0+q-1})}{p_\infty(\bm y_{t_0}, \cdots, \bm y_{t_0+q-1})} + \log\frac{p_{t_0}(\bm y_{t_0+q}, \cdots, \bm y_{t_0+2q-1}|\bm y_{t_0}, \cdots, \bm y_{t_0+q-1})}{p_\infty(\bm y_{t_0+q}, \cdots, \bm y_{t_0+2q-1}|\bm y_{t_0}, \cdots, \bm y_{t_0+q-1})}\nn\\
&\hspace{0.5cm} +\cdots + \log\frac{p_{t_0}(\bm y_{t_0+q(\lfloor \frac{t}{q}\rfloor-1)}, \cdots, \bm y_{t_0+q\lfloor \frac{t}{q}\rfloor -1}|\bm y_{t_0}, \cdots, \bm y_{t_0+q(\lfloor \frac{t}{q}\rfloor-1)-1})}{p_\infty(\bm y_{t_0+q(\lfloor \frac{t}{q}\rfloor-1)}, \cdots, \bm y_{t_0+q\lfloor\frac{t}{q}\rfloor-1}|\bm y_{t_0}, \cdots, \bm y_{t_0+q(\lfloor \frac{t}{q}\rfloor-1)-1})}\bigg) 
\nn\\
&= 
% \lim_{t\rightarrow\infty}\frac{1}{t}\bigg(\log\frac{p_{t_0}(\tilde{\bm y}_{t_0}^{(q)})}{p_\infty(\tilde{\bm y}_{t_0}^{(q)})} + \log\frac{p_{t_0}(\tilde{\bm y}_{t_0+1}^{(q)}|\tilde{\bm y}_{t_0}^{(q)})}{p_\infty(\tilde{\bm y}_{t_0+1}^{(q)}|\tilde{\bm y}_{t_0}^{(q)})} + \cdots + \log\frac{p_{t_0}(\tilde{\bm y}_{t_0+(\lfloor \frac{t}{q}\rfloor-1)}^{(q)}|\tilde{\bm y}_{t_0}^{(q)}, \cdots, \tilde{\bm y}_{t_0+(\lfloor \frac{t}{q}\rfloor-2)}^{(q)})}{p_\infty(\tilde{\bm y}_{t_0+(\lfloor \frac{t}{q}\rfloor-1)}^{(q)}|\tilde{\bm y}_{t_0}^{(q)},\cdots, \tilde{\bm y}_{t_0+(\lfloor \frac{t}{q}\rfloor-2)}^{(q)})}\bigg)\nn\\&= 
% \lim_{t\rightarrow\infty}\frac{1}{t}\log \frac{p_{t_0}(\tilde{\bm y}_{t_0}^{(q)}, \cdots, \tilde{\bm y}_{t_0+(\lfloor \frac{t}{q}\rfloor-1)}^{(q)})}{p_\infty(\tilde{\bm y}_{t_0}^{(q)}, \cdots, \tilde{\bm y}_{t_0+(\lfloor \frac{t}{q}\rfloor-1)}^{(q)})}\nn\\&=
\lim_{t\rightarrow\infty}\frac{1}{t}\log \frac{p_{t_0}\Big((\bm y_{t_0}, \cdots, \bm y_{t_0+q-1}), (\bm y_{t_0+q}, \cdots, \bm y_{t_0+2q-1}),\cdots, (\bm y_{t_0+q(\lfloor \frac{t}{q}\rfloor-1)}, \cdots, \bm y_{t_0+q\lfloor \frac{t}{q}\rfloor -1})\Big)}{p_\infty\Big((\bm y_{t_0}, \cdots, \bm y_{t_0+q-1}), (\bm y_{t_0+q}, \cdots, \bm y_{t_0+2q-1}),\cdots, (\bm y_{t_0+q(\lfloor \frac{t}{q}\rfloor-1)}, \cdots, \bm y_{t_0+q\lfloor \frac{t}{q}\rfloor -1})\Big)}\nn\\&=\lim_{t^\prime\rightarrow\infty}\frac{1}{qt^\prime}\log \frac{p_{t_0}\Big((\bm y_{t_0}, \cdots, \bm y_{t_0+q-1}), (\bm y_{t_0+q}, \cdots, \bm y_{t_0+2q-1}),\cdots, (\bm y_{t_0+q(t^\prime-1)}, \cdots, \bm y_{t_0+qt^\prime -1})\Big)}{p_\infty\Big((\bm y_{t_0}, \cdots, \bm y_{t_0+q-1}), (\bm y_{t_0+q}, \cdots, \bm y_{t_0+2q-1}),\cdots, (\bm y_{t_0+q(t^\prime-1)}, \cdots, \bm y_{t_0+qt^\prime -1})\Big)}\nn\\&=\widetilde{\mathcal{K}}
\end{flalign}
almost surely under $\mathbb{P}_{t_0}$, where the last equality is from \eqref{eq:convergek} and the fact that $\{(\bm y_{t_0+qi}, \cdots, \bm y_{t_0+q(i+1)-1})\}_{i=0}^\infty$ follows a first-order AR model. For the quickest change detection in $q$-th order AR models, we then have that as $\gamma\rightarrow\infty$,
\begin{flalign}
\inf_{\tau:\text{ARL}(\tau)\geq \gamma}\text{WADD}(\tau)\geq \frac{\log\gamma}{\widetilde{\mathcal{K}}}(1+o(1)).
\end{flalign}
\end{proof}

\section{Asymptotically Optimal Stopping Time}\label{sec:algorithm}
In this section, we first present the algorithm and its optimality results for the first-order AR model. We then show the generalization to the $q$-th order AR model.
\subsection{First-Order AR Model}
The CuSum algorithm based on the generalized likelihood ratio (GLR) approach has been widely used for QCD problems. For the AR model, the GLR statistic is defined as follows:
\begin{flalign}\label{eq:wt}
W_t = \max_{1\leq k\leq t}\sum_{i=k}^t \log \frac{p_k(\bm y_i|\bm y_{1}, \cdots, \bm y_{i-1})}{p_\infty(\bm y_i|\bm y_{1}, \cdots, \bm y_{i-1})}.
\end{flalign}
The CuSum algorithm \cite{lai1998information} can then be designed:
\begin{flalign}\label{eq:glrcusum}
\tau_c = \inf\Big\{t: W_t \geq c\Big\}.
\end{flalign}

In \cite{lai1998information}, a special non-i.i.d. case was studied where the post-change distribution does not depend on the change point $t_0$, i.e., $p_k(\bm y_i|\bm y_{1}, \cdots, \bm y_{i-1})$ in \eqref{eq:wt} does not depend on $k$.
Under the assumption that $\frac{1}{t}\sum_{i=t_0}^{t_0+t-1} \log \frac{p_{t_0}(\bm y_i|\bm y_{1}, \cdots, \bm y_{i-1})}{p_\infty(\bm y_i|\bm y_{1}, \cdots, \bm y_{i-1})}$ converges in probability under $\mathbb{P}_{t_0}$ to some positive constant $\mathcal{K}$, the asymptotic optimality of the CuSum algorithm was proved. The non-i.i.d. case under the Bayesian setting was studied in \cite{tartakovsky2005asymptotic}, where the post-change distribution does not depend on the change point $t_0$ as in \cite{lai1998information}. The asymptotic optimality of Shiryaev procedure was established under the same stability assumption.

If the post-change distribution does not depend on the change-point, then the CuSum algorithm can be updated recursively and is computationally efficient. 
% that $\frac{1}{t}\sum_{i=t_0}^{t_0+t} \log \frac{p_1(\bm y_i|\bm y_{1}, \cdots, \bm y_{i-1})}{p_\infty(\bm y_i|\bm y_{1}, \cdots, \bm y_{i-1})}$ converges almost surely under $\mathbb{P}_{t_0}$
However, for our AR model, under $\mathbb{P}_{t_0}$, $p_{t_0}(\bm y_t|\bm y_{t_0}, \cdots, \bm y_{t-1})$ depends on $t_0$ for $t\geq t_0$. 
% Note that the maximization over $k$ in $W_t$ does not admit a recursive update for dependent observations unless we have 
% \begin{flalign}
% p_{t_0}(\bm y_t|\bm y_{1}, \cdots, \bm y_{t-1}) = p_1(\bm y_t|\bm y_1, \cdots, \bm y_{t-1})
% \end{flalign} 
% for $t_0 = 1,2,\cdots$ and $t\geq t_0$, that is, the post-change distribution doest not depend on the change point. 
At each time $t$, we need to update $p_k(\bm y_t|\bm y_{1}, \cdots, \bm y_{t-1})$ for every $1\leq k\leq t$, the complexity of which scales with $t$, which is not practically feasible. %Therefore, the CuSum algorithm can not be directly applied to our problem. 

In this section, we propose a computationally efficient Ergodic CuSum algorithm and further show that it is asymptotically optimal.

%Secondly, the asymptotic optimality of the CuSum algorithm was established in \cite{lai1998information} under the assumption that $\frac{1}{t}\sum_{i=t_0}^{t_0+t} \log \frac{p^{t_0}(\bm y_i|\bm y_{1}, \cdots, \bm y_{i-1})}{p_0(\bm y_i|\bm y_{1}, \cdots, \bm y_{i-1})}$ converges in probability under $\mathbb{P}_{t_0}$ to $\mathcal{K}$. For the AR model, this assumption need to be verified. Therefore, it is unknown if the CuSum algorithm is asymptotically optimal for our problem in \eqref{eq:goal}.

%In this paper, we aim to design a computationally efficient Ergodic CuSum algorithm which is also asymptotically optimal to detect the change of AR model in \eqref{eq:goal}. The major challenge lies in that the existing approaches for proving the existence of $\mathcal{K}$ for i.i.d cases, Markov models and discrete-state HMMs are not applicable for our problem. Therefore, we develop new approaches to show that the constant $\mathcal{K}$ exists for the AR model. We further propose a computationally efficient Ergodic CuSum algorithm and show that it is asymptotically optimal.

Motivated by the fact that the WADD is lower bounded by $\frac{\log\gamma}{\mathcal{K}}(1 + o(1))$, we aim to find a computationally efficient test statistic that has a positive drift of $\mathcal{K}$ under the post-change distribution and has a negative drift under the pre-change distribution. Define the likelihood ratio of the first $t$ observations when the change point $t_0 = 1$ by 
\begin{flalign}\label{eq:recursivel}
L_t = \frac{p_1(\bm y_1, \cdots, \bm y_t)}{p_\infty(\bm y_1, \cdots, \bm y_t)},
\end{flalign}
and let $L_0=1$. 
For any $t_0\geq 1$, we have that
\begin{flalign}\label{eq:pesudo}
&\lim_{t\rightarrow\infty}\frac{1}{t}\log \frac{p_1(\bm y_1, \cdots, \bm y_t)}{p_\infty(\bm y_1, \cdots, \bm y_t)} \nn\\& = \lim_{t\rightarrow\infty}\frac{1}{t}\bigg(\log\frac{p_1(\bm y_1, \cdots, \bm y_{t_0-1})}{p_\infty(\bm y_1, \cdots, \bm y_{t_0-1})} + \log \frac{p_1(\bm y_{t_0}, \cdots, \bm y_t|\bm y_1, \cdots, \bm y_{t_0-1})}{p_\infty(\bm y_{t_0}, \cdots, \bm y_t|\bm y_1, \cdots, \bm y_{t_0-1})}\bigg)\nn\\& = \lim_{t\rightarrow\infty}\frac{1}{t} \log \frac{p_1(\bm y_{t_0}, \cdots, \bm y_t|\bm y_1, \cdots, \bm y_{t_0-1})}{p_\infty(\bm y_{t_0}, \cdots, \bm y_t|\bm y_1, \cdots, \bm y_{t_0-1})}.
\end{flalign}
Note that different sample trajectories $\bm y_1, \cdots, \bm y_{t_0-1}$ lead to different values of $\bm \mu_{t_0}$. However, the ergodic theorem of Markov chain \cite{roberts2004general, meyn2012markov} implies that the convergence of $\frac{1}{t} \log \frac{p_1(\bm y_{t_0}, \cdots, \bm y_t|\bm y_1, \cdots, \bm y_{t_0-1})}{p_\infty(\bm y_{t_0}, \cdots, \bm y_t|\bm y_1, \cdots, \bm y_{t_0-1})}$ does not depend on $\bm \mu_{t_0}$ and thus does not depend on the sample trajectory $\bm y_1, \cdots, \bm y_{t_0-1}$. Therefore, 
we have that for any $\bm y_1, \cdots, \bm y_{t_0-1}$, under $\mathbb{P}_{t_0}$ almost surely
\begin{flalign}
    &\lim_{t\rightarrow\infty}\frac{1}{t} \log \frac{p_1(\bm y_{t_0}, \cdots, \bm y_t|\bm y_1, \cdots, \bm y_{t_0-1})}{p_\infty(\bm y_{t_0}, \cdots, \bm y_t|\bm y_1, \cdots, \bm y_{t_0-1})} \nn\\&= \lim_{t\rightarrow\infty}\frac{1}{t}\log\frac{p_{t_0}(\bm y_{t_0}, \cdots, \bm y_{t})}{p_\infty(\bm y_{t_0}, \cdots, \bm y_{t})}\nn\\& = \mathcal{K},
\end{flalign}
where the first equality is due to the fact that the convergence of the ergodic Markov chain doesn't depend on the initial state \cite{roberts2004general, meyn2012markov}. Therefore, under the post-change distribution, $L_t$ has a positive drift $\mathcal{K}$ on average. Under the pre-change distribution, we have that $\mathbb{E}_\infty\big[\log\frac{p_1(\bm y_t|\bm y_1, \cdots, \bm y_{t-1})}{p_\infty(\bm y_t)}\big] = -D\big(p_\infty(\bm y_t)||p_1(\bm y_t|\bm y_1, \cdots, \bm y_{t-1})\big) \leq 0$ where the equality holds when $p_1(\bm y_t|\bm y_1, \cdots, \bm y_{t-1}) = p_\infty(\bm y_t)$ almost surely. Therefore, under the pre-change distribution, $L_t$ has a negative drift as long as $p_1(\bm y_t|\bm y_1, \cdots, \bm y_{t-1}) \neq p_\infty(\bm y_t)$.

Motivated by these facts, we define the Ergodic CuSum statistic
\begin{flalign}\label{eq:ergodics}
S_t&= \max_{0\leq i\leq t}(\log L_t - \log L_i)\nn\\
% &= \log L_t - \min_{0\leq i\leq t}\log L_i\nn\\
& =  \max\big(0, S_{t-1} + \log L_t - \log L_{t-1}\big).
\end{flalign}
The Ergodic CuSum algorithm is then defined as 
\begin{flalign}\label{eq:cusumtype}
\tau_c^* = \inf\{t: S_t \geq c\}.
\end{flalign}

At each time $t$, we only need to compute $L_t$. It can be easily verified using \eqref{eq:additive} that $L_t$ admits an additive form. Moreover, $L_t$ can be recursively updated using the forward variable as shown in Lemma \ref{lemma:forward}. Specifically, let $\bm \mu_0, \bm \Sigma_0$ be the parameters of the initial distribution of $\bm x_0$. We have that $\bm \mu_t, \bm \Sigma_t$ can be updated recursively according to \eqref{eq:updaterule} for each $t$. We then have that $L_t =\sum_{i=1}^t\frac{c_i}{c_{i-1}} = L_{t-1} + \frac{c_t}{c_{t-1}}$, where $\frac{c_t}{c_{t-1}}$ can be computed using $\bm \mu_t, \bm \Sigma_t$ as shown in \eqref{eq:forwardsigma}.

We note that $S_t$ is not the actual generalized likelihood ratio, and is different from $W_t$ in $\tau_c$. However, at each time $t$, $\log\frac{p_1(\bm y_t|\bm y_1, \cdots, \bm y_{t-1})}{p_\infty(\bm y_t)}$ can still be viewed as a log likelihood ratio of the current sample $\bm y_t$. Therefore, the ARL lower bound can be derived following the proof of the ARL lower bound in \cite{lai1998information} for general non-i.i.d. case.
Since the convergence of $\lim_{t\rightarrow\infty}\frac{1}{t} \log \frac{p_1(\bm y_{t_0}, \cdots, \bm y_t|\bm y_1, \cdots, \bm y_{t_0-1})}{p_\infty(\bm y_{t_0}, \cdots, \bm y_t|\bm y_1, \cdots, \bm y_{t_0-1})}$ under $\mathbb{P}_{t_0}$ does not depend on the initial state $\bm \mu_{t_0}, \bm y_{t_0}$, it can be shown that $\tau_c^*$ is asymptotically optimal for \eqref{eq:goal}.% even when 
% \eqref{eq:pseudo} doest not hold.

In the following theorem, we show 1) the ARL lower bound of $\tau_c^*$ and 2) the WADD upper bound of $\tau_c^*$.
\begin{theorem}\label{theorem:waddup}
1) Let $c = \log\gamma$ in \eqref{eq:cusumtype}, then $\mathbb{E}_\infty[\tau_c^*] \geq \gamma$; and 2)  as $\gamma \rightarrow \infty$, $\text{WADD}(\tau_c^*)\leq \frac{\log\gamma}{\mathcal{K}}(1+o(1))$.
\end{theorem}
\begin{proof}
The proof can be found in Appendix \ref{sec:theorem4}.
\end{proof}

%The difficulty for deriving the WADD upper bound of $\tau_c$ is that we need to show the convergence of $\lim_{t\rightarrow\infty} \frac{1}{t} \log \\\frac{p_1(\bm y_1,\cdots, \bm y_t)}{p_0(\bm y_1,\cdots, \bm y_t)}$, which is proved in Lemma \ref{lemma:stonglaw}. With Lemma \ref{lemma:stonglaw}, we can apply the standard techniques in \cite{lai1998information} to obtain the WADD upper bound of $\tau_c$.

Based on Theorem \ref{theorem:waddlow} and Theorem \ref{theorem:waddup}, we establish the asymptotic optimality of $\tau_c^*$ in the following theorem.
\begin{theorem}
$\tau_c^*$ is asymptotically optimal.
\end{theorem}
\begin{proof}
By Theorem \ref{theorem:waddlow} and Theorem \ref{theorem:waddup}, we establish the asymptotic optimality of $\tau_c^*$.
\end{proof}
The Ergodic CuSum algorithm in \eqref{eq:cusumtype} is computationally efficient and asymptotically optimal for detecting the change in the AR model in Section \ref{sec:problemmodel}.

\subsection{ $q$-th Order AR Model}
We first convert the post-change $q$-th order AR model equivalently into a first-order AR model  (see Appendix \ref{app:convert}). We then partition the sequence of the observations into a sequence of non-overlapping blocks with size $q$. Specially, 
define $\tilde y_t=(y_{(t-1)q+1},\ldots,y_{tq})$, for $t=1,2,\ldots$. 
We then apply our Ergodic Cusum algorithm  in \eqref{eq:cusumtype} on the sequence of $\{\tilde y_t\}_{t=1}^\infty$.
We show that our Ergodic Cusum algorithm is asymptotically optimal for the problem of QCD in $q$-th order AR models. 

For the ARL lower bound and WADD upper bound of $\tau^*_c$, we have the following theorem.
\begin{theorem}\label{theorem:wadduporder2}
For the QCD problem in $q$-th order AR models, consider $\tau^*_c$ applied on $\{\tilde y_t\}_{t=1}^\infty$.  1) Let $c = \log\gamma$ in \eqref{eq:cusumtype}, then $\mathbb{E}_\infty[\tau_c^*] \geq \gamma$. 2)  As $\gamma \rightarrow \infty$, $\text{WADD}(\tau_c^*)\leq \frac{\log\gamma}{\widetilde{\mathcal{K}}}(1+o(1))$.
\end{theorem}

The proof of Theorem \ref{theorem:wadduporder2} is similar to the proof of Theorem \ref{theorem:waddup}. Here, we only provide a proof sketch.

\begin{proof}[Proof Sketch]
When the post-change disturbance signal follows a $q$-th order AR model, we have that for any $\bm y_1, \cdots, \bm y_{t_0-1}$, under $\mathbb{P}_{t_0}$ almost surely
\begin{flalign}
&\lim_{t\rightarrow\infty}\frac{1}{t}\log \frac{p_1(\bm y_1, \cdots, \bm y_t)}{p_\infty(\bm y_1, \cdots, \bm y_t)} \nn\\& = \lim_{t\rightarrow\infty}\frac{1}{t}\bigg(\log\frac{p_1(\bm y_1, \cdots, \bm y_{t_0-1})}{p_\infty(\bm y_1, \cdots, \bm y_{t_0-1})} + \log \frac{p_1(\bm y_{t_0}, \cdots, \bm y_t|\bm y_1, \cdots, \bm y_{t_0-1})}{p_\infty(\bm y_{t_0}, \cdots, \bm y_t|\bm y_1, \cdots, \bm y_{t_0-1})}\bigg)\nn\\
& = \lim_{t\rightarrow\infty}\frac{1}{t} \log \frac{p_1(\bm y_{t_0}, \cdots, \bm y_t|\bm y_1, \cdots, \bm y_{t_0-1})}{p_\infty(\bm y_{t_0}, \cdots, \bm y_t|\bm y_1, \cdots, \bm y_{t_0-1})}\nn\\& = \lim_{t\rightarrow\infty}\frac{1}{t} \log \frac{p_{t_0}(\bm y_{t_0}, \cdots, \bm y_t)}{p_\infty(\bm y_{t_0}, \cdots, \bm y_t)}\nn\\
&=\lim_{t\rightarrow\infty}\frac{1}{t}\log\frac{p_{t_0}\Big((\bm y_{t_0}, \cdots, \bm y_{t_0+q-1}), (\bm y_{t_0+q}, \cdots, \bm y_{t_0+2q-1}),\cdots, (\bm y_{t_0+q(\lfloor \frac{t}{q}\rfloor-1)}, \cdots, \bm y_{t_0+q\lfloor \frac{t}{q}\rfloor -1})\Big)}{p_\infty\Big((\bm y_{t_0}, \cdots, \bm y_{t_0+q-1}), (\bm y_{t_0+q}, \cdots, \bm y_{t_0+2q-1}),\cdots, (\bm y_{t_0+q(\lfloor \frac{t}{q}\rfloor-1)}, \cdots, \bm y_{t_0+q\lfloor \frac{t}{q}\rfloor -1})\Big)}\nn\\
&= \widetilde{\mathcal{K}}.
\end{flalign}
where the third equality is due to the fact that the convergence of $\frac{1}{t} \log \frac{p_1(\bm y_{t_0}, \cdots, \bm y_t|\bm y_1, \cdots, \bm y_{t_0-1})}{p_\infty(\bm y_{t_0}, \cdots, \bm y_t|\bm y_1, \cdots, \bm y_{t_0-1})}$ does not depend on the sample trajectory $\bm y_1, \cdots, \bm y_{t_0-1}$ and the last equality is from \eqref{eq:qconverge} and the fact that $\{(\bm y_{t_0+qi}, \cdots, \bm y_{t_0+q(i+1)-1})\}_{i=0}^\infty$ follows a first-order AR model.
Following the same steps as in Theorem \ref{theorem:waddup}, we have that when $c = \log\gamma$, $\mathbb{E}_\infty[\tau_c^*] \geq \gamma$ and $\text{WADD}(\tau_c^*)\leq \frac{\log\gamma}{\widetilde{\mathcal{K}}}(1+o(1))$ as $\gamma \rightarrow \infty$.
\end{proof}

It can be seen that the WADD upper bound of $\tau_c^*$ matches with the universal lower bound in Theorem \ref{theorem:waddloworder2} for the $q$-th order AR models. We then have the following result of asymptotic optimality.
\begin{theorem}
For the problem of QCD in $q$-th order AR models, $\tau_c^*$ (applied on its first-order equivalence) is asymptotically optimal.
\end{theorem}
\begin{proof}
By Theorem \ref{theorem:waddloworder2} and Theorem \ref{theorem:wadduporder2}, we establish the asymptotic optimality of $\tau_c^*$.
\end{proof}

\section{Data-Driven Setting: Online Gradient Ascent CuSum}\label{sec:datadriven}

In this section, we consider the practical data-driven setting where the post-change parameters are unknown. This is motivated by practical applications such as detecting cyber-attacks in dynamic systems, where there is no prior knowledge about the disturbance signal $\bm x_t$. As shown in Sections \ref{sec:waddlower} and \ref{sec:algorithm}, the QCD problem in a $q$-th order AR model can be equivalently solved using a first-order AR model. Therefore, in this section, we focus on the first-order AR model. Specifically, the post-change parameters $\bm A, \bm R_\omega$ are unknown.

The generalized likelihood ratio test (GLRT) which replaces the unknown post-change parameter with its maximum likelihood estimate (MLE) is widely used when there are unknown parameters. Specifically, at each time $t$, hypothesizing on the change point being $k$, we use samples $\bm y_k, \cdots, \bm y_t$ to compute the MLE of the unknown post-change parameters, i.e., $\hat{\theta}_{k, t} = \arg\max_\theta p_{k, \theta}(\bm y_k, \cdots, \bm y_t)$, where $p_{k, \theta}(\bm y_k, \cdots, \bm y_t)$ denotes the likelihood of $\bm y_k, \cdots, \bm y_t$ when the change point is $k$ and post-change parameter is $\theta$. The MLE of the hypothesized change point is then computed. Specifically, let $S_{\hat{\theta}_{k, t}}$ be the test statistic with post-change parameters  $\hat{\theta}_{k, t}$, and take the maximum of $S_{\hat{\theta}_{k, t}}$  over all $k$, i.e, $S_t = \max_k S_{\hat{\theta}_{k, t}}$. This approach needs to store all the historical samples $\bm y_1, \cdots,\bm y_t$ and recompute $\hat{\theta}_{k, t}$ when there is a new sample, which is not efficient for memory and computation. Though in practice, a sliding-window approach can be used, but it is still hard in the non-i.i.d. setting \cite{lai1998information}. 

In this section, we apply an online gradient ascent algorithm \cite{raginsky2012sequential, cao2018sequential} to estimate the unknown parameters $\bm A, \bm R_\omega$ and plug them back to the Ergodic CuSum algorithm to design our online gradient ascent CuSum algorithm (OGA-CuSum). 
The OGA-CuSum can be updated using only the most recent sample and in a recursive way, and thus is much more memory and computationally efficient.

Recall the definition of $h(\bm\mu, \bm y)$ in \eqref{eq:hmuy}. Define $h_t(\bm\mu, \bm y)$ by replacing $\bm \Sigma^*$ with $\bm \Sigma_t$ in $h(\bm\mu, \bm y)$. From Appendix \ref{sec:theorem1}, letting $\bm y = \bm y_t$ and $\bm \mu = \bm\mu_t$, we have that $h_t(\bm\mu_t,\bm y_t) = p_{t_0}(\bm y_t|\bm y_{t_0}, \cdots, \bm y_{t-1})$. Given the initial $\bm \mu_0, \bm\Sigma_0$, it then follows that 
\begin{flalign}\label{eq:updatel}
    \log L_t = \sum_{i=1}^t h_i(\bm\mu_i, \bm y_i) - \log p_\infty(\bm y_i).
\end{flalign}
Since $\bm A, \bm R_\omega$ are unknown, $h_t(\bm\mu_t, \bm y_t)$ cannot be computed. We then propose a one-step update rule to efficiently estimate $\bm A$, $\bm R_\omega$ and further estimate $h_t(\bm\mu_t, \bm y_t)$ at each time $t$. 

Denote by $\widehat{\bm A}_t, \widehat{\bm R}_{\omega, t}$ the estimate of $\bm A$, $\bm R_\omega$ at time $t$. Denote by $\widehat{\bm\mu}_{t+1}, \widehat{\bm\Sigma}_{t+1}$ the parameters of the forward variable. From \eqref{eq:updaterule}, we have that $\widehat{\bm\mu}_{t+1}, \widehat{\bm\Sigma}_{t+1}$ can be updated recursively using $\widehat{\bm A}_t, \widehat{\bm R}_{\omega, t}$, i.e.,
\begin{flalign}\label{eq:onlineupdate}
&\widehat{\bm\Sigma}_{t+1} = (\widehat{\bm A}_t\widehat{\bm\Sigma}_{t}\widehat{\bm A}_t^\top + \widehat{\bm R}_{\omega, t})(\widehat{\bm A}_t\widehat{\bm\Sigma}_{t}\widehat{\bm A}_t^\top + \widehat{\bm R}_{\omega, t} +\bm I)^{-1},\nn\\&
\widehat{\bm\mu}_{t+1} = (\widehat{\bm A}_t\widehat{\bm\Sigma}_{t}\widehat{\bm A}_t^\top + \widehat{\bm R}_{\omega, t} +\bm I)^{-1}\widehat{\bm A}_t\widehat{\bm\mu}_{t} + (\widehat{\bm A}_t\widehat{\bm\Sigma}_{t}\widehat{\bm A}_t^\top + \widehat{\bm R}_{\omega, t})(\widehat{\bm A}_t\widehat{\bm\Sigma}_{t}\widehat{\bm A}_t^\top + \widehat{\bm R}_{\omega, t} +\bm I)^{-1}\bm y_{t+1}.
\end{flalign}
Define $\widehat{h}_{t+1}(\widehat{\bm\mu}_{t+1}, \bm y_{t+1})$ by replacing $\bm A, \bm R_\omega, \bm\mu_{t+1}, \bm\Sigma_{t+1}$ in $h_{t+1}(\bm\mu_{t+1}, \bm y_{t+1})$ with $\widehat{\bm A}_{t}, \widehat{\bm R}_{\omega, t}, \widehat{\bm\mu}_{t+1}, \widehat{\bm\Sigma}_{t+1}$. We iteratively update the estimate of the parameters when there is a new observation based on the maximum likelihood principle using gradient ascent. We note that $h_t(\bm\mu_t, \bm y_t) - \log p_\infty(\bm y_t)$ is the log-likelihood ratio of the observation $\bm y_t$ at time $t$, i.e., $h_t(\bm\mu_t, \bm y_t) - \log p_\infty(\bm y_t) = \frac{p_{1}(\bm y_t|\bm y_{1}, \cdots, \bm y_{t-1})}{p_\infty(\bm y_t)}$. Therefore, based on the maximum likelihood principle, the online gradient ascent estimator is updated as follows
\begin{flalign}
&\widehat{\bm A}_t = \widehat{\bm A}_{t-1} + \beta \nabla_{\bm A} \widehat{h}_{t}(\widehat{\bm\mu}_{t}, \bm y_{t}),\nn\\& \widehat{\bm R}_{\omega, t} = \text{Proj}\big(\widehat{\bm R}_{\omega, t-1} + \beta \nabla_{\bm R_\omega} \widehat{h}_{t}(\widehat{\bm\mu}_{t}, \bm y_{t})\big),
\end{flalign}
where $\beta$ is a pre-specified step-size and Proj$(\bm X)$ is an operator that projects a matrix $\bm X$ to the set of positive definite matrices to guarantee that $\widehat{\bm R}_{\omega, t}$ is positive definite. Define the eigenvalue decomposition $\bm X = \sum_{i=1}^n \lambda_iv_iv_i^\top$, where $\lambda_i$ is the eigenvalue of $\bm X$ and $v_i$ is its corresponding eigenvector. We define Proj$(\bm X) = \sum_{i=1}^n \max\{\lambda_i, \epsilon\}v_iv_i^\top$ \cite{boyd2004convex}, where $\epsilon>0$ is chosen to guarantee that the eigenvalue of Proj$(\bm X)$ is at least $\epsilon$. With $\widehat{\bm A}_t, \widehat{\bm R}_{\omega, t}$, we can update $\widehat{\bm\mu}_{t+1}, \widehat{\bm\Sigma}_{t+1}$ according to \eqref{eq:onlineupdate}. We further plug them back to \eqref{eq:ergodics} and \eqref{eq:updatel} to compute $\widehat{S}_{t+1}$, which serves as the estimate of $S_{t+1}$. When $\widehat{S}_{t+1} < 0$, we claim that the change hasn't occurred. Therefore, we ignore the previous samples and reset our parameters $\widehat{\bm A}_{t+1}, \widehat{\bm R}_{\omega, t+1}$. The OGA-CuSum algorithm is then defined as follows
\begin{flalign}
    \tau_{\text{OGA}} = \inf \{t: \widehat{S}_t \geq c\}.
\end{flalign}
We summarize the algorithm in Algorithm \ref{alg:ogdcusum}.

In general, it is hard to obtain theoretical optimality performance guarantees for data-driven approaches. Nevertheless, we provide a lower bound on the ARL for our OGA-CuSum algorithm,  so that a threshold can be chosen analytically to control the false alarm in
practice. We also provide simulation results to demonstrate the good performance of our algorithm.

In the following theorem, we present a lower bound on the ARL for our OGA-CuSum algorithm. 
\begin{theorem}\label{theorem:ogdarl}
Let $c = \log\gamma$ in Algorithm \ref{alg:ogdcusum}, then $\mathbb{E}_\infty[\tau_{\text{OGA}}] \geq \gamma$.
\end{theorem}
\begin{proof}
    The proof can be found in Appendix \ref{sec:proofarl}.
\end{proof}

\begin{algorithm}
\caption{Online gradient ascent CuSum algorithm}\label{alg:ogdcusum}
\begin{algorithmic}
\Require Initial parameters $c, \beta, \widehat{\bm \mu}_0, \widehat{\bm \Sigma}_0, \widehat{\bm A}_0, \widehat{\bm R}_{\omega, 0}$, $t \leftarrow 0, \widehat{L}_0 \leftarrow 1, \widehat{S}_0 \leftarrow 0$

\While{$\widehat{S}_t < c$}

\State Take a new observation $\bm y_{t+1}$ 
\State Update $\widehat{\bm \mu}_{t+1}, \widehat{\bm \Sigma}_{t+1}$ according to \eqref{eq:onlineupdate}
\State $\log\widehat{L}_{t+1} \leftarrow \log \widehat{L}_{t} + \widehat{h}_{t+1}(\widehat{\bm\mu}_{t+1}, \bm y_{t+1}) - \log p_\infty(\bm y_{t+1})$
\State $\widehat{S}_{t+1} \leftarrow  \widehat{S}_{t} + \log\widehat{L}_{t+1} - \log\widehat{L}_{t}$
\If {$\widehat{S}_{t+1}<0$}
    \State $\widehat{\bm A}_{t+1} \leftarrow \widehat{\bm A}_0$ 
    \State $\widehat{\bm R}_{\omega, t+1} \leftarrow \widehat{\bm R}_{\omega, 0}$
    \State $\widehat{S}_{t+1}\leftarrow 0$
    \State $\widehat{L}_{t+1} \leftarrow 1$
\Else
\State $\widehat{\bm A}_{t+1} = \widehat{\bm A}_{t} + \beta \nabla_{\bm A} \widehat{h}_{t+1}(\widehat{\bm\mu}_{t+1}, \bm y_{t+1})$
\State $\widehat{\bm R}_{\omega, t+1} = \text{Proj}\big(\widehat{\bm R}_{\omega, t} + \beta \nabla_{\bm R_\omega} \widehat{h}_{t+1}(\widehat{\bm\mu}_{t+1}, \bm y_{t+1})\big)$
\EndIf
\State $t\leftarrow\ t+1$
\EndWhile
\end{algorithmic}
\end{algorithm}

\section{Simulation Results}\label{sec:numerical}
In this section, we provide some numerical results to demonstrate the performance of our Ergodic CuSum algorithm and OGA-CuSum algorithm.

\subsection{Model-Based Setting}\label{sec:model-base}
In this section, we consider the setting where the parameters of the disturbance signal are known. We compare the Ergodic CuSum algorithm $\tau_c^*$ with a stationary CuSum algorithm designed using the stationary distribution of $\bm y$. The stationary CuSum algorithm uses the CuSum algorithm that detects a change from $p_\infty$ to the stationary distribution $\pi$:
\begin{flalign}
\tau_s = \inf\Big\{t: \max_{1\leq k\leq t}\sum_{i=k}^t \log \frac{\pi(\bm y_i)}{p_\infty(\bm y_i)}\geq c\Big\}.
\end{flalign}

We consider three different cases. For the first case, we set $\bm A = [[0.7, 0.4], [0.2, 0.6]]$ and $\bm R_\omega = [[1, 0.5], [0.5, 1]]$. For the second case, we let $\bm A$ be a $10\times 10$ matrix and $\bm R_\omega = \bm I$. For the third case, we consider a $q$-th order AR model. We set $\bm A_1 = [[0.4, 0.3], [0.2, 0.1]], \bm A_2 = [[0.3, 0.2], [0.1, 0.2]]$ and $\bm R_\omega = \bm I$. Note that our Ergodic CuSum algorithm can be easily generalized to the $q$-th order AR model by reformulating the $q$-th order AR model to a first-order AR model.
In Fig.~\ref{fig:1}, Fig.~\ref{fig:2} and Fig.~\ref{fig:high}, we plot the WADD as a function of ARL. It can be seen from Fig.~\ref{fig:1}, Fig.~\ref{fig:2} and Fig.~\ref{fig:high} that with the same constraint on the ARL, our Ergodic CuSum algorithm has a lower detection delay. Therefore, our Ergodic CuSum algorithm outperforms the stationary CuSum algorithm under both the low-dimensional setting and the high-dimensional setting. This further implies that for detecting changes in the AR model, simply applying the stationary distribution and ignoring the dependence among samples may not perform well. Moreover, the relationship between WADD and log of the WARL is linear, which validates our theoretical analysis. 

\begin{figure*}[htbp]
\begin{multicols}{3}
\centering 
\includegraphics[width=\linewidth]{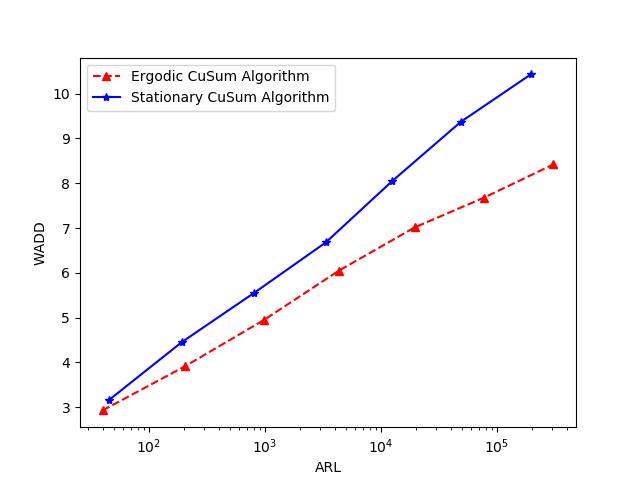}\par\caption{Comparison of the two algorithms: case 1.}\label{fig:1}
\includegraphics[width=\linewidth]{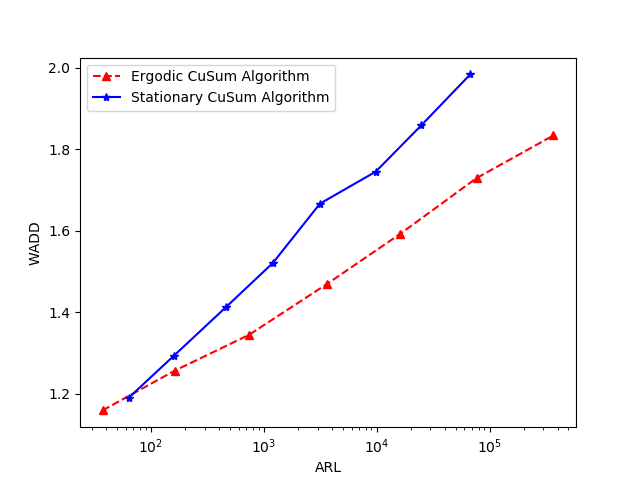}\par\caption{Comparison of the two algorithms: case 2.}\label{fig:2}
\includegraphics[width=\linewidth]{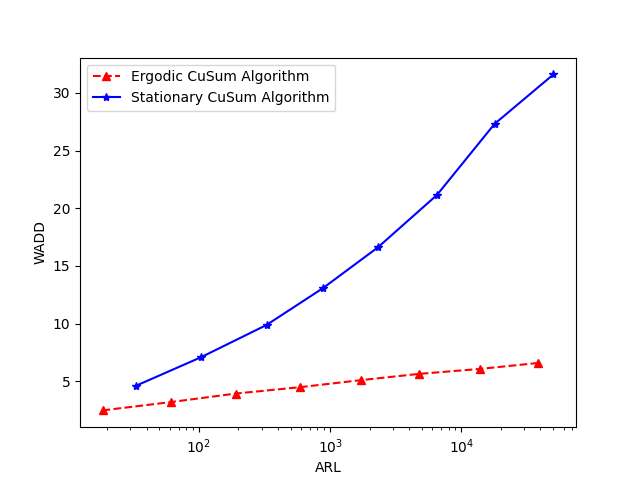}\par\caption{Comparison of the two algorithms: case 3.}\label{fig:high}
\end{multicols}
\end{figure*}

\subsection{Data-Driven Setting}
In this section, we focus on the data-driven setting. We first plot the estimation errors of $\widehat{\bm A}_t$, $\widehat{\bm R}_{\omega, t}$ as functions of number of samples under the pre- and post-change distribution, where the estimation errors are defined as the Frobenius norm of $\widehat{\bm A}_t-\bm A$ and $\widehat{\bm R}_{\omega, t}-\bm R_\omega$. Under the post-change distribution, we set $\bm A = [[0.7, 0.4], [0.2, 0.6]]$ and $\bm R_\omega = [[1, 0.5], [0.5, 1]]$. We note that the pre-change distribution is equivalent to the AR model with $\bm A = \bm 0$ and $\bm R_\omega =\bm 0$. It can be seen from Fig.~\ref{fig:3}, Fig.~\ref{fig:4}, Fig.~\ref{fig:5} and Fig.~\ref{fig:6} that the estimated parameters $\widehat{\bm A}_t$, $\widehat{\bm R}_{\omega, t}$ converge to the true parameters under the pre- and post-change distribution, which demonstrates that our OGA-CuSum scheme provides reliable estimates for the true parameters of the disturbance signal. 

\begin{figure}[htbp]
\begin{multicols}{2}
 \centering 
\includegraphics[width=\linewidth]{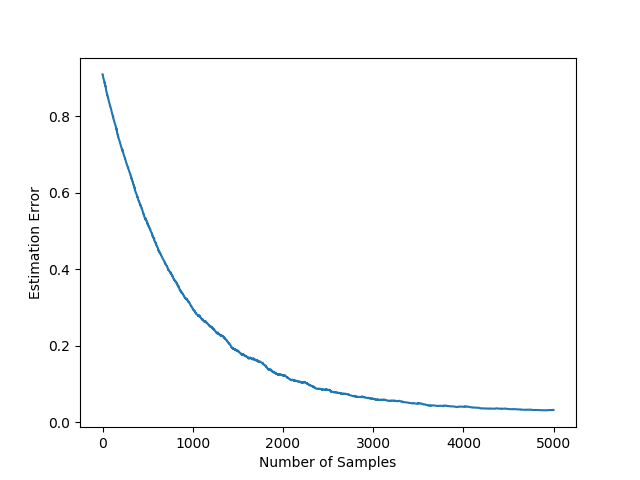}\par\caption{Convergence of $\widehat{\bm A}_t$ under the pre-change distribution.}\label{fig:3}   
\includegraphics[width=\linewidth]{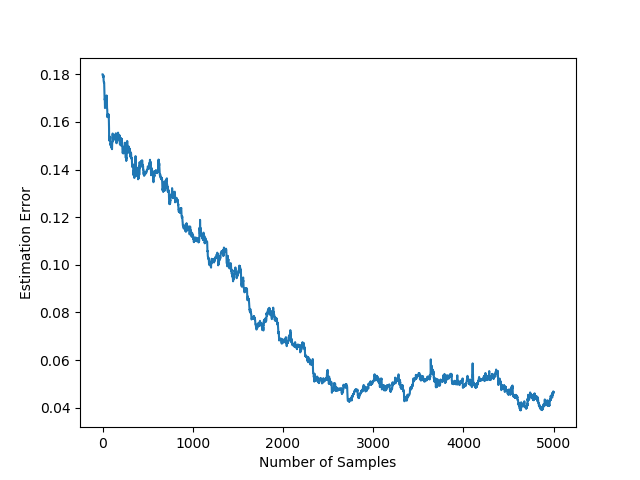}\par\caption{Convergence of $\widehat{\bm A}_t$ under the post-change distribution.}\label{fig:4}
\end{multicols}
\end{figure}

\begin{figure}[htbp]
\begin{multicols}{2}
\centering 
\includegraphics[width=\linewidth]{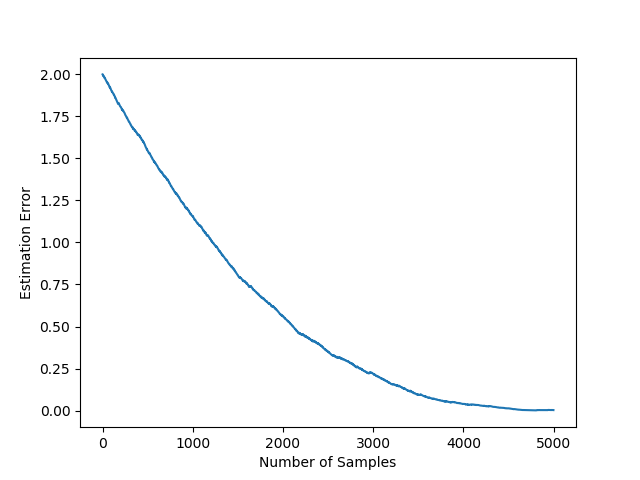}\par\caption{Convergence of $\widehat{\bm R}_{\omega,t}$ under the pre-change distribution.}\label{fig:5}    
\includegraphics[width=\linewidth]{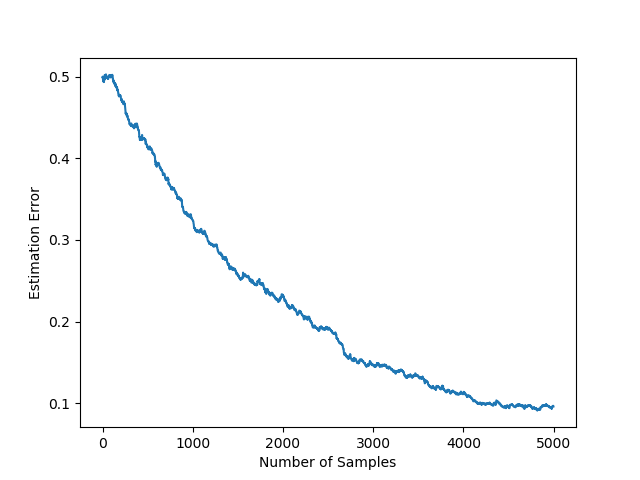}\par\caption{Convergence of $\widehat{\bm R}_{\omega, t}$ under the post-change distribution}\label{fig:6}
\end{multicols}
\end{figure}

We then compare our OGA-CuSum algorithm with an 
existing approach in \cite{chen2018quickest}. In \cite{chen2018quickest}, a GLRT based data-driven algorithm was proposed to detect disturbance signals in the AR model. We also plot the curve of the Ergodic CuSum algorithm under the model-based setting, which serves as a lower bound for the data-driven setting. We use the same parameters as in Section \ref{sec:model-base} and plot the WADD as a function of ARL. From Fig.~\ref{fig:7}, Fig.~\ref{fig:8} and Fig.~\ref{fig:datahigh}, it can be seen that our OGA-CuSum algorithm outperforms the GLRT based algorithm in \cite{chen2018quickest}. Moreover, the performance of the OGA-CuSum algorithm is close to the performance of the Ergodic CuSum algorithm. Therefore, our OGA-CuSum algorithm has a small performance loss compared with the asymptotically optimal algorithm under the model-based setting.

\begin{figure}[htbp]
\begin{multicols}{3}
\centering 
\includegraphics[width=\linewidth]{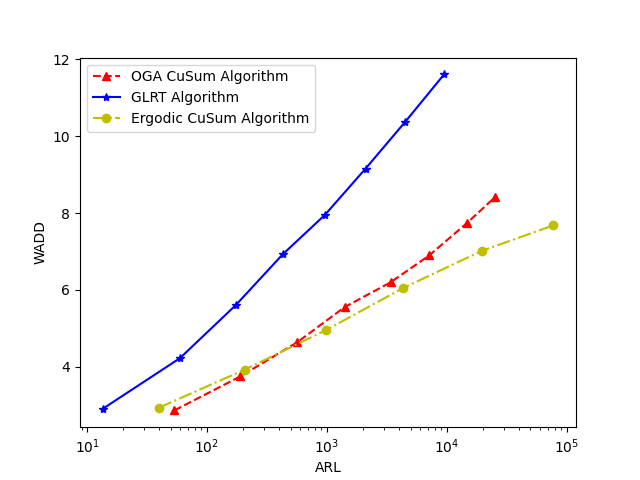}\par\caption{Comparison of the three algorithms: case 1.}\label{fig:7}    
\includegraphics[width=\linewidth]{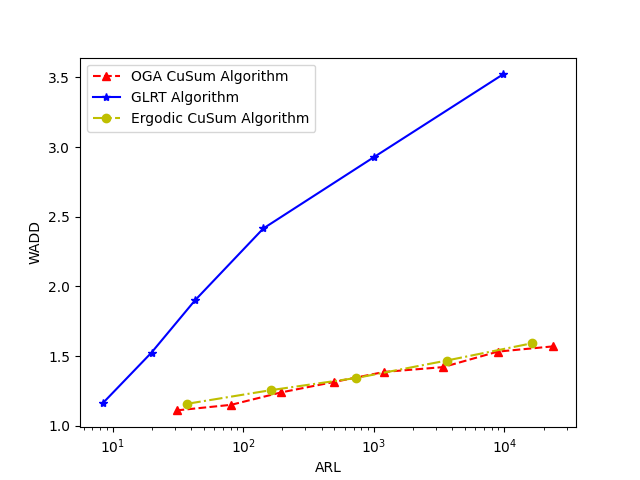}\par\caption{Comparison of the three algorithms: case 2.}\label{fig:8}
\includegraphics[width=\linewidth]{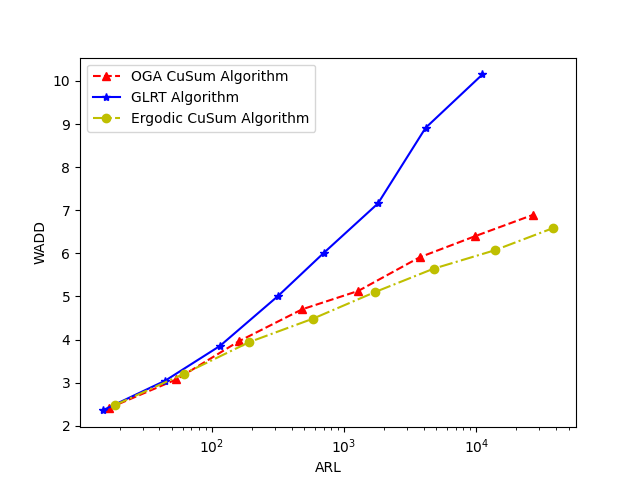}\par\caption{Comparison of the three algorithms: case 3.}\label{fig:datahigh}
\end{multicols}
\end{figure}

\section{Conclusion}\label{sec:conclusion}
In this paper, we studied the QCD problem in AR models under the model-based setting and the data-driven setting. For the model-based setting, we proposed a novel Ergodic CuSum algorithm.
By introducing the forward variable for general state HMMs, our Ergodic CuSum algorithm can be implemented efficiently. We further showed that our algorithm is asymptotically optimal under the Lorden's criterion \cite{lorden1971procedures} based on the ergodic theorem \cite{roberts2004general,meyn2012markov}. For the data-driven setting, we proposed an OGA-CuSum algorithm, which replaces the unknown parameters in the Ergodic CuSum algorithm with their estimates based on the online convex optimization. We provided simulation results to demonstrate the performance of our algorithm. Our approaches provide useful insights for general detection problems for AR models and general state HMMs.

% In this paper, the innovation noise and the measurement noise are Gaussian distributed. In the future, it is of interest to consider the case where the innovation noise and the measurement noise follow general distributions, for which the approach developed in this paper can not be applied directly. 

\newpage

\appendices

\section{Convert $q$-th order AR Model to First-order AR Model}\label{app:convert}
Consider the case with $q=2$. Let $\tilde{\bm x}_{t}^{(2)} = [\bm x_{2t-1}^\top, \bm x_{2t}^\top]^\top$. From \eqref{eq:highorder}, we have that 
\begin{flalign}
    \tilde{\bm x}_{t}^{(2)} = \tilde{\bm A}\tilde{\bm x}_{t-1}^{(2)} + \tilde{\bm\omega}_{t}^{(2)},
\end{flalign}
where $\tilde{\bm A} = \begin{pmatrix}
\bm A_2 & \bm A_1\\
\bm A_1\bm A_2 & \bm A_1\bm A_1 + \bm A_2
\end{pmatrix}$ and $\tilde{\bm \omega}_{t}^{(2)} = \begin{pmatrix}
\bm\omega_{2t-1}\\
\bm A_1\bm\omega_{2t-1} + \bm \omega_{2t}
\end{pmatrix}$ is Gaussian distributed. Therefore, $\tilde{\bm x}_{t}^{(2)}$ follows a first-order AR model. Let $\tilde{\bm y}_{t}^{(2)} = [\bm y_{2t-1}^\top, \bm y_{2t}^\top]^\top$. We have that $\tilde{\bm y}_{t}^{(2)} = \tilde{\bm x}_{t}^{(2)} + \tilde{\bm\nu}_{t}^{(2)}$, where $\tilde{\bm\nu}_{t}^{(2)} = \begin{pmatrix}
\bm\nu_{2t-1}\\
\bm \nu_{2t}
\end{pmatrix}$. 

Following the same steps, for a $q$-th order AR model, let $\tilde{\bm x}_{t}^{(q)} = [\bm x_{q(t-1)+1}^\top, \bm x_{q(t-1)+2}^\top, \cdots, \bm x_{qt}^\top]^\top$ and $\tilde{\bm y}_{t}^{(q)} = [\bm y_{q(t-1)+1}^\top, \bm y_{q(t-1)+2}^\top, \cdots, \bm y_{qt}^\top]^\top$. We have that
\begin{flalign}\label{eq:convert}
    \tilde{\bm x}_{t}^{(q)} &= \tilde{\bm A}\tilde{\bm x}_{(t-1)}^{(q)} + \tilde{\bm \omega}_{t}^{(q)},\nn\\\tilde{\bm y}_{t}^{(q)} &= \tilde{\bm x}_{t}^{(q)} + \tilde{\bm\nu}_{t}^{(q)},
\end{flalign}
where $\tilde{\bm \omega}_{t}^{(q)}$ and $\tilde{\bm\nu}_{t}^{(q)}$ are the innovative noise and measurement noise respectively for the $q$-th order AR model. We denote the covariance of $\tilde{\bm \omega}_{t}^{(q)}$ by ${\widetilde{\bm R}}_{\bm \omega}$.
Here, we omit the expression of $\widetilde{\bm A}$ since it is cumbersome but it can be computed directly. It can be easily verified that the covariance matrix of $\tilde{\bm\nu}_{t}^{(q)}$ is $\bm I$.
Therefore,  $\{\tilde{\bm x}_{t}^{(q)}, \tilde{\bm y}_{t}^{(q)}\}_{t=1}^\infty$ is a first-order AR model.

\section{Proof of Lemma \ref{lemma:sigma}}\label{sec:lemma4}
\begin{proof}
To show that $\bm\Sigma_t$ converges, it suffices to show that for any $\epsilon>0$, there exists an integer $T$ such that for any $s, t>T$, $\|\bm\Sigma_s - \bm\Sigma_t\| < \epsilon$. Without loss of generality, we assume that $s>t$.

We first note that
\begin{flalign}\label{eq:updatesigma}
\bm\Sigma_t &= \big((\bm A\bm\Sigma_{t-1}\bm A^\top + \bm R_{\omega})^{-1}+\bm I\big)^{-1},\nn\\
\bm\Sigma_{s} &= \big((\bm A\bm\Sigma_{s-1}\bm A^\top + \bm R_{\omega})^{-1}+\bm I\big)^{-1}.
\end{flalign}
We then have that
\begin{flalign}
&\bm\Sigma_s - \bm \Sigma_t \nn\\&= \bm\Sigma_t(\bm\Sigma_t^{-1}-\bm\Sigma_s^{-1})\bm\Sigma_s\nn\\ &= \big((\bm A\bm\Sigma_{t-1}\bm A^\top + \bm R_{\omega})^{-1}+\bm I\big)^{-1}\big((\bm A\bm\Sigma_{t-1}\bm A^\top + \bm R_{\omega})^{-1} - (\bm A\bm\Sigma_{s-1}\bm A^\top + \bm R_{\omega})^{-1}\big)\nn\\&\hspace{0.5cm}\big((\bm A\bm\Sigma_{s-1}\bm A^\top + \bm R_{\omega})^{-1}+\bm I\big)^{-1} \nn\\&= \big((\bm A\bm\Sigma_{t-1}\bm A^\top + \bm R_{\omega})^{-1}+\bm I\big)^{-1}(\bm A\bm\Sigma_{t-1}\bm A^\top + \bm R_{\omega})^{-1}\big((\bm A\bm\Sigma_{s-1}\bm A^\top + \bm R_{\omega})-(\bm A\bm\Sigma_{t-1}\bm A^\top + \bm R_{\omega})\big)\nn\\&\hspace{0.5cm}(\bm A\bm\Sigma_{s-1}\bm A^\top + \bm R_{\omega})^{-1}\big((\bm A\bm\Sigma_{s-1}\bm A^\top + \bm R_{\omega})^{-1}+\bm I\big)^{-1}\nn\\&= (\bm I + \bm A\bm\Sigma_{t-1}\bm A^\top + \bm R_{\omega})^{-1}\bm A(\bm\Sigma_{s-1}-\bm\Sigma_{t-1})\bm A^\top(\bm I + \bm A\bm\Sigma_{s-1}\bm A^\top + \bm R_{\omega})^{-1}.
\end{flalign}

From the update rule of $\bm \Sigma_t$ in \eqref{eq:updatesigma}, $\bm \Sigma_t$ is positive definite for every $t\geq t_0$. Therefore, $\bm A\bm\Sigma_{t-1}\bm A^\top + \bm R_{\omega}$ is positive definite. It then allows the eigendecomposition and can be factorized as $\bm A\bm\Sigma_{t-1}\bm A^\top + \bm R_{\omega} = \bm Q\bm\Lambda \bm Q^{-1}$, where $\bm Q$ is the square $K\times K$ matrix whose $i$th column is the eigenvector $q_i$ of $\bm A$ and $\bm \Lambda$ is the diagonal matrix whose $i$th diagonal elements are the corresponding eigenvalues $\lambda_i$. We then have that
\begin{flalign}
(\bm I + \bm A\bm\Sigma_{t-1}\bm A^\top + \bm R_{\omega})^{-1}& = \big(\bm Q(\bm\Lambda +\bm I)\bm Q^{-1}\big)^{-1} = \bm Q(\bm\Lambda +\bm I)^{-1}\bm Q^{-1}.
\end{flalign}
Since $\bm A\bm\Sigma_{t-1}\bm A^\top + \bm R_{\omega}$ is positive definite for any $t>0$, we have that $\lambda_i>0$ for $i=1,\cdots, K$. Therefore, all the diagonal elements of $\bm\Lambda +\bm I$ are strictly larger than 1. Therefore, the diagonal elements of $(\bm\Lambda +\bm I)^{-1}$ are strictly less than 1. From the definition of operator norm, we have that there exists a $0<\delta<1$ such that $\|(\bm I + \bm A\bm\Sigma_{t-1}\bm A^\top + \bm R_{\omega})^{-1}\| < (1-\delta)$ for any $t>0$. 

It then follows that 
\begin{flalign}\label{eq:converge}
&\|\bm\Sigma_s - \bm \Sigma_t\|\nn\\& = \big\|(\bm I + \bm A\bm\Sigma_{t-1}\bm A^\top + \bm R_{\omega})^{-1}\bm A(\bm\Sigma_{s-1}-\bm\Sigma_{t-1})\bm A^\top (\bm I + \bm A\bm\Sigma_{s-1}\bm A^\top + \bm R_{\omega})^{-1}\big\|\nn\\&\leq \|(\bm I + \bm A\bm\Sigma_{t-1}\bm A^\top + \bm R_{\omega})^{-1}\| \|\bm A\|\|\bm\Sigma_{s-1}-\bm\Sigma_{t-1}\|\|\bm A^\top\|\|(\bm I + \bm A\bm\Sigma_{s-1}\bm A^\top + \bm R_{\omega})^{-1}\|\nn\\& <(1-\delta)^2 \|\bm\Sigma_{s-1}-\bm\Sigma_{t-1}\|,
\end{flalign}
where the first inequality is from the submultiplicative of operator norm and the second inequality is due to the facts that $\|(\bm I + \bm A\bm\Sigma_{t-1}\bm A^\top + \bm R_{\omega})^{-1}\| < 1-\delta, \|(\bm I + \bm A\bm\Sigma^*\bm A^\top + \bm R_{\omega})^{-1}\| < 1-\delta$ and the assumption $\|\bm A\| <1$. By applying \eqref{eq:converge} recursively, we have that
\begin{flalign}
\|\bm\Sigma_s - \bm \Sigma_t\| < (1-\delta)^{2(t-t_0)} \|\bm\Sigma_{s-t_0+t}-\bm\Sigma_{t_0}\|.
\end{flalign}
Moreover, from the update rule of $\bm\Sigma_t$, it can be easily verified that $\|\bm\Sigma_{s-t_0+t}-\bm\Sigma_{1}\|$ is bounded. Therefore, for any $\epsilon>0$, there exists an integer $T$ such that for any $s, t>T$, $\|\bm\Sigma_s - \bm\Sigma_t\| < \epsilon$. This completes the proof.
\end{proof}

\section{Proof of Lemma \ref{lemma:station}}\label{sec:lemma1}
\begin{proof}
Let $\pi(\bm x_t)$ be a Gaussian distribution with mean $\bm \mu$ and covariance matrix $\bm \Sigma$. We will show that $\pi(\bm x_t)$ is a stationary distribution of $\{\bm x_t\}_{t=t_0}^{\infty}$
if $\bm\mu = \bm 0$ and $\bm \Sigma = \sum_{i=0}^\infty (\bm A^\top)^i\bm R_{\omega}\bm A^{i}$. It can be easily proved that $\bm \Sigma$ exists since $\|\bm A\|<1$, where $\|\cdot\|$ denotes the operator norm of a matrix.

From the definition of stationary distribution, $\pi(\bm x_t)$ should satisfy that 
\begin{flalign}\label{eq:25}
\int \pi(\bm x_{t-1}) f(\bm x_t|\bm x_{t-1})d\bm x_{t-1} = \pi(\bm x_t).
\end{flalign}
We have that 
\begin{flalign}
&\int \pi(\bm x_{t-1}) f(\bm x_t|\bm x_{t-1})d\bm x_{t-1} \nn\\&
=\int \pi(\bm x_{t-1}) \frac{1}{\sqrt{(2\pi)^K \det(\bm R_\omega)}}\exp\Big(-\frac{1}{2}(\bm x_t - \bm A \bm x_{t-1})^\top \bm R_\omega^{-1}(\bm x_t - \bm A \bm x_{t-1})\Big)d\bm x_{t-1} \nn\\&= \int \pi(\bm x_{t-1}) \frac{1}{\sqrt{(2\pi)^K \det(\bm R_\omega)}}\exp\Big(-\frac{1}{2}(\bm A^{-1}\bm x_t - \bm x_{t-1})^\top \bm A^\top\bm R_\omega^{-1}\bm A(\bm A^{-1} \bm x_t - \bm x_{t-1})\Big)d\bm x_{t-1} \nn\\&= \int \pi(\bm x_{t-1}) \frac{\sqrt{\det((\bm A^\top\bm R_\omega^{-1} \bm A)^{-1})}}{\sqrt{\det(\bm R_\omega)}\sqrt{(2\pi)^K \det((\bm A^\top\bm R_\omega^{-1} \bm A))^{-1}}}\exp\nn\\&\hspace{0.2cm}\Big(-\frac{1}{2}(\bm A^{-1}\bm x_t - \bm x_{t-1})^\top \bm A^\top\bm R_\omega^{-1}\bm A(\bm A^{-1} \bm x_t - \bm x_{t-1})\Big)d\bm x_{t-1}\nn\\&\overset{(a)}{=} \frac{\sqrt{\det((\bm A^\top\bm R_\omega^{-1} \bm A)^{-1})}}{\sqrt{\det(\bm R_\omega)}\sqrt{(2\pi)^K \det((\bm A^\top\bm R_\omega^{-1} \bm A)^{-1}+\bm\Sigma)}}\exp\nn\\&\hspace{0.2cm} \Big(-\frac{1}{2}(\bm A^{-1}\bm x_t - \bm\mu)^\top \big((\bm A^\top\bm R_\omega^{-1} \bm A)^{-1} + \bm\Sigma\big)^{-1}(\bm A^{-1}\bm x_t - \bm\mu)\Big)\nn\\& = \frac{\sqrt{\det((\bm A^\top\bm R_\omega^{-1} \bm A)^{-1})}\sqrt{\det\big(\bm A ((\bm A^\top\bm R_\omega^{-1} \bm A)^{-1}+\bm\Sigma)\bm A^\top\big)}}{\sqrt{\det(\bm R_\omega)}\sqrt{ \det\big((\bm A^\top\bm R_\omega^{-1} \bm A)^{-1}+\bm\Sigma\big)}}\nn\\&\hspace{0.2cm} \frac{1}{\sqrt{(2\pi)^K\det\big(\bm A ((\bm A^\top\bm R_\omega^{-1} \bm A)^{-1}+\bm\Sigma)\bm A^\top\big)}}\exp\Big(-\frac{1}{2}\nn\\&\hspace{0.2cm} (\bm x_t - \bm A \bm\mu)^\top \big(\bm A ((\bm A^\top\bm R_\omega^{-1} \bm A)^{-1}+\bm\Sigma)\bm A^\top\big)^{-1}( \bm x_t - \bm A \bm\mu)\Big),
\end{flalign}
where $(a)$ is from the convolution of two Gaussian functions. 

Let $\bm\Sigma = \sum_{i=0}^\infty (\bm A^\top)^i\bm R_{\omega}\bm A^{i}$. We have that $\bm\Sigma = \bm R_\omega + \bm A\bm\Sigma\bm A^\top$.
If we choose $\bm \mu = 0$ and $\bm\Sigma = \sum_{i=0}^\infty (\bm A^\top)^i\bm R_{\omega}\bm A^{i}$, it can then be easily verified that \eqref{eq:25} holds.
% Therefore, we choose $\bm \mu$ and $\bm\Sigma$ satisfying that $\bm \mu = \bm 0, \bm\Sigma = \bm R_\omega + \bm A\bm\Sigma\bm A^T$. Moreover, the coefficient should satisfy
% \begin{flalign}
% &\frac{\sqrt{\det((\bm A^T\bm R_\omega^{-1} \bm A)^{-1})}}{\sqrt{\det(\bm R_\omega)}}\frac{\sqrt{\det\big(\bm A ((\bm A^T\bm R_\omega^{-1} \bm A)^{-1}+\bm\Sigma)\bm A^T\big)}}{\sqrt{\det\big((\bm A^T\bm R_\omega^{-1} \bm A)^{-1}+\bm\Sigma\big)}}\nn\\&= 1.
% \end{flalign}
% We have that if $\bm\Sigma = \bm R_\omega + \bm A\bm\Sigma\bm A^T$
% \begin{flalign}
% &\frac{\sqrt{\det((\bm A^T\bm R_\omega^{-1} \bm A)^{-1})}}{\sqrt{\det(\bm R_\omega)}}\frac{\sqrt{\det\big(\bm A ((\bm A^T\bm R_\omega^{-1} \bm A)^{-1}+\bm\Sigma)\bm A^T\big)}}{\sqrt{\det\big((\bm A^T\bm R_\omega^{-1} \bm A)^{-1}+\bm\Sigma\big)}} \nn\\& = \frac{\sqrt{\det((\bm A^T\bm R_\omega^{-1} \bm A)^{-1})}}{\sqrt{\det(\bm R_\omega)}}\frac{\sqrt{\det(\bm\Sigma)}}{\sqrt{\det\big((\bm A^T\bm R_\omega^{-1} \bm A)^{-1}+\bm\Sigma\big)}}\nn\\& = \frac{\sqrt{\det((\bm A^T\bm R_\omega^{-1} \bm A)^{-1})}}{\sqrt{\det(\bm R_\omega)}}\frac{\sqrt{\det(\bm\Sigma)}}{\sqrt{\det\big(\bm A^{-1}\bm\Sigma(\bm A^T)^{-1}\big)}}\nn\\&\overset{(a)}{=} \frac{\sqrt{\det(\bm R_\omega)\det(\bm\Sigma)\det(\bm A^{-1})\det((\bm A^T)^{-1})}}{\sqrt{\det(\bm R_\omega)\det(\bm\Sigma)\det(\bm A^{-1})\det((\bm A^T)^{-1})}}\nn\\& = 1,
% \end{flalign}
% where $(a)$ is due to the fact that $\det(\bm A\bm B) = \det(\bm A)\det(\bm B)$ for two matrices $\bm A, \bm B$. 
Therefore, $\pi(\bm x_t)$ is a stationary distribution of $\{\bm x_t\}_{t=t_0}^{\infty}$. It then follows that $\{\bm x_t, \bm y_t\}_{t=t_0}^{\infty}$ has a stationary distribution $\pi(\bm x_t) g(\bm y_t|\bm x_t)$.

Let $P^1\big(\{\bm x, \bm y\}, E\big)$ denote the probability of reaching a measurable set $E$ from state $\{\bm x, \bm y\}$ in one step.
We have that for any $\{\bm x, \bm y\}$ and $E\in\mathbb{R}^{2K}$ such that $\pi(E)>0$,
\begin{flalign}
P^1\big(\{\bm x, \bm y\}, E\big) &= \int_E f\big(\bm x^\prime|\bm x\big)g(\bm y^\prime|\bm x^\prime)d\bm x^\prime\bm y^\prime %= \int_E \frac{1}{\sqrt{(2\pi)^K \det(\bm R_\omega)}}\exp\Big(-\frac{1}{2}(\bm x^\prime - \bm A \bm x)^T \bm R_\omega^{-1}(\bm x^\prime - \nn\\&\hspace{0.2cm}\bm A \bm x)\Big)\frac{1}{\sqrt{(2\pi)^K}}\exp\Big(-\frac{1}{2}(\bm y^\prime - \bm x^\prime)^T (\bm y^\prime - \bm x^\prime)\Big)d\bm x^\prime\bm y^\prime\nn\\& 
> 0,
\end{flalign}
where the inequality is due to the fact that Gaussian density functions are positive.
From the definition of irreducible \cite{meyn2012markov} Markov chain, we have that $\{\bm x_t, \bm y_t\}_{t=t_0}^{\infty}$ is $\pi$-irreducible.
\end{proof}

\section{Proof of Lemma \ref{lemma:ymumarkov}}\label{sec:lemma5}
\begin{proof}
To prove Lemma \ref{lemma:ymumarkov}, we will first show that 
$\{\bm x_t, \bm y_t, \bm \mu_t^*\}_{t=t_0}^{\infty}$ is a Markov chain. We will then show that the stationary distribution for $\bm \mu_t^*$ exists by finding the limiting distribution of $\bm \mu_t^*$. We then construct a new Markov chain using $\bm y_t, \bm \mu_t^*$ and show that its stationary distribution exists and is unique.

\textbf{Step 1.} We first show that $\{\bm x_t, \bm y_t, \bm \mu_t^*\}_{t=t_0}^{\infty}$ is a Markov chain.
From the update rule of $\bm\mu^*_t$, we have that for $t\geq t_0$,
\begin{flalign}
&p_{t_0}(\bm\mu_t^*|\bm x_{t_0}, \cdots, \bm x_{t}, \bm y_{t_0},\cdots, \bm y_{t}, \bm \mu_{t_0}^*, \cdots, \bm \mu_{t-1}^*) = p_{t_0}(\bm\mu_{t}^*|\bm \mu_{t-1}^*, \bm y_t).
\end{flalign}
We then have that
\begin{flalign}
&p_{t_0}(\bm x_t, \bm y_t, \bm\mu_t^*|\bm x_{t_0}, \cdots, \bm x_{t-1}, \bm y_{t_0},\cdots, \bm y_{t-1}, \bm \mu_{t_0}^*, \cdots, \bm \mu_{t-1}^*)\nn\\& =p_{t_0}(\bm x_t|\bm x_{t-1})p_{t_0}(\bm y_t|\bm x_t) p_{t_0}(\bm\mu_{t}^*|\bm \mu_{t-1}^*, \bm y_t)\nn\\& = p_{t_0}(\bm x_t, \bm y_t, \bm \mu_t^*|\bm x_{t-1}, \bm y_{t-1}, \bm \mu_{t-1}^*).
\end{flalign}
Therefore, $\{\bm x_t, \bm y_t, \bm \mu_t^*\}_{t={t_0}}^{\infty}$ is a Markov chain.

\textbf{Step 2.} We then show that the stationary distribution for $\bm \mu_t^*$ exists by finding the limiting distribution of $\bm \mu_t^*$.
Let the initial state be $\{\bm x_{t_0}, \bm y_{t_0}, \bm \mu_{t_0}^*\}$. From \eqref{eq:mustar}, we have that 
\begin{flalign}\label{eq:mut}
\bm\mu^*_t &= \big((\bm A\bm\Sigma^*\bm A^\top + \bm R_{\omega} +\bm I)^{-1}\bm A\big)^{t-{t_0}}\bm\mu_{t_0}^*  + \big((\bm A\bm\Sigma^*\bm A^\top + \bm R_{\omega} +\bm I)^{-1}\bm A\big)^{t-{t_0}-1}\bm\Sigma^*\bm y_{t_0+1} +\cdots\nn\\&\hspace{0.5cm} +(\bm A\bm\Sigma^*\bm A^\top + \bm R_{\omega} +\bm I)^{-1}\bm A\bm\Sigma^*\bm y_{t-1} + \bm\Sigma^*\bm y_t.
\end{flalign}

Note that $\bm\mu^*_t$ is the sum of Gaussian random variables. If $\bm y_{t_0}, \cdots, \bm y_t$ are jointly Gaussian distributed, then $\bm\mu^*_t$ is a Gaussian random variable. Since the Gaussian density function is continuous in its mean and covariance, it then suffices to show that the limiting mean and limiting covariance matrix of $\bm\mu^*_t$ exist and are independent of the initial state.

Given the initial state $\{\bm x_{t_0}, \bm y_{t_0}, \bm \mu_{t_0}^*\}$, we have that 
\begin{flalign}
&p_{t_0}(\bm y_{t_0+1}|\bm x_{t_0}, \bm y_{t_0}, \bm \mu_{t_0}^*)\nn\\ &= \int p_{t_0}(\bm x_{t_0+1}|\bm x_{t_0})p_{t_0}(\bm y_{t_0+1}|\bm x_{t_0+1})d\bm x_{t_0+1} \nn\\&= \frac{1}{\sqrt{(2\pi)^K \det(\bm R_\omega)}}\exp\Big(-\frac{1}{2}(\bm x_{t_0+1} - \bm A \bm x_{t_0})^\top \bm R_\omega^{-1}(\bm x_{t_0+1} - \bm A \bm x_{t_0})\Big)\nn\\&\hspace{0.2cm}\frac{1}{\sqrt{(2\pi)^K}}\exp(-\frac{1}{2}(\bm y_{t_0+1}-\bm x_{t_0+1})^\top(\bm y_{t_0+1}-\bm x_{t_0+1}))d\bm x_{t_0+1}\nn\\& = \frac{1}{\sqrt{(2\pi)^K \det(\bm R_\omega + \bm I)}}\exp\Big(-\frac{1}{2}(\bm y_{t_0+1} - \bm A \bm x_{t_0})^\top (\bm I + \bm R_\omega)^{-1}(\bm y_{t_0+1} - \bm A \bm x_{t_0})\Big).
\end{flalign}
Therefore, conditioning on $\{\bm x_{t_0}, \bm y_{t_0}, \bm \mu_{t_0}^*\}$, $\bm y_{t_0+1}$ is Gaussian distributed with mean $\bm A \bm x_{t_0}$ and covariance matrix $\bm I + \bm R_\omega$. Similarly, we can show that for any $t>t_0$,  conditioning on $\{\bm x_{t_0}, \bm y_{t_0}, \bm \mu_{t_0}^*\}$, $\bm y_t \sim \mathcal{N}(\bm A^{(t-{t_0})}\bm x_{t_0}, \bm I + \bm R_\omega + \bm A^\top\bm R_{\omega}\bm A +\cdots+(\bm A^\top)^{(t-t_0-1)}\bm R_{\omega}\bm A^{t-t_0-1})$. Moreover, since $\bm x_t$ is Gaussian distributed and $\bm x_{t+1}|\bm x_t\sim \mathcal{N}(\bm A\bm x_t, \bm R_{\omega})$, we have that $\bm x_t, \bm x_{t+1}$ are jointly Gaussian distributed. Similarly, since $\bm x_t, \bm x_{t+1}$ are jointly Gaussian and $\bm y_t|(\bm x_t, \bm x_{t+1})\sim \mathcal{N}(\bm x_t, \bm I)$, we have that $\bm x_t, \bm y_t, \bm x_{t+1}$ are jointly Gaussian distributed. Following the same idea, we can show that $\bm x_{t_0}, \bm x_{t_0+1},\cdots, \bm x_t, \bm y_{t_0}, \bm y_{t_0+1}, \cdots, \bm y_t$ are jointly Gaussian distributed and thus $\bm y_{t_0}, \bm y_{t_0+1}, \cdots, \bm y_t$ are jointly Gaussian distributed. Therefore, $\bm \mu_t^*$ follows a Gaussian distribution.

\begin{comment}
To find the limiting distribution of $\bm \mu_t^*$, we first consider the limiting expectation of $\bm\mu^*_t$, i.e., $\mathbb{E}_{t_0}[\lim_{t\rightarrow\infty} \bm\mu^*_t|\bm x_{t_0}, \bm y_{t_0}, \bm \mu_{t_0}^*]$. From the linearity of expectation \cite{mitzen2005probability}, when
\begin{flalign}\label{eq:linearity}
&\lim_{t\rightarrow\infty}\mathbb{E}_{t_0}\Big[\big\|\big((\bm A\bm\Sigma^*\bm A^T + \bm R_{\omega} +\bm I)^{-1}\bm A\big)^{t-{t_0}}\bm\mu_{t_0}^*\big\|\nn\\&\hspace{0.2cm} + \big\|\big((\bm A\bm\Sigma^*\bm A^T + \bm R_{\omega} +\bm I)^{-1}\bm A\big)^{t-{t_0}-1}\bm\Sigma^*\bm y_{t_0+1}\big\|\nn\\&\hspace{0.2cm}+ \cdots + \big\|(\bm A\bm\Sigma^*\bm A^T + \bm R_{\omega} +\bm I)^{-1}\bm A\bm\Sigma^*\bm y_{t-1}\big\|\nn\\&\hspace{0.2cm} + \big\|\bm\Sigma^*\bm y_t\big\|\big|\bm x_{t_0}, \bm y_{t_0}, \bm \mu_{t_0}^*\Big]\nn\\&<\infty,
\end{flalign} 
we have that $\mathbb{E}_{t_0}[\lim_{t\rightarrow\infty} \bm\mu^*_t|\bm x_{t_0}, \bm y_{t_0}, \bm \mu_{t_0}^*] = \lim_{t\rightarrow\infty}\mathbb{E}_{t_0}[\bm\mu^*_t|\bm x_{t_0}, \bm y_{t_0}, \bm \mu_{t_0}^*]$. 
\end{comment}
To find the limiting distribution of $\bm \mu_t^*$, we first consider the limiting expectation of $\bm\mu^*_t$. We have that
% We have that 
\begin{flalign}\label{eq:emut}
&\mathbb{E}_{t_0}[\bm\mu^*_t|\bm x_{t_0}, \bm y_{t_0}, \bm \mu_{t_0}^*]\nn\\
% &= \big((\bm A\bm\Sigma^*\bm A^T + \bm R_{\omega} +\bm I)^{-1}\bm A\big)^{t-{t_0}}\bm\mu_{t_0}^* + \big((\bm A\bm\Sigma^*\bm A^T \nn\\&\hspace{0.2cm} + \bm R_{\omega} +\bm I)^{-1}\bm A\big)^{t-{t_0}-1}\bm\Sigma^*\mathbb{E}_{t_0}[\bm y_{t_0+1}|\bm x_{t_0}, \bm y_{t_0}, \bm \mu_{t_0}^*] +\cdots\nn\\&\hspace{0.2cm} +(\bm A\bm\Sigma^*\bm A^T + \bm R_{\omega} +\bm I)^{-1}\bm A\bm\Sigma^*\mathbb{E}_{t_0}[\bm y_{t-1}|\bm x_{t_0}, \bm y_{t_0}, \bm \mu_{t_0}^*]\nn\\&\hspace{0.2cm} + \bm\Sigma^*\mathbb{E}_{t_0}[\bm y_t|\bm x_{t_0}, \bm y_{t_0}, \bm \mu_{t_0}^*]\nn\\
& = \big((\bm A\bm\Sigma^*\bm A^\top + \bm R_{\omega} +\bm I)^{-1}\bm A\big)^{t-{t_0}}\bm\mu_{t_0}^*  + \big((\bm A\bm\Sigma^*\bm A^\top + \bm R_{\omega} +\bm I)^{-1}\bm A\big)^{t-{t_0}-1}\bm\Sigma^*\bm A\bm x_{t_0} +\cdots\nn\\&\hspace{0.2cm} +(\bm A\bm\Sigma^*\bm A^\top + \bm R_{\omega} +\bm I)^{-1}\bm A\bm\Sigma^*\bm A^{t-{t_0}-1}\bm x_{t_0} + \bm\Sigma^*\bm A^{t-{t_0}}\bm x_{t_0}.\nn
\end{flalign}
Since there exists $\delta >0$ such that $\|(\bm A\bm\Sigma^*\bm A^\top + \bm R_{\omega} +\bm I)^{-1}\bm A\|<1-\delta$ and $\|\bm A\| < 1-\delta$, we have that 
\begin{flalign*}
&\big\|\mathbb{E}_{t_0}[\bm \mu^*_t|\bm x_{t_0}, \bm y_{t_0}, \bm \mu_{t_0}^*]\big\| \nn\\&\overset{(a)}{\leq} \big\|\big((\bm A\bm\Sigma^*\bm A^\top + \bm R_{\omega} +\bm I)^{-1}\bm A\big)^{t-{t_0}}\bm\mu_{t_0}^*\big\|  + \big\|\big((\bm A\bm\Sigma^*\bm A^\top + \bm R_{\omega} +\bm I)^{-1}\bm A\big)^{t-{t_0}-1}\bm\Sigma^*\bm A\bm x_{t_0}\big\| +\cdots\nn\\&\hspace{0.2cm} +\big\|(\bm A\bm\Sigma^*\bm A^\top + \bm R_{\omega} +\bm I)^{-1}\bm A\bm\Sigma^*\bm A^{t-{t_0}-1}\bm x_{t_0}\big\| + \big\|\bm\Sigma^*\bm A^{t-{t_0}}\bm x_{t_0}\big\|\nn\\&\overset{(b)}{\leq} \big\|(\bm A\bm\Sigma^*\bm A^\top + \bm R_{\omega} +\bm I)^{-1}\bm A\big\|^{t-{t_0}}\big\|\bm\mu_{t_0}^*\big\|  + \big\|(\bm A\bm\Sigma^*\bm A^\top + \bm R_{\omega} +\bm I)^{-1}\bm A\big\|^{t-{t_0}-1}\big\|\bm\Sigma^*\big\|\big\|\bm A\big\|\big\|\bm x_{t_0}\big\| +\cdots\nn\\&\hspace{0.2cm} +\big\|(\bm A\bm\Sigma^*\bm A^\top + \bm R_{\omega} +\bm I)^{-1}\bm A\big\|\big\|\bm\Sigma^*\big\|\big\|\bm A\big\|^{t-{t_0}-1}\big\|\bm x_{t_0}\big\| + \big\|\bm\Sigma^*\big\|\big\|\bm A\big\|^{t-{t_0}}\big\|\bm x_{t_0}\big\|\nn\\
&\leq (1-\delta)^{t-{t_0}}\big\|\bm\mu_{t_0}^*\big\| + (t-{t_0})(1-\delta)^{t-{t_0}}\|\bm\Sigma^*\|\|\bm x_{t_0}\|,
\end{flalign*}
where $(a)$ is from the triangle inequality of norms and $(b)$ is due to the submultiplicative of the operator norm.
We then have that $\lim_{t\rightarrow\infty}\big\|\mathbb{E}_{t_0}[\bm \mu^*_t|\bm x_{t_0}, \bm y_{t_0}, \bm \mu_{t_0}^*]\big\| =\bm 0$. 
%Similarly, we can show that \eqref{eq:linearity} holds. It then follows that $\mathbb{E}_{t_0}[\lim_{t\rightarrow\infty}\bm \mu^*_t|\bm x_{t_0}, \bm y_{t_0}, \bm \mu_{t_0}^*] = \lim_{t\rightarrow\infty}\mathbb{E}_{t_0}[\bm \mu^*_t|\bm x_{t_0}, \bm y_{t_0}, \bm \mu_{t_0}^*] = \bm0$. 
Therefore, we have that $\lim_{t\rightarrow\infty}\mathbb{E}_{t_0}[\bm \mu^*_t|\bm x_{t_0}, \bm y_{t_0}, \bm \mu_{t_0}^*]$ exists and is independent of the initial state $\{\bm x_{t_0}, \bm y_{t_0}, \bm \mu_{t_0}^*\}$.

We then consider the covariance matrix of $\bm\mu^*_t$. Given the initial state $\{\bm x_{t_0}, \bm y_{t_0}, \bm \mu_{t_0}^*\}$, let
\begin{flalign}
&\bm\Sigma_{\bm\mu_t^*} = \mathbb{E}_{t_0}\Big[\big(\bm \mu_t^{*} - \mathbb{E}_{t_0}[\bm \mu_t^{*}|\bm x_{t_0}, \bm y_{t_0}, \bm \mu_{t_0}^*]\big)\big(\bm\mu_t^*-\mathbb{E}_{t_0}[\bm \mu_t^{*}|\bm x_{t_0}, \bm y_{t_0}, \bm \mu_{t_0}^*]\big)^\top\big|\bm x_{t_0}, \bm y_{t_0}, \bm \mu_{t_0}^*\Big],\nn\\
&\bm\Sigma_{ij} = \mathbb{E}_{t_0}\Big[\big(\bm y_i-\mathbb{E}_{t_0}[\bm y_i|\bm x_{t_0}, \bm y_{t_0}, \bm \mu_{t_0}^*]\big)\big(\bm y_j-\mathbb{E}_{t_0}[\bm y_j|\bm x_{t_0}, \bm y_{t_0}, \bm \mu_{t_0}^*]\big)^\top\big|\bm x_{t_0}, \bm y_{t_0}, \bm \mu_{t_0}^*\Big].
\end{flalign}
From \eqref{eq:mut} and \eqref{eq:emut}, we have that 
\begin{flalign}\label{eq:mucovariance}
&\bm\Sigma_{\bm \mu_t^*} = \sum_{i={t_0}}^{t-1}\big((\bm A\bm\Sigma^*\bm A^\top + \bm R_{\omega} +\bm I)^{-1}\bm A\big)^{i-t_0}\bm\Sigma^{*}\bm\Sigma_{(t-i+t_0)(t-i+t_0)}\Big(\big((\bm A\bm\Sigma^*\bm A^\top + \bm R_{\omega} +\bm I)^{-1}\bm A\big)^{i-t_0}\bm\Sigma^{*}\Big)^\top\nn\\& + 2\sum_{i={t_0}}^{t-1}\sum_{j={t_0}}^{i-1}\big((\bm A\bm\Sigma^*\bm A^\top + \bm R_{\omega} +\bm I)^{-1}\bm A\big)^{i-t_0}\bm\Sigma^{*}\bm\Sigma_{(t-i+t_0)(t-j+t_0)}\Big(\big((\bm A\bm\Sigma^*\bm A^\top + \bm R_{\omega} +\bm I)^{-1}\bm A\big)^{j-t_0}\bm\Sigma^{*}\Big)^\top.
\end{flalign}
We first show that $\|\bm\Sigma_{ij}\|$ is bounded for any $i, j > t_0$. When $i = j$, we have that
\begin{flalign}
&\|\bm\Sigma_{ii}\| \nn\\&= \|\bm I + \bm R_\omega + \bm A^\top\bm R_{\omega}\bm A +\cdots+(\bm A^\top)^{(i-{t_0}-1)}\bm R_{\omega}\bm A^{i-{t_0}-1}\|\nn\\& \leq \|\bm I\| + \|\bm R_\omega\| + \|\bm A^\top\|\|\bm R_{\omega}\|\|\bm A\| +\cdots + \|\bm A^\top\|^{i-{t_0}-1}\|\bm R_{\omega}\|\|\bm A\|^{i-{t_0}-1} \nn\\&= \|\bm I\| + \big(1+\|\bm A\|^2+\cdots+\|\bm A\|^{2(i-{t_0})}\big)\|\bm R_{\omega}\| \nn\\&= \|\bm I\| + \frac{1-\|\bm A\|^{2(i-{t_0})}}{1-\|\bm A\|^2}\|\bm R_{\omega}\|\nn\\&\leq \|\bm I\| + \frac{1}{1-\|\bm A\|^2}\|\bm R_{\omega}\|,
\end{flalign}
where the last inequality is due to the fact that there exists $0<\delta<1$ such that $\|\bm A\| <1-\delta$. When $i\neq j$, we have that 
\begin{flalign}
\|\bm\Sigma_{ij}\| &= \Big\|\mathbb{E}_{t_0}\Big[\big(\bm y_i-\mathbb{E}_{t_0}[\bm y_i|\bm x_{t_0}, \bm y_{t_0}, \bm \mu_{t_0}^*]\big)\big(\bm y_j-\mathbb{E}_{t_0}[\bm y_j|\bm x_{t_0}, \bm y_{t_0}, \bm \mu_{t_0}^*]\big)^\top\big|\bm x_{t_0}, \bm y_{t_0}, \bm \mu_{t_0}^*\Big]\Big\|\nn\\&\leq \frac{1}{2}\big(\|\bm\Sigma_{ii}\| +  \|\bm\Sigma_{jj}\|\big)\nn\\&\leq \|\bm I\| + \frac{1}{1-\|\bm A\|^2}\|\bm R_{\omega}\|,
\end{flalign}
where the first inequality is due to the fact that for any two vectors $\bm u, \bm v \in \mathbb{R}^K$, $(\bm u-\bm v)(\bm u-\bm v)^\top$ is positive semi-definite, and thus $\bm u\bm u^\top + \bm v\bm v^\top - 2\bm u\bm v^\top \succeq 0$. Since the operator norm equals to the largest eigenvalue for the positive semi-definite matrix, we have that $\|\bm u\bm v^\top\|\leq \frac{1}{2}\|\bm u\bm u^\top + \bm v\bm v^\top\|\leq \frac{1}{2}(\|\bm u\bm u^\top\| + \|\bm v\bm v^\top\|)$.

To show that $\bm\Sigma_{\bm\mu_t^*}$ converges, it suffices to show that for any $\epsilon>0$, there exists an integer $T\geq t_0$ such that for any $s, t>T$, $\|\bm\Sigma_{\bm\mu_s^*} - \bm\Sigma_{\bm\mu_t^*}\| < \epsilon$. Without loss of generality, we assume that $s>t$. We then have that
\begin{flalign}
&\|\bm\Sigma_{\bm\mu_s^*} - \bm\Sigma_{\bm\mu_t^*}\|\nn\\& \leq  \Big\|\sum_{i=t}^{s-1}\big((\bm A\bm\Sigma^*\bm A^\top + \bm R_{\omega} +\bm I)^{-1}\bm A\big)^{i-t_0}\bm\Sigma^{*}\bm\Sigma_{(s-i+t_0)(s-i+t_0)}\Big(\big((\bm A\bm\Sigma^*\bm A^\top + \bm R_{\omega} +\bm I)^{-1}\bm A\big)^{i-t_0}\bm\Sigma^{*} \Big)^\top\Big\|\nn\\&\hspace{0.1cm}+ 2 \Big\|\sum_{i=t}^{s-1}\sum_{j=t_0}^{i-1}\big((\bm A\bm\Sigma^*\bm A^\top + \bm R_{\omega} +\bm I)^{-1}\bm A\big)^{i-t_0}\bm\Sigma^{*}\nn\\&\hspace{0.5cm}\bm\Sigma_{(s-i+t_0)(s-j+t_0)}\Big(\big((\bm A\bm\Sigma^*\bm A^\top + \bm R_{\omega} +\bm I)^{-1}\bm A\big)^{j-t_0}\bm\Sigma^{*}\Big)^\top\Big\|,\nn
\end{flalign}
where the inequality is from the triangle inequality of operator norm. To simplify the notation, let $b = \|(\bm A\bm\Sigma^*\bm A^\top + \bm R_{\omega} +\bm I)^{-1}\bm A\|<1-\delta, \|\widehat{\bm\Sigma}\| =\|\bm I\| + \frac{1}{1-\|\bm A\|^2}\|\bm R_{\omega}\|$. It then follows that
\begin{flalign}\label{eq:conv}
&\|\bm\Sigma_{\bm\mu_s^*} - \bm\Sigma_{\bm\mu_t^*}\|\nn\\& \leq \sum_{i=T+1}^{s-1} b^{2(i-t_0)}\|\bm\Sigma^*\|^2\|\bm\Sigma_{(s-i+t_0)(s-i+t_0)}\| + 2\sum_{i=T+1}^{s-1}\sum_{j=t_0}^{i-1}b^{(i+j-2t_0)}\|\bm\Sigma^*\|^2\|\bm\Sigma_{(s-i+t_0)(s-j+t_0)}\|\nn\\&\leq\sum_{i=T+1}^{s-1} b^{2(i-t_0)}\|\bm\Sigma^*\|^2\|\widehat{\bm\Sigma}\| + 2\sum_{i=T+1}^{s-1}\sum_{j=t_0}^{i-1}b^{(i+j-2t_0)}\|\bm\Sigma^*\|^2\|\widehat{\bm\Sigma}\|\nn\\&= \frac{b^{2(T+1-t_0)}(1-b^{2(s-T-1)})}{1-b^2}\|\bm\Sigma^*\|^2\|\widehat{\bm\Sigma}\| + 2\sum_{i=T+1}^{s-1}\frac{b^{(i-t_0)}(1-b^{(i-t_0)})}{1-b}\|\bm\Sigma^*\|^2\|\widehat{\bm\Sigma}\|\nn\\&\leq \frac{b^{2(T+1-t_0)}}{1-b^2}\|\bm\Sigma^*\|^2\|\widehat{\bm\Sigma}\| + 2\frac{b^{(T+1-t_0)}(1-b^{(s-T-1)})}{(1-b)^2}\|\bm\Sigma^*\|^2\|\widehat{\bm\Sigma}\|\nn\\&\leq \frac{b^{2(T+1-t_0)}}{1-b^2}\|\bm\Sigma^*\|^2\|\widehat{\bm\Sigma}\| + 2 \frac{b^{(T+1-t_0)}}{(1-b)^2}\|\bm\Sigma^*\|^2\|\widehat{\bm\Sigma}\|.
\end{flalign}

Since there exists $0<\delta<1$ such that $b<1-\delta$, for any $\epsilon > 0$, there exists an integer $T$ such that for any $s,t>T$, $\|\bm\Sigma_{\bm\mu_s^*} - \bm\Sigma_{\bm\mu_t^*}\| < \epsilon$. Therefore, $\bm\Sigma_{\mu^*_t}$ converges as $t\rightarrow\infty$. 
Therefore, given the initial state $\{\bm x_{t_0}, \bm y_{t_0}, \bm \mu_{t_0}^*\}$, %$\mathbb{E}_{t_0}\big[\lim_{t\rightarrow\infty}(\bm \mu_t^{*}-\mathbb{E}_{t_0}[\bm \mu_t^{*}|\bm x_{t_0}, \bm y_{t_0}, \bm \mu_{t_0}^*])(\bm\mu_t^* - \mathbb{E}_{t_0}[\bm \mu_t^{*}|\bm x_{t_0}, \bm y_{t_0}, \bm \mu_{t_0}^*])^T|\bm x_{t_0}, \bm y_{t_0}, \bm \mu_{t_0}^*\big]$ 
the limiting covariance matrix of $\bm \mu^*_t$ exists. Moreover, from the definition of $\bm\Sigma_{\bm\mu^*_t}$ in \eqref{eq:mucovariance}, $\bm \Sigma_{\bm\mu^*_t}$ is independent of the initial state $\{\bm x_{t_0}, \bm y_{t_0}, \bm \mu_{t_0}^*\}$. Since the limiting distribution of $\bm\mu^*_t$ exists and is independent of the initial state, from the definition of stationary distribution of Markov chain \cite{meyn2012markov}, $\bm\mu^*_t$ has a stationary distribution, which is a Gaussian distribution.

\textbf{Step 3.} We then show that $\{\bm y_t, \bm\mu^*_t\}_{t={t_0}}^\infty$ is a Markov chain and its stationary distribution exists and is unique.
We have that 
\begin{flalign}
&p_{t_0}(\bm y_t|\bm y_{t_0}, \cdots,\bm y_{t-1}, \bm \mu_{t_0}^*,\cdots, \bm \mu^*_{t-1})\nn\\& = \frac{p_{t_0}(\bm y_{t_0}, \cdots, \bm y_t, \bm\mu^*_{t_0})}{p_{t_0}(\bm y_{t_0}, \cdots, \bm y_{t-1}, \bm\mu^*_{t_0})}\cdot\frac{p_{t_0}(\bm \mu^*_{t_0+1}, \cdots, \bm \mu^*_{t-1}|\bm y_{t_0}, \cdots, \bm y_t, \bm\mu^*_{t_0})}{p_{t_0}(\bm \mu^*_{t_0+1}, \cdots, \bm \mu^*_{t-1}|\bm y_{t_0}, \cdots, \bm y_{t-1}, \bm\mu^*_{t_0})} \nn\\&\overset{(a)}{=} \frac{p_{t_0}(\bm y_{t_0}, \cdots, \bm y_t, \bm\mu^*_{t_0})}{p_{t_0}(\bm y_{t_0}, \cdots, \bm y_{t-1}, \bm\mu^*_{t_0})}\nn\\&= p_{t_0}(\bm y_t|\bm y_{t_0}, \cdots, \bm y_{t-1}, \bm\mu^*_{t_0})\nn\\& \overset{(b)}{=}  \frac{1}{\sqrt{(2\pi)^K\det(\bm A\bm\Sigma^*\bm A^\top + \bm R_{\omega})}}\frac{1}{\sqrt{\det((\bm A\bm\Sigma^*\bm A^\top + \bm R_{\omega})^{-1} + \bm I)}}\nn\\&\hspace{0.5cm}\exp\bigg(-\frac{1}{2}\Big((\bm A\bm\mu^*_{t-1})^\top(\bm A\bm\Sigma^*\bm A^\top + \bm R_{\omega})^{-1}(\bm A\bm\mu^*_{t-1})\nn\\&\hspace{0.5cm}+\bm y_{t}^\top \bm y_t - \big((\bm A\bm\Sigma^*\bm A^\top + \bm R_{\omega})^{-1}(\bm A\bm\mu^*_{t-1}) + \bm y_t\big)^\top\nn\\&\hspace{0.5cm}\big((\bm A\bm\Sigma^*\bm A^\top + \bm R_{\omega})^{-1} + \bm I\big)^{-1}\big((\bm A\bm\Sigma^*\bm A^\top + \bm R_{\omega})^{-1}(\bm A\bm\mu^*_{t-1}) + \bm y_t\big) \Big)\bigg),
\end{flalign}
which depends only on $\bm \mu^*_{t-1}$, and $(a)$ is due to the fact that conditioning on $(\bm y_{t_0}, \cdots, \bm y_{t-1})$, $(\bm \mu^*_{t_0+1}, \cdots, \bm \mu^*_{t-1})$ is independent of $\bm y_{t}$ and $(b)$ is from the update rule of the forward variable in Lemma \ref{lemma:forward}. We then have that
\begin{flalign}
&p_{t_0}(\bm y_t, \bm\mu_t^*|\bm y_{t_0},\cdots, \bm y_{t-1}, \bm \mu_{t_0}^*, \cdots, \bm \mu_{t-1}^*)\nn\\& = p_{t_0}(\bm y_t|\bm \mu_{t-1}^*) p_{t_0}(\bm\mu_{t}^*|\bm \mu_{t-1}^*, \bm y_t)\nn\\& = p_{t_0}(\bm y_t, \bm \mu_t^*|\bm y_{t-1}, \bm \mu_{t-1}^*).
\end{flalign}
Therefore, $\{\bm y_t, \bm\mu^*_t\}_{t={t_0}}^\infty$ is a Markov chain.

Since the stationary distribution of $\bm\mu^*_t$ exists, we have that 
\begin{flalign}
&\int\big(\pi(\bm\mu_{t-1}^*)p_{t_0}(\bm y_t, \bm \mu^*_t|\bm\mu^*_{t-1})d\bm\mu^*_{t-1}\big)p_{t_0}(\bm y_{t+1}|\bm \mu_t^*)p_{t_0}(\bm \mu_{t+1}^*|\bm\mu_{t},\bm y_{t+1})d\bm\mu^*_t\bm y_t \nn\\&= \int\pi(\bm\mu^*_t)p_{t_0}(\bm y_{t+1}|\bm \mu_t^*)p_{t_0}(\bm \mu_{t+1}^*|\bm\mu_{t},\bm y_{t+1})d\bm\mu^*_t\nn\\& = \int\pi(\bm\mu^*_t)p_{t_0}(\bm y_{t+1}, \bm\mu^*_{t+1}|\bm \mu_t^*)d\bm\mu^*_t.
\end{flalign}
Therefore, $\int\pi(\bm\mu_{t-1}^*)p_{t_0}(\bm y_t, \bm \mu^*_t|\bm\mu^*_{t-1})d\bm\mu^*_{t-1}$ is a stationary distribution of $\{\bm y_t, \bm\mu^*_t\}_{t={t_0}}^\infty$. Following the same techniques as in the proof of Lemma \ref{lemma:station}, we have that $\{\bm y_t, \bm\mu^*_t\}_{t={t_0}}^\infty$ is $\pi$-irreducible. Therefore, the stationary distribution of $\{\bm y_t, \bm\mu^*_t\}_{t={t_0}}^\infty$ is unique.
\end{proof}

\section{Proof of Theorem \ref{theorem:waddlow}}\label{sec:theorem1}
\begin{proof}
\textbf{Step 1.} We first consider a special case where $\bm\Sigma_{t_0} = \bm\Sigma^*$. We have that 
\begin{flalign}
&\log p_{t_0}^*(\bm y_{t_0},\cdots, \bm y_t)\nn\\& = \log \frac{p_{t_0}^*(\bm y_{t_0},\cdots, \bm y_t)}{p_{t_0}^*(\bm y_{t_0},\cdots, \bm y_{t-1})} + \log \frac{p_{t_0}^*(\bm y_{t_0},\cdots, \bm y_{t-1})}{p_{t_0}^*(\bm y_{t_0},\cdots, \bm y_{t-2})} + \cdots + \log\frac{p_{t_0}^*(\bm y_{t_0}, \bm y_{t_0+1})}{p_{t_0}^*(\bm y_{t_0})} + \log p_{t_0}^*(\bm y_{t_0})
\end{flalign}
and 
\begin{flalign}\label{eq:ycondition}
&\log \frac{p_{t_0}^*(\bm y_{t_0},\cdots, \bm y_t)}{p_{t_0}^*(\bm y_{t_0},\cdots, \bm y_{t-1})}\nn\\&  = \log\frac{\int\alpha_t(\bm x_t)d\bm x_t}{\int\alpha_{t-1}(\bm x_{t-1})d\bm x_{t-1}}\nn\\& =\log\Big( \frac{1}{\sqrt{(2\pi)^K\det(\bm A\bm\Sigma^*\bm A^\top + \bm R_{\omega})}}\frac{1}{\sqrt{\det((\bm A\bm\Sigma^*\bm A^\top + \bm R_{\omega})^{-1} + \bm I)}}\Big) \nn\\&\hspace{0.2cm}-\frac{1}{2}\Big((\bm A\bm\mu_{t-1}^*)^\top(\bm A\bm\Sigma^*\bm A^\top + \bm R_{\omega})^{-1}(\bm A\bm\mu_{t-1}^*)\nn\\&\hspace{0.2cm}+\bm y_{t}^\top \bm y_t - \big((\bm A\bm\Sigma^*\bm A^\top + \bm R_{\omega})^{-1}(\bm A\bm\mu_{t-1}^*) + \bm y_t\big)^\top\nn\\&\hspace{0.2cm}\big((\bm A\bm\Sigma^*\bm A^\top + \bm R_{\omega})^{-1} + \bm I\big)^{-1} \big((\bm A\bm\Sigma^*\bm A^\top + \bm R_{\omega})^{-1}(\bm A\bm\mu_{t-1}^*) + \bm y_t\big) \Big),
\end{flalign}
where the second equality is from the updated rule of $\alpha_t(\bm x_t)$. 

Since $\bm\mu_t^* = (\bm A\bm\Sigma^*\bm A^\top + \bm R_{\omega} +\bm I)^{-1}\bm A\bm\mu_{t-1}^* + \bm\Sigma^*\bm y_t$, we replace $\bm A\bm\mu_{t-1}^*$ by $(\bm A\bm\Sigma^*\bm A^\top + \bm R_{\omega} +\bm I)(\bm\mu_t^* - \bm\Sigma^*\bm y_t)$. Therefore, $\log \frac{p_{t_0}^*(\bm y_{t_0},\cdots, \bm y_t)}{p_{t_0}^*(\bm y_{t_0},\cdots, \bm y_{t-1})}$ can be equivalently written as a function of only $\bm\mu^*_t, \bm y_t$. Let $h(\bm\mu^*_t, \bm y_t) = \log \frac{p_{t_0}^*(\bm y_{t_0},\cdots, \bm y_t)}{p_{t_0}^*(\bm y_{t_0},\cdots, \bm y_{t-1})}$. We have the following explicit expression for $h(\bm\mu^*_t, \bm y_t)$, 
\begin{flalign}
    &h(\bm\mu^*_t, \bm y_t)\nn\\ &= \log\Big( \frac{1}{\sqrt{(2\pi)^K\det(\bm A\bm\Sigma^*\bm A^\top + \bm R_{\omega})}}\frac{1}{\sqrt{\det((\bm A\bm\Sigma^*\bm A^\top + \bm R_{\omega})^{-1} + \bm I)}}\Big) \nn\\&\hspace{0.2cm}-\frac{1}{2}\Big(((\bm A\bm\Sigma^*\bm A^\top + \bm R_{\omega} +\bm I)(\bm\mu_t^* - \bm\Sigma^*\bm y_t))^\top(\bm A\bm\Sigma^*\bm A^\top + \bm R_{\omega})^{-1}((\bm A\bm\Sigma^*\bm A^\top + \bm R_{\omega} +\bm I)(\bm\mu_t^* - \bm\Sigma^*\bm y_t))\nn\\&\hspace{0.2cm}+\bm y_{t}^\top \bm y_t - \big((\bm A\bm\Sigma^*\bm A^\top + \bm R_{\omega})^{-1}((\bm A\bm\Sigma^*\bm A^\top + \bm R_{\omega} +\bm I)(\bm\mu_t^* - \bm\Sigma^*\bm y_t)) + \bm y_t\big)^\top\nn\\&\hspace{0.2cm}\big((\bm A\bm\Sigma^*\bm A^\top + \bm R_{\omega})^{-1} + \bm I\big)^{-1} \big((\bm A\bm\Sigma^*\bm A^\top + \bm R_{\omega})^{-1}((\bm A\bm\Sigma^*\bm A^\top + \bm R_{\omega} +\bm I)(\bm\mu_t^* - \bm\Sigma^*\bm y_t)) + \bm y_t\big) \Big)
\end{flalign}
Since $\pi(\bm\mu^*_t), \pi(\bm y_t)$ are Gaussian distributions and $h(\bm\mu, \bm y)$ is a quadratic function of $\bm\mu, \bm y$, it can be easily verified that $\mathbb{E}_\pi[|h(\bm\mu, \bm y)|]<\infty$. We further note that $\mathbb{R}^{2K}$ is a state space with countably generated $\sigma$-algebra \cite{roberts2004general} and $\{\bm y_t, \bm\mu^*_t\}_{t={t_0}}^\infty$ is $\pi$-irreducible and aperiodic. Therefore, from the ergodic theorem of Markov chain \cite{roberts2004general,meyn2012markov}, we have that under $\mathbb{P}_{t_0}$,
\begin{flalign}
\lim_{t\rightarrow\infty}\frac{1}{t}\log p_{t_0}^*(\bm y_{t_0},\cdots, \bm y_{t_0+t-1}) = \mathbb{E}_\pi[h(\bm\mu, \bm y)]
\end{flalign}
almost surely. Moreover, the convergence result does not depend on the initial state of $\{\bm y_t, \bm\mu^*_t\}_{t={t_0}}^\infty$.

\textbf{Step 2.} We then show that $\lim_{t\rightarrow\infty}\frac{1}{t}\log p_{t_0}(\bm y_{t_0},\cdots, \bm y_{t_0+t-1}) = \mathbb{E}_\pi[h(\bm\mu, \bm y)]$ almost surely under $\mathbb{P}_{t_0}$ by showing that 
\begin{flalign}\label{eq:conzero}
&\lim_{t\rightarrow\infty}\frac{1}{t}\big(\log p_{t_0}(\bm y_{t_0},\cdots, \bm y_{t_0+t-1}) - \log p_{t_0}^*(\bm y_{t_0},\cdots, \bm y_{t_0+t-1})\big) = 0
\end{flalign}
almost surely under $\mathbb{P}_{t_0}$. 
%The basic idea is that we represent $\bm\Sigma_t$ as $\bm \Sigma^* + \bm\Lambda_t$, where $\bm \Lambda_t\rightarrow \bm 0$ as $t\rightarrow 0$. We then represent $\log p_{t_0}(\bm y_{t_0},\cdots, \bm y_{t_0+t-1})$ using $\bm \Sigma^*, \bm \Lambda_t, \bm \mu^*$. We further show that \eqref{eq:conzero} holds due to the fact that $\bm \Lambda_t\rightarrow \bm 0$ and the ergodic theorem \cite{roberts2004general, meyn2012markov}. Due to the space limitation, we omit the details of this part. 

We have that 
\begin{flalign}
&\log p_{t_0}(\bm y_{t_0},\cdots, \bm y_{t})\nn\\& = \log \frac{p_{t_0}(\bm y_{t_0},\cdots, \bm y_{t})}{p_{t_0}(\bm y_{t_0},\cdots, \bm y_{t-1})} + \log \frac{p_{t_0}(\bm y_{t_0},\cdots, \bm y_{t-1})}{p_{t_0}(\bm y_{t_0},\cdots, \bm y_{t-2})} + \cdots + \log\frac{p_{t_0}(\bm y_{t_0}, \bm y_{t_0+1})}{p_{t_0}(\bm y_{t_0})} + p_{t_0}(\bm y_{t_0})
\end{flalign}
and 
\begin{flalign}\label{eq:eachterm}
&\log \frac{p_{t_0}(\bm y_{t_0},\cdots, \bm y_t)}{p_{t_0}(\bm y_{t_0},\cdots, \bm y_{t-1})}\nn\\&= \log\Big( \frac{1}{\sqrt{(2\pi)^K\det(\bm A\bm\Sigma_{t-1}\bm A^\top + \bm R_{\omega})}}\frac{1}{\sqrt{\det((\bm A\bm\Sigma_{t-1}\bm A^\top + \bm R_{\omega})^{-1} + \bm I)}}\Big) \nn\\&\hspace{0.2cm}-\frac{1}{2}\Big((\bm A\bm\mu_{t-1})^\top(\bm A\bm\Sigma_{t-1}\bm A^\top + \bm R_{\omega})^{-1}(\bm A\bm\mu_{t-1})\nn\\&\hspace{0.2cm}+\bm y_{t}^\top \bm y_t - \big((\bm A\bm\Sigma_{t-1}\bm A^\top + \bm R_{\omega})^{-1}(\bm A\bm\mu_{t-1}) + \bm y_t\big)^\top\nn\\&\hspace{0.2cm}\big((\bm A\bm\Sigma_{t-1}\bm A^\top + \bm R_{\omega})^{-1} + \bm I\big)^{-1} \big((\bm A\bm\Sigma_{t-1}\bm A^\top + \bm R_{\omega})^{-1}(\bm A\bm\mu_{t-1}) + \bm y_t\big) \Big).
\end{flalign}

To show that $\lim_{t\rightarrow\infty}\frac{1}{t}\big(\log p_{t_0}(\bm y_{t_0},\cdots, \bm y_{t_0+t-1}) - \log p_{t_0}^*(\bm y_{t_0},\cdots, \bm y_{t_0+t-1})\big)=0$. It suffices to show the limit sum of each term in \eqref{eq:eachterm} converges. Here we provide the proof of $\lim_{t\rightarrow\infty}\frac{1}{t}\sum_{i=t_0}^{t_0+t-1}\Big((\bm A\bm\mu_i)^\top(\bm A \bm\Sigma_{i-1}\bm A^\top + \bm R_{\omega})^{-1}(\bm A\bm\mu_i) - (\bm A\bm\mu_i^*)^\top(\bm A \bm\Sigma^*\bm A^\top + \bm R_{\omega})^{-1}(\bm A\bm\mu_i^*)\Big) = 0$. The rest of the terms can be proved using the same techniques. Let 
\begin{flalign}
&(\bm A\bm\Sigma_{t}\bm A^\top + \bm R_{\omega} + \bm I)^{-1} = (\bm A\bm\Sigma^*\bm A^\top + \bm R_{\omega} + \bm I)^{-1} + \bm\Lambda_t^1, \nn\\&(\bm A\bm\Sigma_{t}\bm A^\top + \bm R_{\omega})(\bm A\bm\Sigma_{t}\bm A^\top + \bm R_{\omega} + \bm I)^{-1} = (\bm A\bm\Sigma_t\bm A^\top + \bm R_{\omega})(\bm A\bm\Sigma^*\bm A^\top + \bm R_{\omega} + \bm I)^{-1} + \bm\Lambda_t^2,\nn\\& (\bm A\bm\Sigma_{t}\bm A^\top + \bm R_{\omega})^{-1} = (\bm A\bm\Sigma^*\bm A^\top + \bm R_{\omega})^{-1} + \bm\Lambda_t^3.
\nn
\end{flalign}
Since $\bm\Sigma_t$ converges to $\bm\Sigma^*$ as $t\rightarrow\infty$ and $(\bm A\bm\Sigma_{t}\bm A^\top + \bm R_{\omega} + \bm I)^{-1}, (\bm A\bm\Sigma_t\bm A^\top + \bm R_{\omega})(\bm A\bm\Sigma_t\bm A^\top + \bm R_{\omega} + \bm I)^{-1}, (\bm A\bm\Sigma_{t}\bm A^\top + \bm R_{\omega})^{-1}$ are continuous in $\bm\Sigma_t$, we have that $\lim_{t\rightarrow\infty}\|\bm\Lambda_t^1\| = 0$, $\lim_{t\rightarrow\infty}\|\bm\Lambda_t^2\| = 0$ and $\lim_{t\rightarrow\infty}\|\bm\Lambda_t^3\| = 0$. Therefore, $\|\bm\Lambda_t^1\|, \|\bm\Lambda_t^2\|, \|\bm\Lambda_t^3\|$ are uniformly upper bounded for any $t$. Denote by $\|\bm\Lambda^1\|, \|\bm\Lambda^2\|, \|\bm\Lambda^3\|$ the upper bound of $\|\bm\Lambda_t^1\|, \|\bm\Lambda_t^2\|, \|\bm\Lambda_t^3\|$ respectively. Given the initial $\bm\mu_{t_0}$, from \eqref{eq:updaterule}, we have the following representation for $\bm\mu_t$:
\begin{flalign}
&\bm\mu_{t_0+1} = \bm\Lambda_{t_0+1}^1\bm A\bm\mu_{t_0}+ \bm\Lambda_{t_0+1}^2\bm y_{t_0+1} + \bm \mu_{t_0+1}^*,\nn\\
&\bm\mu_{t_0+2} = \big((\bm A\bm\Sigma^*\bm A^\top+\bm R_\omega+\bm I)^{-1}+\bm\Lambda_{t_0+2}^1\big)\bm A(\bm\Lambda_{t_0+1}^1\bm A\bm\mu_{t_0} + \bm\Lambda_{t_0+1}^2\bm y_{t_0+1} + \bm \mu_{t_0+1}^*) \nn\\&\hspace{0.2cm} +
\big((\bm A\bm\Sigma^*\bm A^\top+\bm R_\omega)(\bm A\bm\Sigma^*\bm A^\top+\bm R_\omega+\bm I)^{-1}+\bm\Lambda_{t_0+2}^2\big)\bm y_{t_0+2}\nn\\&= \big((\bm A\bm\Sigma^*\bm A^\top+\bm R_\omega+\bm I)^{-1}+\bm\Lambda_{t_0+2}^1\big)\bm A(\bm\Lambda_{t_0+1}^1\bm A\bm\mu_{t_0} + \bm\Lambda_{t_0+1}^2\bm y_{t_0+1})+\bm\Lambda_{t_0+2}^2\bm y_{t_0+2} + \bm\mu_{t_0+2}^*\nn\\&\hspace{3cm}\vdots\hspace{3cm} \vdots\hspace{3cm}\vdots\nn\\
&\bm\mu_{t_0+t-1} = \bm\Lambda_{t_0+t-1}^2\bm y_{t_0+t-1} + \big((\bm A\bm\Sigma^*\bm A^\top+\bm R_\omega+\bm I)^{-1}+\bm\Lambda_{t_0+t-1}^1\big)\bm A\bm\Lambda_{t_0+t-2}^2\bm y_{t_0+t-2} \nn\\&\hspace{0.2cm}+ \big((\bm A\bm\Sigma^*\bm A^\top+\bm R_\omega+\bm I)^{-1}+\bm\Lambda_{t_0+t-1}^1\big)\bm A\big((\bm A\bm\Sigma^*\bm A^\top+\bm R_\omega+\bm I)^{-1} +\bm\Lambda_{t_0+t-2}^1\big)\bm A\nn\\&\hspace{0.2cm}\bm\Lambda_{t_0+t-3}^2\bm y_{t_0+t-3}+\cdots+\prod_{i=t_0+2}^{t_0+t-1}\Big(\big((\bm A\bm\Sigma^*\bm A^\top+\bm R_\omega+\bm I)^{-1}+\bm\Lambda_{i}^1\big)\bm A\Big)(\bm\Lambda_{t_0+1}^1\bm A\bm\mu_{t_0} +\bm\Lambda_{t_0+1}^2\bm y_{t_0+1})\nn\\&\hspace{0.2cm} + \bm\mu_{t_0+t-1}^*.
\end{flalign}
Let $\bm Y_t = (\bm y_{t_0}, \bm y_{t_0+1},\cdots,\bm y_{t_0+t-1})$ and denote $\bm\mu_t$ as $F_t(\bm Y_t) + \bm\mu^*_t$. We then have that
\begin{flalign}
&\frac{1}{t}\sum_{i=t_0}^{t_0+t-1} \Big((\bm A\bm\mu_i)^\top\big((\bm A \bm\Sigma^*\bm A^\top + \bm R_{\omega})^{-1}+\bm\Lambda_i^3\big)(\bm A\bm\mu_i) - (\bm A\bm\mu_i^*)^\top(\bm A \bm\Sigma^*\bm A^\top + \bm R_{\omega})^{-1}(\bm A\bm\mu_i^*)\Big)\nn\\&= \frac{1}{t}\sum_{i=t_0}^{t_0+t-1}\Big(\big(\bm AF_i(\bm Y_i)+\bm A\bm\mu^*_i\big)^\top\big((\bm A \bm\Sigma^*\bm A^\top + \bm R_{\omega})^{-1}+\bm\Lambda_i^3\big)\nn\\&\hspace{0.2cm}\big(\bm AF_i(\bm Y_i)+\bm A\bm\mu^*_i\big)-(\bm A\bm\mu_i^*)^\top(\bm A \bm\Sigma^*\bm A^\top + \bm R_{\omega})^{-1}(\bm A\bm\mu_i^*)\Big)\nn\\&= \frac{1}{t}\sum_{i=t_0}^{t_0+t-1} \Big(\big(\bm AF_i(\bm Y_i)\big)^\top(\bm A \bm\Sigma^*\bm A^\top + \bm R_{\omega})^{-1}\big(\bm AF_i(\bm Y_i)\big)\nn\\&\hspace{0.2cm}+\big(\bm AF_i(\bm Y_i)\big)^\top\bm\Lambda_i^3\big(\bm AF_i(\bm Y_i)\big) + \big(\bm A\bm \mu^*_i\big)^\top\bm\Lambda_i^3\big(\bm A\bm\mu^*_i\big)\nn\\&\hspace{0.2cm} + 2\big(\bm AF_i(\bm Y_i)\big)^\top(\bm A \bm\Sigma^*\bm A^\top + \bm R_{\omega})^{-1}(\bm A\bm\mu_i^*) + 2\big(\bm AF_i(\bm Y_i)\big)^\top\bm\Lambda_i^3(\bm A\bm\mu_i^*)\Big).
\end{flalign}

In the following, we will show that 
\begin{flalign}
&\lim_{t\rightarrow\infty}\frac{1}{t}\sum_{i=t_0}^{t_0+t-1} \big(\bm AF_i(\bm Y_i)\big)^\top (\bm A \bm\Sigma^*\bm A^\top+\bm R_{\omega})^{-1} \big(\bm AF_i(\bm Y_i)\big)=0\nn
\end{flalign}
and 
\begin{flalign}
\lim_{t\rightarrow\infty}\frac{1}{t}\sum_{i=t_0}^{t_0+t-1}\big(\bm AF_i(\bm Y_i)\big)^\top(\bm A \bm\Sigma^*\bm A^\top + \bm R_{\omega})^{-1}(\bm A\bm\mu_i^*)=0. \nn
\end{flalign}
The convergence of the rest of terms can be proved using the same techniques. To simplify the notation, we set $\bm \mu_{t_0} = \bm 0$. The proof for any arbitrary $\bm\mu_{t_0}$ can also be derived similarly. For $\big(\bm AF_i(\bm Y_i)\big)^\top (\bm A \bm\Sigma^*\bm A^\top+\bm R_{\omega})^{-1} \big(\bm AF_i(\bm Y_i)\big)$, $i = t_0, \cdots, t$, we have that 
\begin{flalign}
&\big(\bm AF_{t_0+1}(\bm Y_{t_0+1})\big)^\top (\bm A \bm\Sigma^*\bm A^\top+\bm R_{\omega})^{-1} \big(\bm AF_{t_0+1}(\bm Y_{t_0+1})\big)\nn\\& = (\bm\Lambda_{t_0+1}^2\bm y_{t_0+1})^\top\bm A^\top(\bm A \bm\Sigma^*\bm A^\top+\bm R_{\omega})^{-1}\bm A\bm\Lambda_{t_0+1}^2\bm y_{t_0+1},\nn\\
& \big(\bm AF_{t_0+2}(\bm Y_{t_0+2})\big)^\top(\bm A \bm\Sigma^*\bm A^\top + \bm R_{\omega})^{-1}\big(\bm AF_{t_0+2}(\bm Y_{t_0+2})\big)\nn\\&= \Big(\big((\bm A\bm\Sigma^*\bm A^\top+\bm R_\omega+\bm I)^{-1}+\bm\Lambda_{t_0+2}^1\big)\bm A(\bm\Lambda_{t_0+1}^2\bm y_{t_0+1}) + \bm \Lambda_{t_0+2}^2\bm y_{t_0+2}\Big)^\top\bm A^\top(\bm A \bm\Sigma^*\bm A^\top+\bm R_{\omega})^{-1}\bm A \nn\\&\hspace{0.2cm}\Big(\big((\bm A\bm\Sigma^*\bm A^\top+\bm R_\omega+\bm I)^{-1}+\bm\Lambda_{t_0+2}^1\big)\bm A(\bm\Lambda_{t_0+1}^2\bm y_{t_0+1}) + \bm \Lambda_{t_0+2}^2\bm y_{t_0+2}\Big),\nn\\
&\big(\bm AF_{t_0+3}(\bm Y_{t_0+3})\big)^\top(\bm A \bm\Sigma^*\bm A^\top + \bm R_{\omega})^{-1}\big(\bm AF_{t_0+3}(\bm Y_{t_0+3})\big)\nn\\&= \Big(\big((\bm A\bm\Sigma^*\bm A^\top+\bm R_\omega+\bm I)^{-1}+\bm\Lambda_{t_0+3}^1\big)\bm A\big((\bm A\bm\Sigma^*\bm A^\top+\bm R_\omega+\bm I)^{-1}+\bm\Lambda_{t_0+2}^1\big)\bm A(\bm\Lambda_{t_0+1}^2\bm y_{t_0+1})\nn\\&\hspace{0.2cm} +\big((\bm A\bm\Sigma^*\bm A^\top+\bm R_\omega+\bm I)^{-1}+\bm\Lambda_{t_0+3}^1\big)\bm A(\bm\Lambda_{t_0+2}^2\bm y_{t_0+2}) + \bm \Lambda_{t_0+3}^2\bm y_{t_0+3}\Big)^\top\bm A^\top(\bm A \bm\Sigma^*\bm A^\top+\bm R_{\omega})^{-1}\bm A \nn\\&\hspace{0.2cm}\Big(\big((\bm A\bm\Sigma^*\bm A^\top+\bm R_\omega+\bm I)^{-1}+\bm\Lambda_{t_0+3}^1\big)\bm A\big((\bm A\bm\Sigma^*\bm A^\top+\bm R_\omega+\bm I)^{-1}+\bm\Lambda_{t_0+2}^1\big)\bm A(\bm\Lambda_{t_0+1}^2\bm y_{t_0+1})\nn\\&\hspace{0.2cm} +\big((\bm A\bm\Sigma^*\bm A^\top+\bm R_\omega+\bm I)^{-1}+\bm\Lambda_{t_0+3}^1\big)\bm A(\bm\Lambda_{t_0+2}^2\bm y_{t_0+2}) + \bm \Lambda_{t_0+3}^2\bm y_{t_0+3}\Big)\nn\\&\hspace{3cm}\vdots\hspace{3cm} \vdots\hspace{3cm} \vdots
\end{flalign}

We note that there exists a $0<\delta<1$ such that for any $i\geq t_0$, $\big\|\big((\bm A\bm\Sigma^*\bm A^\top+\bm R_\omega+\bm I)^{-1}+\bm\Lambda_{i}^1\big)\bm A\big\| = \big\|\big(\bm A\bm\Sigma_{i}\bm A^\top + \bm R_{\omega} + \bm I\big)^{-1}\bm A\big\| \leq \big\|\big(\bm A\bm\Sigma_{i}\bm A^\top + \bm R_{\omega} + \bm I\big)^{-1}\big\|\|\bm A\| \leq (1-\delta)$. Let $B=1-\delta$ and $C = \big\|\bm A^\top (\bm A \bm\Sigma^*\bm A^\top+\bm R_{\omega})^{-1}\bm A\big\|$. We then have that 
\begin{flalign}
&\lim_{t\rightarrow\infty}\frac{1}{t}\sum_{i=t_0}^{t_0+t-1} \Big\|\Big(\big(\bm AF_i(\bm Y_i)\big)^\top (\bm A \bm\Sigma^*\bm A^\top+\bm R_{\omega})^{-1} \big(\bm AF_i(\bm Y_i)\big)\Big\| \nn\\& \leq \lim_{t\rightarrow\infty}\frac{1}{t}\sum_{i=t_0}^{t_0+t-1}\big\|F_i(\bm Y_i)\big\|^2 \big\|\bm A^\top (\bm A \bm\Sigma^*\bm A^\top+\bm R_{\omega})^{-1}\bm A\big\|\nn\\&\leq \lim_{t\rightarrow\infty}\frac{1}{t}\sum_{i=t_0}^{t_0+t-1} C\big\|F_i(\bm Y_i)\big\|^2 \nn\\&\leq \lim_{t\rightarrow\infty}\frac{1}{t}\sum_{i=t_0}^{t_0+t-1} C \bigg(\sum_{j=t_0}^i\ B^{i-j}\|\bm\Lambda^{2}_{j}\|\|\bm y_{j}\|\bigg)^2 \nn\\&= \lim_{t\rightarrow\infty}\frac{1}{t}\sum_{i=t_0}^{t_0+t-1}\sum_{j=t_0}^{i}\sum_{k=t_0}^{i}C B^{2i-j-k}\|\bm\Lambda^2_j\|\|\bm y_j\|\|\bm\Lambda^2_k\|\|\bm y_k\|\nn\\&\leq \lim_{t\rightarrow\infty}\frac{1}{t}\sum_{i=t_0}^{t_0+t-1}\sum_{j=t_0}^{i}\sum_{k=t_0}^{i}CB^{2i-j-k} \big(\|\bm\Lambda^2_j\|^2\|\bm y_j\|^2 + \|\bm\Lambda^2_k\|^2\|\bm y_k\|^2\big), 
\end{flalign}
where the first inequality is due to the submultiplicative of operator norm, and the third inequality is from the explicit expression of $F_i(\bm Y_i)$ and the triangle inequality and the submultiplicative of operator norm. We then consider the coefficient of the sum of all terms containing $\|\bm\Lambda^2_{t_0}\|^2\|\bm y_{t_0}\|^2$, denoted by $Co(\|\bm\Lambda^2_{t_0}\|^2\|\bm y_{t_0}\|^2)$. We have that 
\begin{flalign}
Co(\|\bm\Lambda^2_{t_0}\|^2\|\bm y_{t_0}\|^2) &= \lim_{t\rightarrow\infty}\sum_{i=t_0}^{t_0+t-1}\sum_{j=t_0}^{i}2CB^{2i-t_0-j} \nn\\&\leq\lim_{t\rightarrow\infty}\sum_{i=t_0}^{t_0+t-1} \frac{2CB^{i-t_0}}{1-B}\nn\\&\leq \frac{2C}{(1-B)^2}.
\end{flalign}

Similarly, we can show that the coefficient of $\|\bm\Lambda^2_{t_0+1}\|^2\|\bm y_{t_0+1}\|^2, \cdots, \|\bm\Lambda^2_{t_0+t-1}\|^2\|\bm y_{t_0+t-1}\|^2$ are not larger than $\frac{2C}{(1-B)^2}$. Since $\lim_{t\rightarrow\infty}\|\bm\Lambda_t^2\|^2 = 0$, for any $\epsilon>0$, there exists an integer $T$ such that for any $t>T$, $\|\bm\Lambda_t^2\|^2 < \epsilon$. We then have that 
\begin{flalign}\label{eq:sumofy}
    &\lim_{t\rightarrow\infty}\frac{1}{t}\sum_{i=t_0}^{t_0+t-1} \big(\bm AF_i(\bm Y_i)\big)^\top (\bm A \bm\Sigma^*\bm A^\top+\bm R_{\omega})^{-1} \big(\bm AF_i(\bm Y_i)\big)\nn\\& \leq \lim_{t\rightarrow\infty}\frac{1}{t}2C\frac{\|\bm\Lambda^2_{t_0}\|^2\|\bm y_{t_0}\|^2 + \cdots + \|\bm\Lambda^2_{t_0+t-1}\|^2\|\bm y_{t_0+t-1}\|^2}{(1-B)^2}\nn\\& \leq\lim_{t\rightarrow\infty}\frac{1}{t}2C\bigg(\frac{\|\bm\Lambda^2_{t_0}\|^2\|\bm y_{t_0}\|^2 +\cdots+\|\bm\Lambda^2_T\|^2\|\bm y_T\|^2)}{(1-B)^2} +\frac{\epsilon\big(\|\bm y_{T+1}\|^2 + \cdots +\|\bm y_{t_0+t-1}\|^2\big)}{(1-B)^2} \bigg)\nn\\&\leq \frac{2\epsilon C\mathbb{E}_\pi[\|\bm y\|^2]}{(1-B)^2},
\end{flalign}
where the last inequality is from the ergodic theorem of Markov chain \cite{roberts2004general,meyn2012markov}. 

Since $\mathbb{E}_\pi[\|\bm y\|^2]$ is bounded, by letting $\epsilon\rightarrow 0$, we have that 
\begin{flalign}
    &\lim_{t\rightarrow\infty}\frac{1}{t}\sum_{i=t_0}^{t_0+t-1} \big\|\big(\bm AF_i(\bm Y_i)\big)^\top (\bm A \bm\Sigma^*\bm A^\top+\bm R_{\omega})^{-1} \big(\bm AF_i(\bm Y_i)\big)\big\| = 0.
\end{flalign}
Therefore, we have that 
\begin{flalign}
    &\lim_{t\rightarrow\infty}\frac{1}{t}\sum_{i=t_0}^{t_0+t-1}\big(\bm AF_i(\bm Y_i)\big)^\top (\bm A \bm\Sigma^*\bm A^\top + \bm R_{\omega})^{-1} \big(\bm AF_i(\bm Y_i)\big)=\bm 0.
\end{flalign}

We then consider $\lim_{t\rightarrow\infty}\frac{1}{t}\sum_{i=t_0}^{t_0+t-1}\big(\bm AF_i(\bm Y_i)\big)^\top(\bm A \bm\Sigma^*\bm A^\top + \bm R_{\omega})^{-1}(\bm A\bm\mu_i^*)$. We have that 
\begin{flalign}
&\lim_{t\rightarrow\infty}\frac{1}{t}\sum_{i=t_0}^{t_0+t-1}\big\|\big(\bm AF_i(\bm Y_i)\big)^\top(\bm A \bm\Sigma^*\bm A^\top + \bm R_{\omega})^{-1}(\bm A\bm\mu_i^*)\big\|\nn\\&\leq \lim_{t\rightarrow\infty}\frac{1}{t}\sum_{i=t_0}^{t_0+t-1} \big\|F_i(\bm Y_i)\big\| \big\|\bm A^\top (\bm A \bm\Sigma^*\bm A^\top+\bm R_{\omega})^{-1}\bm A\big\|\big\|\bm\mu^*_i\big\|\nn\\&\leq\lim_{t\rightarrow\infty}\frac{1}{t}\sum_{i=t_0}^{t_0+t-1} C\big\|F_i(\bm Y_i)\big\|\big\|\bm\mu^*_i\big\|\nn\\& = \lim_{t\rightarrow\infty}\frac{1}{t}\sum_{i=t_0}^{t_0+t-1} \sum_{j=t_0}^i C B^{i-j}\|\bm\Lambda^2_j\|\|\bm y_j\| \|\bm\mu^*_i\| \nn\\&\leq \lim_{t\rightarrow\infty}\frac{1}{t}\sum_{i=t_0}^{t_0+t-1} \sum_{j=t_0}^i CB^{i-j} \|\bm\Lambda^2_j\|\big(\|\bm y_j\|^2+\|\bm\mu^*_i\|^2\big).
\end{flalign}
It suffices to show that 
\begin{flalign}
    \lim_{t\rightarrow\infty}\frac{1}{t}\sum_{i=t_0}^{t_0+t-1} \sum_{j=t_0}^i CB^{i-j} \|\bm\Lambda^2_j\|\|\bm y_j\|^2 = 0
\end{flalign}
and 
\begin{flalign}
\lim_{t\rightarrow\infty}\frac{1}{t}\sum_{i=t_0}^{t_0+t-1} \sum_{j=t_0}^i CB^{i-j} \|\bm\Lambda^2_j\|\|\bm\mu^*_i\|^2=0.
\end{flalign}
We have that 
\begin{flalign}
&\lim_{t\rightarrow\infty}\frac{1}{t}\sum_{i=t_0}^{t_0+t-1} \sum_{j=t_0}^i CB^{i-j} \|\bm\Lambda^2_j\|\|\bm y_j\|^2 
\nn\\&= \lim_{t\rightarrow\infty}\frac{1}{t}\sum_{j=t_0}^{t_0+t-1} \sum_{i=j}^{t_0+t-1} CB^{i-j} \|\bm\Lambda^2_j\|\|\bm y_j\|^2
%\nn\\&= \lim_{t\rightarrow\infty}\frac{1}{t}\sum_{j=t_0}^{t_0+t-1} C\frac{1-B^{t-j+1}}{1-B}\|\bm\Lambda^2_j\|\|\bm y_j\|^2
\nn\\&\leq \lim_{t\rightarrow\infty}\frac{1}{t}\sum_{j=t_0}^{t_0+t-1}\frac{C}{1-B}\|\bm\Lambda^2_j\|\|\bm y_j\|^2.
\end{flalign}
For any $\epsilon>0$, there exists an integer $T$ such that for any $t>T$, $\|\bm\Lambda^2_t\|<\epsilon$. It then follows that
\begin{flalign}
&\lim_{t\rightarrow\infty}\frac{1}{t}\sum_{j=t_0}^{t_0+t-1}\frac{C}{1-B}\|\bm\Lambda^2_j\|\|\bm y_j\|^2\nn\\&\leq \lim_{t\rightarrow\infty}\frac{1}{t}C\Big(\frac{\|\bm\Lambda^2_{t_0}\|\|\bm y_{t_0}\|^2 + \cdots+\|\bm\Lambda^2_T\|\|\bm y_T\|^2 }{1-B}+\frac{\epsilon \|\bm y_{T+1}\|^2 +\cdots + \epsilon \|\bm y_{t_0+t-1}\|^2}{1-B}\Big)\nn\\&\leq \frac{\epsilon C\mathbb{E}_\pi[\|\bm y\|^2]}{1-B}.
\end{flalign}
Since $\mathbb{E}_\pi[\|\bm y\|^2]$ is bounded, by letting $\epsilon\rightarrow 0$, we have that $\lim_{t\rightarrow\infty}\frac{1}{t}\sum_{j=t_0}^{t_0+t-1}\frac{C}{1-B}\|\bm\Lambda^2_j\|\|\bm y_j\|^2=0$.

Consider $\lim_{t\rightarrow\infty}\frac{1}{t}\sum_{i=t_0}^{t_0+t-1} \sum_{j=t_0}^i CB^{i-j} \|\bm\Lambda^2_j\|\|\bm\mu^*_i\|^2$. We have that $\lim_{t\rightarrow\infty} B^t = 0$ and $\lim_{t\rightarrow\infty}\|\bm\Lambda_t^2\| = 0$. Therefore, for any $\epsilon>0$, there exists $S$ such that for $s>S$, $\|\bm\Lambda_{s-1}^2\|<\epsilon$ and $B^{s-1}<\epsilon$. Let $T=2S$. For any $t>t_0 + T$, we have that
\begin{flalign}
&\sum_{j=t_0}^t CB^{t-j} \|\bm\Lambda^2_{j}\|\nn\\& = C\big(\|\bm\Lambda^2_{t}\| + B\|\bm\Lambda^2_{t-1}\| + \cdots + B^{t-t_0}\|\bm\Lambda_{t_0}^2\|\big)\nn\\&\leq C\big(\epsilon + B\epsilon + \cdots + B^{\lfloor\frac{t-t_0}{2}\rfloor}\epsilon + B^{\lfloor\frac{t-t_0}{2}\rfloor+1}\|\bm\Lambda_{t-\lfloor\frac{t-t_0}{2}\rfloor-1}^2\| + \cdots + B^{t-t_0}\|\bm\Lambda_{t_0}^2\| \big)\nn\\& \leq C\Big (\frac{\epsilon}{1-B} + \frac{B^{\lfloor\frac{t-t_0}{2}\rfloor+1}}{1-B}\|\bm\Lambda^2\|\Big) \nn\\&\leq \frac{C(1+\|\bm\Lambda^2\|)}{1-B}\epsilon,
\end{flalign}
where for the second inequality, we compute the sum of the first half and the second half respectively and use the fact that $\|\bm\Lambda_t^2\|\leq \|\bm\Lambda^2\|$ for any $t$. Since $\frac{C(1+\|\bm\Lambda^2\|)}{1-B}$ is bounded, the coefficient of $\|\bm\mu^*_t\|$ will converge to zero as $t\rightarrow\infty$. Therefore, following the same steps as in \eqref{eq:sumofy}, we have that $\lim_{t\rightarrow\infty}\frac{1}{t}\sum_{i=t_0}^{t_0+t-1} \sum_{j=t_0}^i CB^{i-j} \|\bm\Lambda^2_j\|\|\bm\mu^*_i\|^2 =0$. Therefore, we have that 
\begin{flalign}
&\lim_{t\rightarrow\infty}\frac{1}{t}\sum_{i=t_0}^{t_0+t-1} \Big((\bm A\bm\mu_i)^\top\big((\bm A \bm\Sigma^*\bm A^\top + \bm R_{\omega})^{-1}+\bm\Lambda_i^3\big)(\bm A\bm\mu_i) - (\bm A\bm\mu_i^*)^\top(\bm A \bm\Sigma^*\bm A^\top + \bm R_{\omega})^{-1}(\bm A\bm\mu_i^*)\Big) = 0,
\end{flalign}
and thus 
\begin{flalign}
&\lim_{t\rightarrow\infty}\frac{1}{t}\big(\log p_{t_0}(\bm y_{t_0},\cdots, \bm y_{t_0+t-1})- \log p_{t_0}^*(\bm y_{t_0},\cdots, \bm y_{t_0+t-1})\big) = 0.
\end{flalign}

\textbf{Step 3.} Since the observations are independent before the change point, we have that $\lim_{t\rightarrow\infty}\frac{1}{t}\log p_\infty(\bm y_{t_0}, \cdots, \bm y_{t_0+t-1}) = \lim_{t\rightarrow\infty}\frac{1}{t}\sum_{i=1}^{t_0+t-1} \log p_\infty(\bm y_i)$. Since $\log p_\infty(\bm y)$ is a quadratic function of $\bm y$, we have that $\mathbb{E}_\pi[|\log p_\infty(\bm y)|]<\infty$. From the ergodic theorem of Markov chain \cite{roberts2004general,meyn2012markov}, we have that under $\mathbb{P}_{t_0}$,
\begin{flalign}
\lim_{t\rightarrow\infty}\frac{1}{t}\sum_{i={t_0}}^{t_0+t-1} \log p_\infty(\bm y_i) = \mathbb{E}_{\pi}[\log p_\infty(\bm y)]
\end{flalign}
almost surely.

It then follows that under $\mathbb{P}_{t_0}$,
\begin{flalign}\label{eq:llras}
&\lim_{t\rightarrow\infty}\frac{1}{t}\log\frac{p_{t_0}(\bm y_{t_0},\cdots, \bm y_{t_0+t-1})}{p_\infty(\bm y_{t_0},\cdots, \bm y_{t_0+t-1})}\nn\\& = \lim_{t\rightarrow\infty}\frac{1}{t}\Big(\log p_{t_0}(\bm y_{t_0},\cdots, \bm y_{t_0+t-1}) - \sum_{i=t_0}^{t_0+t-1} \log p_\infty(\bm y_i)\Big)\nn\\& = \mathbb{E}_\pi[h(\bm\mu, \bm y)] - \mathbb{E}_\pi[\log p_\infty(\bm y)]
\end{flalign}
almost surely.

Let 
\begin{align}
    \mathcal{K} = \mathbb{E}_\pi[h(\bm\mu, \bm y)] - \mathbb{E}_\pi[\log p_\infty(\bm y)].
\end{align}
We then have that for any initial state $\bm y_{t_0}, \bm\mu_{t_0}$, $\lim_{t\rightarrow\infty}\frac{1}{t}\log\frac{p_{t_0}(\bm y_{t_0},\cdots, \bm y_{t_0+t-1})}{p_\infty(\bm y_{t_0},\cdots, \bm y_{t_0+t-1})} = \mathcal{K}$ under $\mathbb{P}_{t_0}$ almost surely. 

From \eqref{eq:ycondition}, we have that given $\bm \mu^*_t$, the distribution of $\bm y_t$ under $\mathbb{P}_{t_0}$ can be fully specified, denoted by $p_{t_0}^*(\bm y_t|\bm \mu^*_t)$. Therefore, we have that $\pi(\bm\mu, \bm y) = \pi(\bm\mu)p_{t_0}^*(\bm y|\bm\mu)$.
Note that $h(\bm\mu, \bm y) = \log p_{t_0}^*(\bm y|\bm \mu)$. We then have that
\begin{flalign}
    \mathcal{K} &= \mathbb{E}_\pi\Big[
    \log\frac{p_{t_0}^*(\bm y|\bm \mu)}{p_\infty(\bm y)}\Big] \nn\\&= \mathbb{E}_\pi\Big[
    \log\frac{\pi(\bm\mu)p_{t_0}^*(\bm y|\bm \mu)}{\pi(\bm \mu)p_\infty(\bm y)}\Big]\nn\\& = D\big(\pi(\bm\mu)p_{t_0}^*(\bm y|\bm \mu)\|\pi(\bm \mu)p_\infty(\bm y)\big)\nn\\&\geq 0,
\end{flalign}
where $D(\cdot\|\cdot)$ denotes the KL-Divergence between two distributions and the equality holds if and only if $\pi(\bm\mu)p_{t_0}^*(\bm y|\bm \mu) \neq \pi(\bm \mu)p_\infty(\bm y)$. From \eqref{eq:ycondition}, we have that $\pi(\bm\mu)p_{t_0}^*(\bm y|\bm \mu) \neq \pi(\bm \mu)p_\infty(\bm y)$. Therefore, we have that $\mathcal{K}>0$.

Therefore, for any $\eta>0$,
\begin{flalign}
\lim_{t\rightarrow\infty}&\sup_{t_0\geq 1}\text{esssup}\mathbb{P}_{t_0}\Big\{\max_{k\leq t}\sum_{i=t_0}^{t_0+k-1}\log\frac{p_{t_0}(\bm y_i|\bm y_{t_0}, \cdots, \bm y_{i-1})}{p_\infty(\bm y_i)} \mathcal{K}(1+\eta)t\big|\bm y_1, \cdots, \bm y_{t_0-1}\Big\} = 0.
\end{flalign}
Then \eqref{eq:waddup} follows from \cite[Theorem 1]{lai1998information}.
\end{proof}

\section{Proof of Theorem \ref{theorem:waddup}}\label{sec:theorem4}
\begin{proof}
Let $\sigma_0 = 0$ and define the stopping times 
\begin{flalign}
\sigma_{m+1} = \inf\Big\{t>\sigma_m: \sum^t_{i=\sigma_m+1}\log\frac{p_1(\bm y_i|\bm y_1,\cdots, \bm y_{i-1})}{p_\infty(\bm y_i|\bm y_1,\cdots, \bm y_{i-1})}\leq 0\Big\},\ \text{for}\ m\geq 0.
\end{flalign}
We have that 
\begin{flalign}
&\mathbb E_{\infty}\bigg[\prod_{i=\sigma_m+1}^{t+1}\frac{p_1(\bm y_i|\bm y_1,\cdots, \bm y_{i-1})}{p_\infty(\bm y_i|\bm y_1,\cdots, \bm y_{i-1})}\ \text{for some}\ t > \sigma_m \Big|\mathcal{F}_t\bigg] \nn\\&=\prod_{i=\sigma_m+1}^{t}\frac{p_1(\bm y_i|\bm y_1,\cdots, \bm y_{i-1})}{p_\infty(\bm y_i|\bm y_1,\cdots, \bm y_{i-1})}\times\mathbb E_{\infty}\bigg[\frac{p_1(\bm y_{t+1}|\bm y_1,\cdots, \bm y_{t})}{p_\infty(\bm y_{t+1}|\bm y_1,\cdots, \bm y_{t})}\Big|\mathcal{F}_t\bigg]\nn\\&= \prod_{i=\sigma_m+1}^{t}\frac{p_1(\bm y_i|\bm y_1,\cdots, \bm y_{i-1})}{p_\infty(\bm y_i|\bm y_1,\cdots, \bm y_{i-1})}\times\int p_1(\bm y_{t+1}|\bm y_1,\cdots, \bm y_{t})d\bm y_{t+1}\nn\\& = \prod_{i=\sigma_m+1}^{t}\frac{p_1(\bm y_i|\bm y_1,\cdots, \bm y_{i-1})}{p_\infty(\bm y_i|\bm y_1,\cdots, \bm y_{i-1})}.
\end{flalign}
Therefore, $\{\prod_{i=\sigma_m+1}^{t}\frac{p_1(\bm y_i|\bm y_1,\cdots, \bm y_{i-1})}{p_\infty(\bm y_i|\bm y_1,\cdots, \bm y_{i-1})}, \mathcal{F}_t, t>\sigma_m\}$ is a martingale under the pre-change distribution with mean 1. Therefore, from Doob's submartingale inequality\cite{williams_1991} and the optional sampling theorem\cite{williams_1991}, we have that
\begin{flalign}
&\mathbb P_\infty\Big\{\sum_{i=\sigma_m+1}^t \log \frac{p_1(\bm y_i|\bm y_1,\cdots, \bm y_{i-1})}{p_\infty(\bm y_i|\bm y_1,\cdots, \bm y_{i-1})} \geq c\ \text{for some}\ t > \sigma_m\big|\mathcal{F}_{\sigma_m}\Big\}\leq e^{-c}.
\end{flalign}

Let $M = \inf\{m\geq 0: \sigma_m<\infty\ \text{and}\ \sum_{i=\sigma_m+1}^t \log\frac{p_1(\bm y_i|\bm y_1,\cdots, \bm y_{i-1})}{p_\infty(\bm y_i|\bm y_1,\cdots, \bm y_{i-1})}\geq c\ \text{for some}\ t > \sigma_m\}$. We have that 
\begin{flalign}
&\mathbb P_\infty(M\geq m+1|\mathcal{F}_{\sigma_m})\nn\\& =\mathbb P_\infty\Big\{\sum_{i=\sigma_m+1}^t \log \frac{p_1(\bm y_i|\bm y_1,\cdots, \bm y_{i-1})}{p_\infty(\bm y_i|\bm y_1,\cdots, \bm y_{i-1})} < c\ \text{for all}\ t> \sigma_m\big|\mathcal{F}_{\sigma_m}\Big\}\nn\\&\geq 1-e^{-c}.
\end{flalign}
We then have that 
\begin{flalign}
\mathbb P_\infty (M>m) &= \mathbb E_\infty \big[\mathbb P_\infty(M\geq m+1|\mathcal{F}_{\sigma_m})\mathbbm{1}_{\{M\geq m\}}\big]\nn\\&\geq (1-e^{-c})\mathbb P_\infty (M>m-1)\nn\\&\geq (1-e^{-c})^2\mathbb P_\infty (M>m-2)\nn\\&\geq (1-e^{-c})^m\mathbb P_\infty (M>0)\nn\\&\geq (1-e^{-c})^m.
\end{flalign}
It then follows that 
\begin{flalign}
\mathbb E_\infty[\tau_c^*] &\geq \mathbb E_\infty[M]\geq \sum_{m=0}^\infty \mathbb P_\infty (M>m)\nn\\&\geq \sum_{m=0}^\infty(1-e^{-c})^m = e^c.
\end{flalign}
Let $c=\log\gamma$, we have that $\mathbb E_\infty[\tau_c^*] \geq \gamma$.

Define $h_t(\bm\mu, \bm y)$ by replacing $\bm \Sigma^*$ with $\bm \Sigma_t$ in $h(\bm\mu, \bm y)$. We then have that $\log L_t = \sum_{i=1}^t h_i(\bm \mu_i, \bm y_i)-\log p_\infty(\bm y_i)$. From \eqref{eq:llras}, we have that for $0<\delta<1$ and any initial $\bm\mu_t, \bm y_t$,
\begin{flalign}
\lim_{n\rightarrow\infty}\sup_{t\geq t_0\geq 1}\text{esssup}&\mathbb{P}_{t_0}\Big\{n^{-1}\sum_{i=t}^{t+n-1} h_i(\bm y_i,\bm \mu_i)-\log p_\infty(\bm y_i) <\mathcal{K}-\delta|\bm y_1, \cdots, \bm y_{t-1}\Big\} = 0.
\end{flalign}
This implies that 
\begin{flalign}
\sup_{t\geq t_0\geq 1}\text{esssup}&\mathbb{P}_{t_0}\Big\{\sum_{i=t}^{t+n_c-1} h_i(\bm y_i,\bm \mu_i)-\log p_\infty(\bm y_i)<c|\bm y_1, \cdots, \bm y_{t-1}\Big\}\leq \delta
\end{flalign}
for all large $c$, where $n_c$ is the largest integer $\leq (1-\delta)^{-1}\mathcal{K}^{-1}c$.
We then have that for any $t_0\geq 1$ and $l\geq 1$,
\begin{flalign}
&\text{esssup} \mathbb{P}_{t_0}\{\tau_c^* - t_0 >ln_c|\mathcal{F}_{t_0-1} \}\nn\\&\leq \text{esssup} \mathbb{P}_{t_0} \Big\{\sum_{i=t_0+(j-1)n_c}^{t_0+jn_c-1} h_i(\bm y_i,\bm \mu_i)-\log p_\infty(\bm y_i)<c\ \text{for all}\ 1\leq j\leq l |\mathcal{F}_{t_0-1}\Big\}\nn\\&\leq \delta^l.
\end{flalign}
Therefore, 
\begin{flalign}
\sup_{t_0\geq 1}\text{esssup}\mathbb{E}_{t_0}[n_c^{-1}(\tau_c^*-t_0)^+|\mathcal{F}_{t_0-1}]\leq \sum_{l=0}^\infty \delta^l  = (1-\delta)^{-1}.\nn
\end{flalign}
Let $c = \log \gamma$ and $\gamma\rightarrow \infty$. Since $n_c\sim (1-\delta)^{-1}\mathcal{K}^{-1}c$, we have that 
\begin{flalign}
\sup_{t_0\geq 1}\text{esssup}\mathbb{E}_{t_0}[(\tau_c^*-t_0)^+|\bm  y_1,\cdots,\bm y_{t_0-1}]\leq \frac{\log\gamma}{\mathcal{K}}(1+o(1)).\nn
\end{flalign}
This completes the proof.
%Since \eqref{eq:stronglaw} holds, the ARL lower bound and the WADD upper bound of $\tau_c$ follow directly from \cite[Theorem 4]{lai1998information}.
\end{proof}

\section{Proof of Theorem \ref{theorem:ogdarl}}\label{sec:proofarl}
\begin{proof}
Define $p_{\widehat{\theta}_{i}}(\bm y_i|\bm y_1,\cdots, \bm y_{i-1})$ by replacing $\bm A, \bm R_\omega, \bm \mu_i, \bm\Sigma_i$ in $p_1(\bm y_i|\bm y_1,\cdots, \bm y_{i-1})$ with 
$\widehat{\bm A}_i, \widehat{\bm R}_{\omega, i}, \widehat{\bm\mu}_i, \widehat{\bm\Sigma}_i$. Note that the OGA-CuSum can be equivalently written as
\begin{flalign}
    \tau_{\text{OGA}} = \inf\Big\{t: \max_{1\leq k\leq t }\sum^t_{i=k}\log\frac{p_{\widehat{\theta}_{i}}(\bm y_i|\bm y_1,\cdots, \bm y_{i-1})}{p_\infty(\bm y_i|\bm y_1,\cdots, \bm y_{i-1})}\geq c\Big\}.
\end{flalign}
Let $\sigma_0 = 0$ and define the stopping times 
\begin{flalign}
\sigma_{m+1} = \inf\Big\{t>\sigma_m: \sum^t_{i=\sigma_m+1}\log\frac{p_{\widehat{\theta}_{i}}(\bm y_i|\bm y_1,\cdots, \bm y_{i-1})}{p_\infty(\bm y_i|\bm y_1,\cdots, \bm y_{i-1})}\leq 0\Big\}, \ \text{for}\ m\geq 0.
\end{flalign}
We have that 
\begin{flalign}
&\mathbb E_{\infty}\bigg[\prod_{i=\sigma_m+1}^{t+1}\frac{p_{\widehat{\theta}_{i}}(\bm y_i|\bm y_1,\cdots, \bm y_{i-1})}{p_\infty(\bm y_i|\bm y_1,\cdots, \bm y_{i-1})}\ \text{for some}\ t>\sigma_m\Big|\mathcal{F}_t\bigg] \nn\\&=\prod_{i=\sigma_m+1}^{t}\frac{p_{\widehat{\theta}_{i}}(\bm y_i|\bm y_1,\cdots, \bm y_{i-1})}{p_\infty(\bm y_i|\bm y_1,\cdots, \bm y_{i-1})}\times\mathbb E_{\infty}\bigg[\frac{p_{\widehat{\theta}_{t+1}}(\bm y_{t+1}|\bm y_1,\cdots, \bm y_{t})}{p_\infty(\bm y_{t+1}|\bm y_1,\cdots, \bm y_{t})}\Big|\mathcal{F}_t\bigg]\nn\\&= \prod_{i=\sigma_m+1}^{t}\frac{p_{\widehat{\theta}_{i}}(\bm y_i|\bm y_1,\cdots, \bm y_{i-1})}{p_\infty(\bm y_i|\bm y_1,\cdots, \bm y_{i-1})}\times\int p_{\widehat{\theta}_{t+1}}(\bm y_{t+1}|\bm y_1,\cdots, \bm y_{t})d\bm y_{t+1}\nn\\& = \prod_{i=\sigma_m+1}^{t}\frac{p_{\widehat{\theta}_{i}}(\bm y_i|\bm y_1,\cdots, \bm y_{i-1})}{p_\infty(\bm y_i|\bm y_1,\cdots, \bm y_{i-1})}.
\end{flalign}
Therefore, $\{\prod_{i=\sigma_m+1}^{t}\frac{p_{\widehat{\theta}_{i}}(\bm y_i|\bm y_1,\cdots, \bm y_{i-1})}{p_\infty(\bm y_i|\bm y_1,\cdots, \bm y_{i-1})}, \mathcal{F}_t, t>\sigma_m\}$ is a martingale under the pre-change distribution with mean 1. Therefore, from Doob's submartingale inequality\cite{williams_1991} and the optional sampling theorem\cite{williams_1991}, we have that
\begin{flalign}
&\mathbb P_\infty\Big\{\sum_{i=\sigma_m+1}^t \log \frac{p_{\widehat{\theta}_{i}}(\bm y_i|\bm y_1,\cdots, \bm y_{i-1})}{p_\infty(\bm y_i|\bm y_1,\cdots, \bm y_{i-1})} \geq c\ \text{for some}\ t > \sigma_m\big|\mathcal{F}_{\sigma_m}\Big\}\leq e^{-c}.
\end{flalign}
Theorem \ref{theorem:ogdarl} can then be proved following the same techniques as in the proof of Theorem \ref{theorem:waddup}.
\end{proof}

\newpage
\bibliographystyle{ieeetr}
\bibliography{QCD}
\end{document}